\def\comments{0}
\def\anonymous{0}
\def\arxivymous{1}

\documentclass[11pt, letter]{article}   	%
\usepackage{preamble}

\title{Private Identity Testing for High-Dimensional Distributions}
\ifnum\anonymous=1
  \author{Anonymous Author(s)}
\else
  \author[1]{Cl\'ement L. Canonne}
  \author[2]{Gautam Kamath}
  \author[3,4]{Audra McMillan}
  \author[4]{\mbox{Jonathan Ullman}}
  \author[4]{Lydia Zakynthinou}
  \affil[1]{School of Computer Science, University of Sydney}
  \affil[2]{Cheriton School of Computer Science, University of Waterloo}
  \affil[3]{Department of Computer Science, Boston University}
  \affil[4]{Khoury College of Computer Sciences, Northeastern University}
  
\fi

\begin{document}
\maketitle

\begin{abstract}
In this work we present novel differentially private identity (goodness-of-fit) testers for natural and widely studied classes of multivariate product distributions: product distributions over $\pmo^{d}$ and Gaussians in $\R^d$ with known covariance.  Our testers have improved sample complexity compared to those derived from previous techniques, and are the first testers whose sample complexity matches the order-optimal minimax sample complexity of $O(d^{1/2}/\alpha^2)$ in many parameter regimes. We construct two types of testers, exhibiting tradeoffs between sample complexity and computational complexity. Finally, we provide a two-way reduction between testing a subclass of multivariate product distributions and testing univariate distributions, and thereby obtain upper and lower bounds for testing this subclass of product distributions.
\end{abstract}

\section{Introduction}

A foundation of statistical inference is \emph{hypothesis testing}: given two disjoint sets of probability distributions $\cH_0$ and $\cH_1$, we want to design an algorithm $T$ that takes a random sample $X$ from some distribution $P \in \cH_0 \cup \cH_1$ and, with high probability, determines whether $P$ is in $\cH_0$ or $\cH_1$.  Hypothesis tests formalize yes-or-no questions about an underlying population given a random sample from that population, and are ubiquitous in the physical, life, and social sciences, where hypothesis tests with high confidence are the gold standard for publication in top journals. 

In many of these applications---clinical trials, social network analysis, or demographic studies, to name a few---this sample contains sensitive data belonging to individuals, in which case it is crucial for the hypothesis test to respect these individuals' \emph{privacy}.  It is particularly desirable to guarantee \emph{differential privacy}~\cite{DworkMNS06}, which has become the \emph{de facto} standard for the analysis of private data.  Differential privacy is used as a measure of privacy for data analysis systems at Google~\cite{ErlingssonPK14},  Apple~\cite{AppleDP17}, and the U.S.~Census Bureau~\cite{DajaniLSKRMGDGKKLSSVA17}. Differential privacy and related notions of \emph{algorithmic stability} are also crucial for statistical validity even when confidentiality is not a direct concern, as they provide generalization guarantees in an adaptive setting~\cite{DworkFHPRR15, BassilyNSSSU16, RogersRST16}.

While differentially private hypothesis testing has been extensively studied (see Section~\ref{sec:relatedwork}), almost all of this work has focused on \emph{low-dimensional distributions}.  The main contribution of this work is to give novel algorithms for hypothesis testing problems on \emph{high-dimensional distributions} with improved sample complexity.  In particular, we give differentially private algorithms for the following natural and fundamental problems:
\begin{enumerate}
\item Given samples from a multivariate Gaussian $P$ in $\R^d$ whose covariance is known to be the identity, decide if $P$ is $\cN(0,\idcov)$ or is $\alpha$-far from $\cN(0,\idcov)$ in total variation distance.  Or, equivalently, decide if $\ex{}{P} = 0$ or if $\| \ex{}{P} \|_2 \geq \alpha$.

\item Given samples from a product distribution $P$ over $\pmo^{d}$, decide if $P$ is the uniform distribution or is $\alpha$-far from the uniform distribution in total variation distance.  Or, equivalently, decide if $\ex{}{P} = 0$ or if $\| \ex{}{P} \|_2 \geq \alpha$.

\item Given samples from a product distribution $P$ over $\zo^{d}$, decide if $P$ is equal to some given extremely biased distribution $Q$ with mean $\ex{}{Q} \preceq O(\frac{1}{d})$ or is $\alpha$-far from $Q$ in total variation distance.  In this case our tester achieves the provably optimal sample complexity.
\end{enumerate}

The main challenge in solving these high-dimensional testing problems privately is that the only known non-private test statistics for these problems have high worst-case sensitivity.  That is, these test statistics can potentially be highly brittle to changing even a single one of the samples. We overcome this challenge by identifying two methods for reducing the sensitivity of the test statistic without substantially changing its average-case behavior on typical datasets sampled from the distributions we consider.  The first is based on a novel \emph{private filtering} method, which gives a computationally efficient tester.  The second combines the method of \emph{Lipschitz extensions}~\cite{BlockiBDS13,KasiviswanathanNRS13} with \emph{recursive preconditioning}, which yields an exponential-time tester, but with improved sample complexity.

\subsection{Background: Private Hypothesis Testing}

We start by giving some background on private hypothesis testing.  First, when we say that we want a differentially private hypothesis tester for a pair $\cH_0,\cH_1$ over domain $\cX$, we mean that we seek an algorithm $A \from \cX^* \to \zo$ such that
\begin{enumerate}
\item $A$ is $\eps$-differentially private in the \emph{worst case}.  That is, for \emph{every} pair of samples $X,X' \in \cX^n$ differing on one sample, $A(X)$ and $A(X')$ are \emph{$\eps$-close} in a precise sense (see Definition~\ref{def:dp}).
\item $A$ correctly distinguishes $\cH_0$ from $\cH_1$ \emph{on average}.  If $X = (X^{(1)},\dots,X^{(n)})$ is drawn i.i.d.\ from some $P \in \cH_0$ then $A(X)$ outputs $0$ with high probability, and similarly if $P \in \cH_1$.  We call the minimum number of samples $n$ such that $A$ distinguishes $\cH_0$ and $\cH_1$ the \emph{sample complexity} of $A$.
\end{enumerate}
Note that we require testers that are private in the worst case, even though we only require accuracy on average.  It is indeed important for privacy to be a worst-case notion, rather than contingent on the assumption that the data is sampled i.i.d.\ from some $P \in \cH_0 \cup \cH_1$; this is because we have no way of verifying this assumption (which may anyway be an oversimplified modeling assumption), yet once privacy is lost it cannot be recovered.  Worst-case privacy notions also enjoy strong composition and generalization properties not shared by average-case privacy notions.

There exists a black-box method for obtaining a differentially private tester from any non-private tester $A$ using the \emph{sample-and-aggregate framework}~\cite{NissimRS07}.  Specifically, given any tester $A$ with sample complexity $n$, we can obtain an $\eps$-differentially private tester with sample complexity $O(n/\eps)$.
When $\eps$ is a constant, the reduction is within a constant factor of the optimal sample complexity; however, this overhead factor in the sample complexity blows up as $\eps \to 0$.

One can often obtain stronger results using a white-box approach.  For example, suppose $P$ is a Bernoulli random variable and we aim to test if $P = \mathrm{Ber}(1/2)$ or $P = \mathrm{Ber}(1/2 + \alpha)$.  Non-privately, $\Theta(1/\alpha^2)$ samples are necessary and sufficient.  Thus the black-box approach gives a sample complexity of $\Theta(1/\alpha^2 \eps)$. 
However, if we work directly with the test statistic $T(X) = \frac{1}{n} \sum_{j=1}^{n} X^{(j)}$, we can obtain privacy by computing $T(X) + Z$ where $Z$ is drawn from an appropriate distribution with standard deviation $O(1/\eps n)$.  One can now show that this private test succeeds using
\begin{equation}
n = \underbrace{O\left( \frac{1}{\alpha^2}\right)}_{\textit{non-private sc}}~~+ \underbrace{O\left(\frac{1}{\alpha \eps}\right)}_{\textit{overhead for privacy}}
\end{equation}
samples, which actually \emph{matches} the non-private sample complexity up to a factor of $1+o(1)$ unless $\eps$ is very small.  Our main contribution is to achieve qualitatively similar results for the high-dimensional testing problem we have described above.

\subsection{Our Results}

\subsubsection{Product Distributions}
\begin{theorem}[Informal] \label{thm:product-fast-intro}
There is a linear-time, $\eps$-differentially private tester $A$ that distinguishes the uniform distribution over $\pmo^d$ from any product distribution over $\pmo^{d}$ that is $\alpha$-far in total variation distance using $n = n(d,\alpha,\eps)$ samples for 
\begin{equation} \label{eq:product-fast-intro}
n = \underbrace{O\left( \frac{d^{1/2}}{\alpha^2}\right)}_{\textit{non-private sc}}~~+ \underbrace{\tilde{O}\left(\frac{d^{1/2}}{\alpha \eps}\right)}_{\textit{overhead for privacy}}
\end{equation}
\end{theorem}

The sample complexity in Theorem~\ref{thm:product-fast-intro} has an appealing form.  One might even conjecture that this sample complexity is optimal by analogy with the case of privately \emph{estimating} a product distribution over $\pmo^{d}$ or a Gaussian in $\R^{d}$ with known covariance, for which the sample complexity is  $\tilde{\Theta}(d/\alpha^2 + d/\alpha \eps)$ in both cases~\cite{KamathLSU19}.  However, the next result shows that there is in fact an exponential-time private tester that has even lower sample complexity in some range of parameters.

\bgroup\color{red}
\begin{framed}
It was brought to our attention by Shyam Narayanan that the proof of Theorem~\ref{thm:product-slow-intro} is incorrect: specifically, the argument establishing Lipschitzness of our statistic, Lemma~\ref{sensitivityT}, does not extend from neighboring databases to arbitrary ones. We discuss this flaw in more detail in Section~\ref{sec:uniformity:inefficient}; however, we chose to keep the statement in this revision, along with the (flawed) argument, as this forms the basis and context for Shyam Narayanan's follow-up work~\cite{Narayanan22}, which fixes this gap in our proof and improves upon our stated bound.
\end{framed}
\egroup

\begin{theorem}[Informal] \label{thm:product-slow-intro}
There is an exponential-time, $\eps$-differentially private tester $A$ that distinguishes the uniform distribution over $\pmo^d$ from any product distribution over $\pmo^{d}$ that is $\alpha$-far in total variation distance using $n = n(d,\alpha,\eps)$ samples for 
\begin{equation} \label{eq:product-slow-intro}
n = \underbrace{O\left( \frac{d^{1/2}}{\alpha^2}\right)}_{\textit{non-private sc}}~~+~~\underbrace{\tilde{O}\left(\frac{d^{1/2}}{\alpha \eps^{1/2}} 
+ \frac{d^{1/3}}{\alpha^{4/3}\eps^{2/3}}  
+ \frac{1}{\alpha \eps}\right)}_{\textit{overhead for privacy}}
\end{equation}
\end{theorem}
First, we remark that the sample complexity in \eqref{eq:product-slow-intro} is precisely, up to logarithmic factors, the optimal sample complexity for testing uniformity of discrete distributions on the domain $\{1,\dots,d\}$~\cite{AcharyaSZ18}, which hints that it may be optimal (especially in view of Theorem~\ref{thm:extreme-intro} below).  

The expression in~\eqref{eq:product-slow-intro} is rather complex and somewhat difficult to interpret and compare to~\eqref{eq:product-fast-intro}.  One way to simplify the comparison is to consider the range of the privacy parameter $\eps$ where privacy comes \emph{for free}, meaning the sample complexity is dominated by the non-private term $\Theta(d^{1/2}/\alpha^2)$.  For the efficient algorithm, privacy comes for free roughly when $\eps = \Omega( \alpha )$.  For the computationally inefficient algorithm, however, one can show that privacy comes for free roughly when $\eps = \Omega(\alpha^2 + \alpha / d^{1/4})$, which is better if both $1/\alpha$ and $d$ are superconstant.

Using a simple reduction (Corollary~\ref{cor:balancedreduction}), Theorems \ref{thm:product-fast-intro} and \ref{thm:product-slow-intro} extend, with a constant-factor loss in sample complexity, to identity testing for  \emph{balanced} product distributions. That is, suppose $Q$ is a product distribution such that every coordinate of $\ex{}{Q}$ is bounded away from $-1$ and $+1$. Then with a constant-factor loss in sample complexity, we can distinguish whether (i)~a product distribution $P$ over $\pmo^{d}$ is either equal to $Q$, or (ii)~whether $P$ is far from $Q$ in total variation distance.

\subsubsection{Extreme Product Distributions}
We then focus in Section~\ref{sec:extreme} on a specific class of Boolean product distributions, which we refer to as \emph{extreme}.  Loosely speaking, a product distribution is extreme if each of its marginals is $O(1/d)$-close to constant. For this restricted class of Boolean product distributions, we provide a two-way reduction argument showing that identity testing is equivalent to the identity testing in the \emph{univariate} setting. This allows us to transfer known lower bounds on private univariate identity testing to our extreme product distribution class, which gives us the first non-trivial lower bounds for privately testing identity of product distributions.

\begin{theorem}[Informal] \label{thm:extreme-intro}
	The sample complexity of privately testing identity of univariate distributions over $[d]$ and the sample complexity of privately testing identity of \emph{extreme} product distributions over $\pmo^d$ are equal, up to constant factors.
\end{theorem}

\subsubsection{Gaussians}
Finally, we can obtain analogous results for identity testing of multivariate Gaussians with known covariance by reduction to uniformity testing of Boolean product distributions (Theorem~\ref{theo:gaussian:product:reduction}).
\begin{theorem}[Informal] \label{thm:gaussian-fast-intro}
There is a linear-time, $\eps$-differentially private tester $A$ that distinguishes the standard normal $\cN(0, \idcov)$ from any normal $\cN(\mu,\idcov)$ that is $\alpha$-far in total variation distance (or, equivalently, $\| \mu \|_2 \geq \alpha$) using $n = n(d,\alpha,\eps)$ samples for $
n = O( \frac{d^{1/2}}{\alpha^2}) + \tilde{O}(\frac{d^{1/2}}{\alpha \eps}).
$
\end{theorem}

\begin{theorem}[Informal] \label{thm:gaussian-slow-intro}
There is an exponential-time $\eps$-differentially private tester $A$ that distinguishes the standard normal $\cN(0, \idcov)$ from any normal $\cN(\mu,\idcov)$ that is $\alpha$-far in total variation distance using $n = n(d,\alpha,\eps)$ samples for 
$
n =O( \frac{d^{1/2}}{\alpha^2}) + \tilde{O}(\frac{d^{1/2}}{\alpha \eps^{1/2}} + \frac{d^{1/3}}{\alpha^{4/3}\eps^{2/3}} + \frac{1}{\alpha \eps}).
$
\end{theorem}
\ifnum\arxivymous=1
We note that we can also obtain these results directly, and circumvent the constant-factor loss incurred by the above reduction, by extending the techniques for the Boolean case and directly constructing the tester. We demonstrate this for the efficient test in Appendix~\ref{app:gauss}.
\anote{I'm assuming our plan is to only do this for the efficient test since its super annoying to do for the inefficient?}
\fi

\subsubsection{Useful Tools for Non-Private Hypothesis Testing}
We highlight some tools in this paper which are useful for hypothesis testing, even without privacy constraints.
\begin{itemize}
  \item A reduction from testing the mean of a Gaussian, to testing uniformity of a product distribution (Theorem~\ref{theo:gaussian:product:reduction}).
  \item A reduction from testing identity to a ``balanced'' product distribution, to testing uniformity of a product distribution (Corollary~\ref{cor:balancedreduction}).
  \item An equivalence between testing identity over a domain of size $d$, and testing identity to an ``extreme'' product distribution in $d$ dimensions (Theorem~\ref{thm:extreme-reduction}).
\end{itemize}

\subsection{Techniques}

To avoid certain technicalities involving continuous and unbounded data, we will describe our tester for product distributions over $\pmo^{d}$, rather than for Gaussians over $\R^d$.  The Gaussian case can be handled using a reduction to the Boolean case, or can be handled directly using a nearly identical approach.

\mypar{First Attempts.} A natural starting point is to study the asymptotically-optimal non-private tester of Canonne \etal~\cite{CanonneDKS17}. Let $P$ be a distribution over $\pmo^{d}$ and $X = (X^{(1)},\dots,X^{(n)}) \in \pmo^{n\times d}$ be $n$ i.i.d.\ samples from $P$.  Define the test statistic
\[
T(X) = \| \bar{X} \|_2^2-nd
~~\textrm{where}~~\bar{X} = \sum_{j=1}^{n} X^{(j)}
\]
The analysis of Canonne \etal~shows that if $P$ is the uniform distribution then $\ex{}{T(X)} = 0$, while if $P$ is a product distribution that is $\alpha$-far from uniform then $\ex{}{T(X)} = \Omega(\alpha^2 n^2)$.  Moreover, the variance of $T(X)$ can be bounded so that, for some $n = O(d^{1/2}/\alpha^2)$ we can distinguish between these two cases.  In order to obtain a private tester, we need to somehow add noise to $T(X)$ whose magnitude is much smaller than the $\alpha^2 n^2$ gap between the two cases.

The standard approach to obtaining a private tester is to add noise to the statistic $T(X)$ calibrated to its \emph{global sensitivity}, which is defined as
\[
\mathrm{GS}_{T} = \max_{X \sim X'} ( T(X) - T(X') )
\]
where $X \sim X'$ denotes that $X$ and $X'$ are \emph{neighboring samples} that differ on at most one sample.  To ensure privacy it is then sufficient to compute a noisy statistic $T(X) + Z$ where $Z$ is chosen from an appropriate distribution (commonly, a Laplace distribution) with mean $0$ and standard deviation $O(\mathrm{GS}_{T}/\eps)$.  

One can easily see that the global sensitivity of $T(X)$ is $O(nd)$, so it suffices to add noise $O(nd/\eps)$.  The tester will still be successful provided $n = \Omega(d/(\alpha^2 \eps))$, resulting in an undesirable linear dependence in $d$.  In particular, this is dominated in all parameter regimes by the sample-and-aggregate approach.

In view of the above, one way to find a better private tester would be to identify an alternative statistic for testing uniformity of product distributions with lower global sensitivity.  
However, we do not know of any other such test statistic that has asymptotically optimal non-private sample complexity; a prerequisite to achieving optimal private sample complexity.

\mypar{Intuition: High-Sensitivity is Atypical.}
Since we need privacy to hold in the \emph{worst case} for every dataset, any private algorithm for computing $T$ must have error proportional to $\mathrm{GS}_{T}$ on \emph{some} dataset $X$.  However, the utility of the tester only applies \emph{on average} to typical datasets drawn from product distributions.  Thus, we can try to find some alternative test statistic $\hat{T}$ that is close to $T$ on these typical datasets, but has lower global sensitivity.  

To this end, it will be instructive to look at the sensitivity of $T$ at a particular pair of datasets $X \sim X'$ that differ between some sample $X^{(j)}$ and $X'^{(j)}$.  
\begin{equation} \label{eq:sensitivity-intro}
T(X) - T(X') = 2(\langle X^{(j)}, \bar{X} \rangle - \langle X'^{(j)}, \bar{X}' \rangle)
\end{equation}
Notice that $T(X)$ and $T(X')$ can only differ by a large amount when one of the two datasets contains some point $X^{(j)}$ such that $\langle X^{(j)}, \bar{X} \rangle$ is large.  Thus, if we could somehow restrict attention to making $T(X)$ private only for datasets in the set
\[
\cC(\Delta) = \set{ X \in \pmo^{n \times d} : \forall j = 1,\dots,n~~|\langle X^{(j)} , \bar{X} \rangle| \leq \Delta}
\]
the sensitivity would be at most $4 \Delta$.  As we will show, for typical datasets drawn from product distributions (except in some corner cases), the data will lie in $\cC(\Delta)$ for $\Delta \ll nd$.  For example, if we draw $X^{(1)},\dots,X^{(n)}$ uniformly at random in $\pmo^{d}$, then we have
\[
\ex{}{\max_{j \in [n]} | \langle X^{(j)}, \bar{X} \rangle| } \lesssim d + (nd)^{1/2}
\]
More generally, if $X^{(1)},\dots,X^{(n)}$ are drawn from any product distribution, then one can show that with high probability we have
\begin{equation} \label{eq:sens}
\max_{j \in [n]} | \langle X^{(j)}, \bar{X} \rangle | \lesssim \frac{1}{n} \| \bar{X} \|_2^2 + \| \bar{X} \|_2.
\end{equation}
Note that, except in pathological cases where the product distribution is very skewed, the right-hand side will be much smaller than $nd$.

This discussion demonstrates that for typical datasets drawn from product distributions, the sensitivity of the test statistic $T(X)$ should be much lower than its global sensitivity.  We can exploit this observation in two different ways to obtain improved testers.

\mypar{A Sample-Efficient Tester via Lipschitz Extensions}
Suppose we fix some value of $\Delta$, and we are only concerned with estimating $T(X)$ accurately on nice datasets that lie in $\cC(\Delta)$, the set of datasets where $| \langle X^{(j)}, \bar{X} \rangle | \leq \Delta$.  Even in this case, it would not suffice to add noise to $T(X)$ proportional to $\Delta/\eps$, since we need privacy to hold for \emph{all} datasets.  

However, we can compute $T(X)$ accurately on these nice datasets while guaranteeing privacy for all datasets using the beautiful machinery of \emph{privacy via Lipschitz extensions} introduced by Blocki \etal~\cite{BlockiBDS13} and Kasiviswanathan \etal~\cite{KasiviswanathanNRS13} in the context of node-differential-privacy for graph data.  Specifically, for a function $T$ defined on domain $\mathcal{X}$, and a set $\cC \subset \mathcal{X}$, a \emph{Lipschitz extension of $T$ from $\cC$} is a function $\hat{T}$ defined on all datasets such that: 
\begin{enumerate}
\item The extension $\hat{T}$ agrees with $T$ on $\cC$, namely $\hat{T}(X) = T(X)$ for every $X \in \cC$, and
\item The global sensitivity of $\hat{T}$ on all of $\mathcal{X}$ is at most the sensitivity of $T$ restricted to $\cC$, namely $\mathrm{GS}_{\hat{T}} \leq \max_{X,X' \in \cC} T(X) - T(X')$.  
\end{enumerate}
Perhaps surprisingly, such a Lipschitz extension exists for \emph{every} real-valued function $T$ and every subset of the domain $\cC$~\cite{McShane34}!  Once we have the Lipschitz extension, we can achieve privacy for all datasets by adding noise to $\hat{T}$ proportional to the sensitivity of $T$ on $\cC$.  In general, the Lipschitz extension can be computed in time polynomial in $|\mathcal{X}|$, which, in our case, is $2^{nd}$.

Thus, if we knew a small value of $\Delta$ such that $X \in \cC(\Delta)$ then we could compute a Lipschitz extension of $T$ from $\cC(\Delta)$ and we would be done.  From \eqref{eq:sens}, we see that a good choice of $\Delta$ is approximately $\frac{1}{n} \| \bar{X} \|_2^2 + \| \bar{X} \|_2$.  Not only do we not know this quantity, but $\| \bar{X} \|_2^2$ is precisely the test statistic we wanted to compute!  Thus, we seem to be stuck.

We untie this knot using a recursive approach. we start with some weak upper bound $\Delta^{(m)}$ such that we know $X \in \cC(\Delta^{(m)})$.  For example, we can set $\Delta^{(1)} = 4nd$ as a base case, which is the worst-case upper bound on the sensitivity.  Using this upper bound, and the Lipschitz extension onto the set $\cC(\Delta^{(m)})$, we can then get a weak private estimate of the test statistic $T(X) = \| \bar{X} \|_2^2 - nd$ with error $\lesssim \Delta^{(m)}/\eps$.  At this point, we may already have enough information to conclude that $X$ was not sampled from the uniform distribution and reject.  Otherwise, we can certify that
\[
\| \bar{X} \|_2^2 \lesssim nd + \frac{\Delta^{(m)}}{\eps}
\]
If $X$ was indeed sampled from a product distribution, then \eqref{eq:sens} tells us that $X \in \cC(\Delta^{(m+1)})$ for
\[
\Delta^{(m+1)} \lesssim \frac{\| \bar{X} \|_2^2}{n} + \| \bar{X} \|_2 \lesssim d + \frac{\Delta^{(m)}}{\eps n} + \sqrt{nd} + \sqrt{\frac{\Delta^{(m)}}{\eps}}
\]
Then, as long as $\Delta^{(m+1)}$ is significantly smaller than $\Delta^{(m)}$ we can recurse, and get a better private estimate of $\| \bar{X} \|_2^2$.  Once the recurrence converges and $\Delta^{(m+1)}$ is no longer getting smaller, we can stop and make a final decision whether to accept or reject.  One can analyze this recurrence and show that it will converge rapidly to some $\Delta^* = O( d + \sqrt{nd} + 1/\eps)$.  Thus the final test will distinguish uniform from far-from-uniform provided $\Delta^*/\eps \ll \alpha^2 n^2$, which occurs at the sample complexity we claim.

This recursive approach is loosely similar in spirit to methods in~\cite{KamathLSU19,FeldmanKT20}, whereby we obtain more and more accurate private estimates, necessitating less noise addition at each step.
However, our technique is very different from theirs, and the first to interact with the Lipschitz extension framework.
We believe our work adds more evidence for the broad applicability of this perspective.

\mypar{A Computationally Efficient Tester via Private Filtering.} The natural way to make the test above computationally efficient is to try to explicitly construct the Lipschitz extension $\hat{T}$ onto $\cC(\Delta)$ using the following \emph{filtering} approach.  Although we do not know how to do so, we can start with the following simple candidate:
\begin{quote}
	$\hat{T}(X):$ Throw out every $X^{(j)}$ such that $\langle X^{(j)}, \bar{X} \rangle$ is larger than some $\Delta$ and let the resulting dataset be $Y$.  Output $T(Y)$.
\end{quote} 

Obviously this will agree with $T(X)$ whenever $X \in \cC(\Delta)$, however it is unclear how to argue that it has sensitivity at most $O(\Delta)$.  The reason is that the decision to throw out $X^{(j)}$ or not depends on the global quantity $\bar{X} = \sum_{j} X^{(j)}$.  Thus, potentially, we could have one dataset $X$ where no points are thrown out so $Y = X$, and a neighboring dataset $X'$ where, because $\bar{X}'$ is slightly different from $\bar{X}$, all points are thrown out and $Y' = \emptyset$.  In this case, the difference between the test statistic on $Y,Y'$ could be much larger than $\Delta$.

We solve this problem by modifying the algorithm to throw out points based on
$\langle X^{(j)}, \tilde{X} \rangle$ for some \emph{private} quantity $\tilde{X} \approx \bar{X}$.  Although the proof is somewhat subtle, we can use the fact that $\tilde{X}$ is private to argue that, under some appropriate coupling of the two executions, $Y$ and $Y'$ differ on at most one sample $Y^{(j)},Y'^{(j)}$ and $\langle Y^{(j)}, \bar{Y} \rangle$ and $\langle Y'^{(j)}, \bar{Y}' \rangle$ are at most $O(\Delta)$. In order to compensate for the fact that $\tilde{X}$ is only an approximation of $\bar{X}$, we cannot take $\Delta$ too small, which ultimately is why this tester has larger sample complexity than the non-computationally efficient one.

\subsection{Related Work} \label{sec:relatedwork}
Over the last couple decades, there has been significant work on hypothesis testing with a focus on minimax rates.
The starting point in the statistics community could be considered the work of Ingster and coauthors~\cite{Ingster94,Ingster97,IngsterS03}.
Within theoretical computer science, study on hypothesis testing arose as a subfield of property testing~\cite{GoldreichGR96, GoldreichR00}.
Work by Batu \etal~\cite{BatuFRSW00,BatuFFKRW01} formalized several of the commonly studied problems, including testing of uniformity, identity, closeness, and independence.
Other representative works in this line include~\cite{BatuKR04,Paninski08,Valiant11,ChanDVV14, ValiantV14,AcharyaDK15,BhattacharyaV15,DiakonikolasKN15a,CanonneDGR16,DiakonikolasK16,Goldreich16,BlaisCG17,DaskalakisKW18}.
Some works on testing in the multivariate setting include testing of independence~\cite{BatuFFKRW01,AlonAKMRX07,RubinfeldX10, LeviRR13, AcharyaDK15,DiakonikolasK16,CanonneDKS18}, and testing on graphical models~\cite{CanonneDKS17, DaskalakisP17,DaskalakisDK18,GheissariLP18,AcharyaBDK18,BezakovaBCSV19}.
We note that graphical models (both Ising models and Bayesian Networks) include the product distribution case we study in this paper.
Surveys and more thorough coverage of related work on minimax hypothesis testing include~\cite{Rubinfeld12,Canonne15a,Goldreich17,BalakrishnanW18,Kamath18}.

Early work on differentially private hypothesis testing began in the statistics community with~\cite{VuS09,UhlerSF13}, and there has more recently been a large body of work on this subject primarily in the computer science community.
Most relevant to this paper is the line of work on minimax sample complexity of private hypothesis testing, initiated by Cai \etal~\cite{CaiDK17}.
\cite{AcharyaSZ18} and~\cite{AliakbarpourDR18} have given worst-case nearly optimal algorithms for goodness-of-fit and closeness testing of arbitrary discrete distributions.
Other recent works in this vein focus on testing of simple hypotheses~\cite{CummingsKMTZ18, CanonneKMSU19}.
A paper of Awan and Slavkovic~\cite{AwanS18} gives a universally optimal test when the domain size is two, however Brenner and Nissim~\cite{BrennerN14} shows that such universally optimal tests cannot exist when the domain has more than two elements.  
A line of work complementary to the minimax results~\cite{WangLK15,GaboardiLRV16,KiferR17,KakizakiSF17,CampbellBRG18,SwanbergGGRGB19,CouchKSBG19} designs differentially private versions of popular test statistics for testing goodness-of-fit, closeness, and independence, as well as private ANOVA, focusing on the performance at small sample sizes.  
Work by Wang \etal~\cite{WangKLK18} focuses on generating statistical approximating distributions for differentially private statistics, which they apply to hypothesis testing problems.
There has also been work on hypothesis testing in the \emph{local model} of differential privacy, with a focus on both the minimax setting~\cite{DuchiJW13, Sheffet18, AcharyaCFT19} and the asymptotic setting~\cite{GaboardiR18}.

Looking more broadly at private algorithms in statistical settings, there has recently been significant study into estimation tasks.
Some settings of interest include univariate discrete distributions~\cite{BunNSV15,DiakonikolasHS15}, Gaussians~\cite{KarwaV18,KamathLSU19}, and multivariate settings~\cite{KamathLSU19,CaiWZ19}.
Upper and lower bounds for learning the mean of a product distribution over the hypercube in $\ell_\infty$-distance include~\cite{BlumDMN05, BunUV14, DworkMNS06, SteinkeU17a}.
\cite{AcharyaKSZ18} designed nearly optimal algorithms for estimating properties like support size and entropy. 
\cite{Smith11} gives an algorithm which allows one to estimate asymptotically normal statistics with minimal increase in the sample complexity.
In the local model, some works on distribution estimation include~\cite{DuchiJW13,WangHWNXYLQ16,KairouzBR16,AcharyaSZ19, DuchiR18,YeB18,GaboardiRS19,JosephKMW19}.

The Lipschitz-extension technique that we build upon was introduced in~\cite{BlockiBDS13,KasiviswanathanNRS13}, who also gave efficient Lipschitz extensions for graph statistics such as the edge density and the number of triangles in sparse graphs.  Later work constructed efficient Lipschitz extensions for richer classes of graph statistics~\cite{RaskhodnikovaS16,SealfonU19}.  
A related work of~\cite{CummingsD20} introduced a variant of the Lipschitz-extension machinery and used it to give efficient algorithms for statistics such as the median and trimmed mean.
Recent results~\cite{BorgsCSZ18a,BorgsCSZ18b} prove the existence of Lipschitz extensions for all differentially private algorithms, though efficiency is not a focus.

\section{Preliminaries}

\subsection{Differential Privacy}\label{sec:DPbackground}

Informally, differential privacy is a property that a randomized algorithm satisfies if its output distribution does not change significantly under the change of a single data point. 

Let $S, S'\in \mathcal{S}^n$ be two datasets of the same size. We say that $S,S'$ are \textit{neighbors}, denoted as $S\sim S'$, if they differ in at most one data point.
\begin{definition}[Differential Privacy,~\cite{DworkMNS06}] \label{def:dp}
A randomized algorithm $\mathcal{A}\colon\mathcal{S}^n \to \mathcal{O}$ is  {\em $(\eps,\delta)$-differentially private} (DP) if for all neighboring datasets $S, S'$ and all measurable $O\subseteq \mathcal{O}$,
\[\Pr[\mathcal{A}(S) \in O] \leq e^{\eps}\Pr[\mathcal{A}(S')\in O] + \delta.\]
Algorithm $\mathcal{A}$ is {\em $\eps$-differentially private} if it satisfies the definition for $\delta=0$.
\end{definition}

Differential privacy satisfies the following two useful properties.
\begin{lemma}[Post-Processing,~\cite{DworkMNS06}]\label{lem:post-processing}
Let $\mathcal{A}\colon\mathcal{S}^n \to \mathcal{O}$ be a randomized algorithm that is $(\eps,\delta)$-differentially private.  For every (possibly randomized) $f : \mathcal{O} \to \mathcal{O}'$, $f \circ \mathcal{A}$ is $(\eps,\delta)$-differentially private.
\end{lemma}

\begin{lemma}[Composition,~\cite{DworkMNS06}]\label{lem:composition}
Let $\mathcal{M}_i\colon\mathcal{S}^n \to \mathcal{O}_i$ be an $(\eps_i, \delta_i)$-differentially private algorithm for $i\in[k]$. If $\mathcal{M}_{[k]}\colon \mathcal{S}^n \to \mathcal{O}_1\times\ldots\times\mathcal{O}_k$ is defined as $\mathcal{M}_{[k]}(S)=\left(\mathcal{M}_1(S),\ldots,\mathcal{M}_k(S)\right)$, then $\mathcal{M}_{[k]}$ is $(\sum_{i=1}^k \eps_i, \sum_{i=1}^k\delta_i)$-differentially private.
\end{lemma}

A common technique that differentially private mechanisms use is to add zero-mean noise of appropriate scale to the quantities computed from the dataset. The scale of the noise depends on the  \textit{sensitivity} of the function of the dataset we aim to compute. Intuitively, the sensitivity represents the maximum change that the change of a single data point can incur on the output of the function.

\begin{definition}[$\ell_1$- and $\ell_2$-sensitivity]
Let $f\colon\mathcal{S}^n \to \R^k$. The {\em $\ell_1$-sensitivity} of $f$ is \[\Delta_1f=\max\limits_{\substack{S,S'\in\mathcal{S}^n\\ S\sim S'}} \|f(S)-f(S')\|_1.\]
Respectively, the {\em $\ell_2$-sensitivity} of $f$ is $\Delta_2f=\max\limits_{\substack{S,S'\in\mathcal{S}^n\\ S\sim S'}} \|f(S)-f(S')\|_2.$
\end{definition}
 
Two standard differentially private mechanisms, which are used very often as building blocks, are the {\em Laplace} and the {\em Gaussian} Mechanism.

\begin{lemma}[Laplace Mechanism,~\cite{DworkMNS06}]\label{laplace}
Let $f\colon\mathcal{S}^n \to \R$, a dataset $S\in\mathcal{S}^n$, and privacy parameter $\eps$. The {\em Laplace Mechanism} 
\[\tilde{f}(S) = f(S) + \mathrm{Lap}(\Delta_1f/\eps)\]
is $\eps$-differentially private and with probability at least $1-\gamma$, 
$ |\tilde{f}(S)-f(S)| \leq \frac{\Delta_1f}{\eps}\ln\frac{1}{\gamma}.$
\end{lemma}

\begin{lemma}[Gaussian Mechanism,~\cite{DworkKMMN06}]\label{gauss}
Let $f\colon\mathcal{S}^n \to \R^d$, a dataset $S\in\mathcal{S}^n$, and privacy parameters $(\eps, \delta)$. The {\em Gaussian Mechanism} \[\tilde{f}(S) = f(S) + \cN(\mathbf{0},\sigma^2\idcov), \text{ where } \sigma = \Delta_2f\sqrt{2\ln(5/4\delta)}/\eps\]
is $(\eps, \delta)$-differentially private and with probability at least $1-\gamma$, $\|\tilde{f}(S)-f(S)\|_2 \le 4\sigma\sqrt{d\ln\frac{1}{\gamma}}$.
\end{lemma}

\subsection{Useful Facts on Distances Between Multivariate Distributions}\label{sec:distances}

We here record some lemmata which will be useful to us, relating total variation distance (equivalently, $L_1$) between the multivariate distributions we consider to the $\ell_2$ distances between their mean vectors.

The first is relatively standard; for the specific constants stated below, it is a direct consequence of~\cite[Theorem 1.2]{DevroyeMR18b}.
\begin{fact}
  \label{fact:dist:gaussians}
    Let $\mu,\nu\in\R^d$. Then, the $L_1$ distance between the two multivariate Normal distributions $\cN(\mu,\id_{d \times d})$, $\cN(\nu,\id_{d \times d})$ satisfies
    \[
        \frac{1}{100}\cdot \lVert \mu-\nu \rVert_2 \leq \lVert \cN(\mu,\id_{d \times d}) - \cN(\nu,\id_{d \times d}) \rVert_1 \leq 9\cdot \lVert \mu-\nu \rVert_2\,.
    \]
\end{fact}

The second relates, similarly, $L_1$ distance between product distributions over $\pmo^d$ to the $\ell_2$ distance between their means, however with a caveat~--~namely, at least one of the distributions needs to be ``balanced,'' i.e., have all its marginals ``not-nearly constant.'' 
\begin{lemma}
  \label{fact:dist:product}
    Fix $\tau\in(0,1]$, and let $P,Q$ be product distributions over $\pmo^d$ with mean vectors $\mu,\nu\in[-1,1]^d$ such that $-1 + \tau \leq \nu_i \leq 1 - \tau$ for all $i$. Then, the $L_1$ distance between $P$ and $Q$ satisfies
    \[
        c_\tau\cdot \lVert \mu-\nu \rVert_2 \leq \lVert P-Q \rVert_1 \leq C_\tau\cdot \lVert \mu-\nu \rVert_2\,,
    \]
    where $C_\tau,c_\tau>0$ are two constants depending only on $\tau$. Moreover, one can take $C_\tau = 1/\sqrt{\tau(1-\frac{\tau}{2})}$.
\end{lemma}
\begin{proof}
The upper bound follows from~\cite[Corollary 3.5]{CanonneDKS17} (note that their parameterization is
for $P,Q$ over $\{0,1\}^d$). As for the lower bound, it is proven analogously to~\cite[Lemma~6.4]{KamathLSU19}; with two main differences. First, their lemma is stated for $\tau = 2/3$, whereas we allow it to be an arbitrarily small constant -- the same argument carries through, although at the cost of larger constant factors in the bound.
  Second, their lemma requires that both $P$ and $Q$ be balanced, whereas we only require one to be balanced.
  This can be dealt with by noting that if one distribution is not balanced in some coordinate, the difference in means in this coordinate is sufficient to witness a large total variation distance.
\end{proof}

\section{From Gaussian to Product-of-Bernoulli Testing}\label{sec:gaussian:reduction}
In this short section, we provide a simple argument which enables us to transfer all our results on uniformity testing for product-of-Bernoulli distributions (described next, in Section~\ref{sec:put}) to testing identity of multivariate Normal distributions. Specifically, we analyze a simple reduction which maps a sample from an unknown multivariate normal $\cN(\mu,\id_{d \times d})$ to a sample from a product distribution $P_\mu$ on $\pmo^d$, such that the standard Normal $\cN(\textbf{0},\id_{d \times d})$ is mapped to $P_{\textbf{0}} = \unif$, while any Normal $\alpha$-far from $\cN(\textbf{0},\id_{d \times d})$ is mapped to some $P_\mu$ that is $\Omega(\alpha)$-far from $\unif$.

\begin{theorem}
  \label{theo:gaussian:product:reduction}
    There exists a function $F\colon\R^d\to\pmo^d$ and an absolute constant $c>0$ such that the following holds. For $\mu\in\R^d$, denote by $P_{\mu}$ the distribution of $F(X)$ when $X$ is drawn from $\cN(\mu,\id_{d \times d})$. Then
    \begin{itemize}
      \item If $\mu=\textbf{0}$, then $P_{\textbf{0}} = \unif$ is the uniform distribution on $\pmo^d$;
      \item $P_{\mu}$ is a product distribution over $\pmo^d$ such that $\|P_\mu - \unif\|_1 \geq c\cdot \|\cN(\mu,\id_{d \times d})-\cN(\textbf{0},\id_{d \times d})\|_1$.
    \end{itemize}
    Moreover, $F$  is computable in linear time.
\end{theorem}
\noindent Before providing the proof of the theorem, we note that, with this reduction, Theorems~\ref{thm:gaussian-fast-intro} and~\ref{thm:gaussian-slow-intro} directly follow from their respective counterparts, Theorems~\ref{thm:product-fast-intro} and~\ref{thm:product-slow-intro}.
\begin{proof}
The mapping $F\colon\R^d\to\pmo^d$ is defined coordinate-wise, by setting $F(x)_i := \operatorname{sgn}(x_i)$ for all $i\in[d]$; thus trivially implying the first item, as well as the time-efficiency statement and the fact that $P_{\mu}$ is a product distribution. We turn to the remaining part of the second item. Fix any $\mu\in\R^d$, and define for convenience $\alpha := \|\cN(\mu,\id_{d \times d})-\cN(\textbf{0},\id_{d \times d})\|_1 \in[0,2]$; by Fact~\ref{fact:dist:gaussians}, we have $\lVert \mu \rVert_2 \geq \frac{\alpha}{9}$.
For every $i\in[d]$, 
\[
    \ex{X\sim \cN(\mu,\id_{d \times d})}{F(X)_i} = 2\pr{}{F(X)_i=1}-1 = 2\pr{}{X_i>0}-1 = \operatorname{Erfc}( -\mu_i/\sqrt{2} ) -1
    = -\operatorname{Erf}( -\mu_i/\sqrt{2} )\,.
\]
Therefore, the mean vector $\mu'$ of $P_\mu$ satisfies
\[
      \lVert \mu'\rVert_2^2 = \sum_{i=1}^d \operatorname{Erf}( -\mu_i/\sqrt{2} )^2 \geq 0.84^2 \sum_{i=1}^d \min( \mu_i^2/2, 1 )
      = 0.84^2 \min( \normtwo{\mu}^2/2, d )
      \geq 0.84^2 \alpha^2/162\,,
\]
the first inequality by Lemma~\ref{lem:gauss:erfbound} (which we will prove momentarily) and the last by our above lower bound on $\lVert\mu\rVert_2$. This shows that $\lVert \mu'\rVert_2 > \alpha/12$. Applying finally Lemma~\ref{fact:dist:product} (with $Q=\unif$, for which $\tau=1$), we get existence of an absolute constant $c>0$ such that $\|P_\mu - \unif\|_1 \geq c\alpha$, as sought.
\end{proof}
\noindent The above argument relied on a technical lemma about the error function $\erf$, which we state and prove below.
\begin{lemma}\label{lem:gauss:erfbound}
For all $t\in\R$, $|\erf(-t)| \geq 0.84\cdot \min\left\{|t|, 1\right\}$.
\end{lemma}
\begin{proof}
Consider the functions $g(t)=\left(\erf(-t)\right)^2$ and $f(t)=\frac{g(t)}{t^2}$ for $t>0$. It holds that
\begin{equation}\label{eq:gprime}
g'(t)=2\erf(-t)\cdot\left(\erf(-t)\right)'
=-\frac{4}{\sqrt{\pi}}\exp(-t^2)\erf(-t).
\end{equation}
So $g(t)$ is increasing in $(0,\infty)$, since $g'(t)>0$. Then, for $t>1$, $g(t)>g(1)=\left(\erf(-1)\right)^2>0.84^2$. Thus,
\begin{equation}\label{eq:gauss:erfbound1}
|\erf(-t)|>0.84~\forall t>1.
\end{equation}
Now consider the function $f(t)$ in $(0,\infty)$.
\begin{align*}
f'(t)&=\frac{g'(t)}{t^2} -\frac{2g(t)}{t^3} 
=-\frac{4}{\sqrt{\pi}t^2}\exp(-t^2)\erf(-t)-\frac{2}{t^3}\left(\erf(-t)\right)^2 \tag{by~\eqref{eq:gprime}}\\
&= -\frac{4\exp(-t^2)\erf(-t)}{\sqrt{\pi}t^3}\left(t-\exp(t^2)\int_{0}^{t}\exp(-x^2)\,dx\right) \tag{since $-t<0$}
\end{align*}
Let us denote $h(t)=t-\exp(t^2)\int_{0}^{t}\exp(-x^2)\,dx$ for $t\in[0,\infty)$. It holds that 
\[
h'(t)=1-\left(\exp(t^2)\right)'\int_{0}^{t} \exp(-x^2)\,dx-\exp(t^2)\exp(-t^2)=-2t\exp(t^2)\int_{0}^{t} \exp(-x^2)\,dx\le 0.
\]
So $h(t)$ is decreasing in $[0,\infty)$ and $h(t)<h(0)=0$ for $t>0$. It follows that $f'(t)<0$, so $f(t)$ is decreasing in $(0,\infty)$. Thus, for $t\in(0,1]$, $f(t) \ge f(1)=0.84^2$. Equivalently, 
\begin{equation}\label{eq:gauss:erfbound2}
|\erf(-t)|\geq 0.84 |t|~\forall t\in(0,1].
\end{equation}
By inequalities~\eqref{eq:gauss:erfbound1} and~\eqref{eq:gauss:erfbound2}, we conclude that $|\erf(-t)| \ge 0.84\cdot \min\left\{|t|, 1\right\}$ for $t>0$ and the claim follows by symmetry.
\end{proof}

\section{Uniformity Testing for Product-of-Bernoulli Distributions}\label{sec:put}
In this section we introduce our algorithms for uniformity testing.  Both algorithms use a noisy version of the test statistic for uniformity testing introduced in~\cite{CanonneDKS17}: 
\begin{equation}\label{teststatistic}
T(X) = \sum_{i=1}^d (\barx_i^2-n).
\end{equation} 
The tests differ in the way that they reduce the sensitivity of the test statistic, giving tradeoffs between sample complexity and computational complexity.

\subsection{A Computationally Inefficient Private Algorithm}\label{sec:uniformity:inefficient}
\newcommand{\BLET}{LipschitzExtensionTest}

\bgroup\color{red}
\begin{framed}
As mentioned in the introduction, the argument provided in this section is not correct, due to an error in the proof of Lipschitzness of our statistic on the ``good set'' $\cC(\Delta)$ (see Lemma~\ref{sensitivityT})
-- Lipschitzness that we need to establish before invoking the McShane--Whitney extension theorem. In more detail, the proof of Lemma~\ref{sensitivityT} first establishes this Lipschitzness on pairs of adjacent databases in $\cC(\Delta)$, before extending it to arbitrary pairs of databases via the triangle inequality. Unfortunately, this last step is not valid, as the ``intermediate'' databases involved in application of the triangle inequality might not lie in $\cC(\Delta)$ themselves. While this issue can be fixed by modifying the definition of $\cC(\Delta)$, we chose not to do so, as this is the route taken by follow-up work by Shyam Narayanan, who fixed and improved upon in~\cite{Narayanan22} the bound stated in Theorem~\ref{thm:main_ineff} (and who pointed out to us the mistake in our proof). We here leave the current, flawed approach as-is to provide context and the basis for the aforementioned follow-up work; we strongly encourage the reader to refer to~\cite{Narayanan22} for the improved (and correct) version of Theorem~\ref{thm:main_ineff}.
\end{framed}
\egroup

In this section we focus on a computationally inefficient algorithm for uniformity testing based on Lipschitz extensions. This algorithm has two main components: a Lipschitz extension that allows us to control the amount of noise added to the test statistic, and an iterative step that rejects a successively larger class of distributions. The key idea is that the statistic $\|\bar{X}\|_2^2$ is related to both the test statistic and the sensitivity of the test statistic. Algorithm~\ref{algo:put_ineff} proceeds in rounds. Suppose that if $X$ has survived until round $m$ then with high probability $X\in \cC(\Delta^{(m)})$.
The test statistic $T$ has sensitivity $\Delta^{(m)}$ on $\cC(\Delta^{(m)})$. Thus, we can then compute a test statistic $\hatt$ that agrees with $T$ on $\cC(\Delta^{(m)})$, but has sensitivity $\Delta^{(m)}$ on all datasets, as in the $\textsc{\BLET}$ (Algorithm~\ref{algo:blet}). The existence of such a Lipschitz extension is guaranteed by the McShane-Whitney extension theorem~\cite{McShane34}. If $X$ passes the test, then we have a new bound for $\|\barx\|_2^2$, which implies a new bound $\Delta^{(m+1)}$ for the sensitivity of $T(X)$, which we use in the next round. It turns out that this bound is decreasing, so we can recurse until we either conclude that the data is not uniform, or we have a sufficiently small upper bound on the sensitivity $\Delta^\ast$. With this bound and if $X$ is not rejected in any of the iterations of Algorithm~\ref{algo:put_ineff}, we run one last Lipschitz test with threshold $O(n^2\alpha^2)$ which will determine if the data is uniform or $\alpha$-far from uniform.

Let $P=P_1\times\cdots\times P_d$ be a product distribution over $\pmo^d$, which is specified by its mean vector $p = (p_1, \cdots, p_d) \in [-1,1]^d$. We will denote the product of $d$ independent copies of the uniform distribution by $\unif$. We draw samples from distribution $P$ and aim to distinguish between the cases $P=\unif$ and $\|P-\unif\|_1\geq \alpha$ with probability at least $2/3$.

The main theorem of this section is the following:
\begin{theorem}\label{thm:main_ineff}
Algorithm~\ref{algo:put_ineff} is $\eps$-differentially private. Furthermore, 
Algorithm~\ref{algo:put_ineff} can distinguish between the cases $P=\unif$ and $\|P-\unif\|_1\geq \alpha$ with probability at least $2/3$ and has sample complexity 
\[
n = \tilde{O}\left(\frac{d^{1/2}}{\alpha^2}+\frac{d^{1/2}}{\alpha\eps^{1/2}}+\frac{d^{1/3}}{\alpha^{4/3}\eps^{2/3}}+\frac{1}{\alpha\eps}\right).
\]
\end{theorem}

Our test (Algorithm~\ref{algo:put_ineff}) will rely on a private adaptation of the test from~\cite{CanonneDKS17}; for this reason, we begin by recalling the guarantees of the test statistic $T$ developed in~\cite{CanonneDKS17}. In their work, Canonne \etal{}\ use Poisson sampling; however the interplay between privacy (which is defined with regard to a fixed set of samples) and Poisson sampling (where the number of samples is itself randomized) is tricky at best. For this reason, we state a version of their result without Poisson sampling, whose proof can be found in Appendix~\ref{app:proofs}.

\begin{lemma}[Non-private Test Guarantees]\label{non-priv}
For the test $T$ defined in \eqref{teststatistic}, the following hold:
\begin{itemize}
    \item If $P=\unif$ then $\E[T(X)]=0$ and $\var(T(X)) \leq 2n^2d$.
    \item If $\|P-\unif\|_1\ge \alpha$ then $\E[T(X)]> \frac{1}{2}n(n-1)\alpha^2$.
    \item $\var(T(X)) \leq 2n^2d + 4n \E[T(X)]$.
\end{itemize}
\end{lemma}

\begin{algorithm}[ht] 
\caption{$\textsc{\BLET}$}\label{algo:blet}
\begin{algorithmic}[1]
\Require{Sample $X = (X^{(1)},\dots,X^{(n)})$.  Parameters $\eps, \Delta > 0, \beta \in(0,1]$.}
  \State Define the set $\cC(\Delta) = \left\{X \in  \pmo^{n\times d} \bigm| \forall j\in [n],\;|\langle X^{(j)}, \barx \rangle| \leq \Delta\right\}$. 
	\State\label{step:blet:extension} Let $\hatt(\cdot)$ be a $4\Delta$-Lipschitz extension of $T$ from $\cC(\Delta)$ to all of $\pmo^{n\times d}$. \label{put_ineff:lipschitz}  
	\State Sample noise $r \sim \Lap(4\Delta/\eps)$ and let $z \gets \hatt(X) + r$. \label{recurlaplace}
	\If{$z > 10n\sqrt{d} + 4\Delta\ln(1/\beta)/\eps$}\label{step:blet:threshold}
	  \State \Return $\reject$.
	\EndIf
	\State \Return $\accept$.
\end{algorithmic}
\end{algorithm}
\begin{algorithm}[ht] 
\caption{Private Uniformity Testing via Lipschitz Extension}\label{algo:put_ineff}
\begin{algorithmic}[1]
\Require{Sample $X = (X^{(1)},\dots,X^{(n)}) \in \pmo^{n\times d}$ drawn from $P^n$.  Parameters $\eps, \alpha, \beta> 0$.}
\State Let $M \gets \lceil\log n\rceil$, $\eps' \gets \eps/M$, $\beta \gets 1/(10n)$.
\State Let $\Delta^{(1)} \gets nd$ and $\Delta^* \gets 1000\max(d, \sqrt{nd}, \ln(1/\beta)/\eps')\cdot\ln(1/\beta)$.
\For{$m \gets 1$ to $M-1$}                    
	\If{$\Delta^{(m)} \leq \Delta^*$}
	\State Let $\Delta^{(M)} \gets \Delta^{(m)}$ and exit the loop.
	\Else
	  \If{ $\textsc{\BLET}(X,\eps',\Delta^{(m)},\beta)$ returns $\reject$} \label{step:call:blet}
	      \State \Return $\reject$
	  \EndIf
	  \State Set $\Delta^{(m+1)}\gets 11\left( d + \sqrt{ nd} + \frac{\Delta^{(m)}}{n\eps'}+ \sqrt{ \frac{\Delta^{(m)}}{\eps'} } \right)\ln\frac{1}{\beta}$.
	\EndIf
\EndFor
\State Define the set $\mathcal{C}(\Delta^{(M)}) = \left\{X \in  \pmo^{n\times d} \bigm| \forall j\in [n],\;|\langle X^{(j)}, \barx \rangle| \leq \Delta^{(M)}\right\}$. 
\State Let $\hatt(\cdot)$ be a $4\Delta^{(M)}$-Lipschitz extension of $T$ from $\mathcal{C}(\Delta^{(M)})$ to all of $\pmo^{n\times d}$. \label{lastlipschitz} 
\State Sample noise  $r \sim \Lap(4\Delta^{(M)}/\eps')$ and let $z \gets \hatt(X) + r$. \label{lastlaplace}
	\If{$z > \frac{1}{4}n(n-1)\alpha^2$}\label{step:laststepM:threshold}
	\Return $\reject$
	\EndIf
\State \Return $\accept$. \label{put_ineff:test} 
\end{algorithmic}
\end{algorithm}

\anote{flagging that it took me a minute to figure out why the stopping conditions were independent of $\alpha$. We say how the stopping condition is determined the intro but we may want to reiterate}

We first focus on the privacy guarantee of Theorem~\ref{thm:main_ineff}. The mechanism is the composition of $M$ invocations of the Laplacian mechanism in lines~\ref{recurlaplace} and line~\ref{lastlaplace}. Our privacy proof is based on the existence of Lipschitz extensions, as established by the following theorem.

\begin{lemma}[McShane--Whitney extension theorem, \cite{McShane34}]\label{lem:kir}
Let $\varphi\colon\cC\to\mathbb{R}$ be a real-valued, $L$-Lipschitz function defined on a subset $\cC$ of a metric space $\mathcal{M}$. Then, there exists an $L$-Lipschitz map $\hat{\varphi}: \mathcal{M}\to\mathbb{R}$ that extends $\varphi$, that is, $\varphi(x)=\hat{\varphi}(x)$ $\forall x\in \cC$.
\end{lemma}
In this work, we will invoke the McShane--Whitney extension theorem with the metric space $\mathcal{M}$ being the space of databases with the metric induced by the neighboring relation.

Let us define for our dataset $X$ and any $\Delta>0$, 
\[\cC(\Delta) = \left\{X \in  \pmo^{n\times d} \bigm| \forall j\in [n],\;|\langle X^{(j)}, \barx \rangle| \leq \Delta\right\}.\] 
The main element that we need for the privacy proof is the bound on the sensitivity of $T$ on $\cC(\Delta^{(m)})$ for all $m\in[M]$. 
This would ensure that $T$ is $4\Delta^{(m)}$-Lipschitz on $\cC(\Delta^{(m)})$, so the $4\Delta^{(m)}$-Lipschitz extensions $\hatt$ exist in all the rounds and lines~\ref{recurlaplace} and~\ref{lastlaplace} add enough noise to maintain privacy.
Note that the algorithm would be private regardless of the choice of values $\Delta^{(m)}$.

\begin{lemma}[Sensitivity of T]\label{sensitivityT}
For any bound $\Delta> 0$, for two neighboring datasets $X, X'\in\cC(\Delta)$, $|T(X)-T(X')|\leq 4\Delta.$
\end{lemma}
\begin{proof}
Without loss of generality, assume that $X$ and $X'$ differ on the $n$-th sample. Then, we can write 
$
X =(X^{(1)},\dots,X^{(n)}) \text{ and } X' =(X^{(1)},\dots,X'^{(n)})
$. We can now calculate the difference:
\begin{align*}
T(X)-T(X') & = \sum_{i=1}^d \Big[ \Big(\sum_{j=1}^n X^{(j)}_i\Big)^2 - \Big(\sum_{j=1}^n X'^{(j)}_i\Big)^2\Big]\\
& =  \sum_{i=1}^d \Big[ \Big( \sum_{j=1}^n X^{(j)}_i - \sum_{j=1}^n X'^{(j)}_i\Big)\cdot \Big(\sum_{j=1}^n X^{(j)}_i + \sum_{j=1}^n X'^{(j)}_i\Big)\Big] \\
& =  \sum_{i=1}^d \Big[ \Big(X^{(n)}_i - X'^{(n)}_i\Big)\cdot \Big(\sum_{j=1}^n X^{(j)}_i + \sum_{j=1}^n X'^{(j)}_i\Big)\Big] \\
& = \sum_{i=1}^d \Big[  X^{(n)}_i\cdot\Big(2\sum_{j=1}^n X^{(j)}_i + X'^{(n)}_i-X^{(n)}_i\Big) -X'^{(n)}_i\cdot\Big(2\sum_{j=1}^n X'^{(j)}_i + X^{(n)}_i-X'^{(n)}_i\Big)\Big]\\
& = 2\langle X^{(n)} , \barx \rangle - 2\langle X'^{(n)}, \barx'\rangle + \sum_{i=1}^d(X^{(n)}_i + X'^{(n)}_i)\cdot (X'^{(n)}_i - X^{(n)}_i)
\end{align*}
Therefore, we have 
\begin{equation}\label{eq:wholesensitivity}
T(X)-T(X') = 2  \langle X^{(n)}, \barx\rangle  - 2 \langle X'^{(n)}, \barx'\rangle + \|X'^{(n)}\|_2^2 - \|X^{(n)}\|_2^2.
\end{equation}
Observe that, because $X_i^{(n)}, X_i'^{(n)}$ are in $\{ \pm 1\}$, we will have $ \|X'^{(n)}\|_2^2 = \|X^{(n)}\|_2^2=d$, leading to
$
  T(X)-T(X') \le 2 ( \langle X^{(n)}, \barx\rangle - \langle X'^{(n)}, \barx'\rangle )\,,
$
as stated earlier, in~\eqref{eq:sensitivity-intro}. This readily implies the bound
\begin{equation*}
|T(X)-T(X')| \leq 2 | \langle X^{(n)}, \barx\rangle | + 2 | \langle X'^{(n)}, \barx'\rangle |.
\end{equation*}
Since $X,X'\in \cC(\Delta)$, we know that $\left| \langle X^{(n)}, \barx\rangle \right| \leq \Delta$ and $ \left| \langle X'^{(n)}, \barx'\rangle \right| \leq \Delta$. It follows that $|T(X)-T(X')| \leq 4\Delta$.
\end{proof}

\begin{lemma}[Privacy]\label{lemma:priv:baselipsch}
$\textsc{\BLET}(X,\eps,\Delta,\beta)$ is $\eps$-differentially private.
\end{lemma}
\begin{proof}
$\textsc{\BLET}(X,\eps,\Delta,\beta)$ only accesses the data $X$ via the Laplace Mechanism in step~\ref{recurlaplace}, which is $\eps$-DP by Lemma~\ref{laplace}.
\end{proof}

\begin{lemma}[Privacy]\label{th:priv_ineff}
Algorithm~\ref{algo:put_ineff} is $\eps$-differentially private.
\end{lemma}
\begin{proof}
Algorithm~\ref{algo:put_ineff} is a composition of $M$ invocations of $\textsc{\BLET}$ with DP parameter $\eps/M$, and of the Laplace mechanism in step~\ref{lastlaplace}. By Lemmas~\ref{laplace} and~\ref{lemma:priv:baselipsch}, each of these steps is individually $\eps'=(\eps/M)$-differentially private. Since these are the only steps that access the data, the privacy guarantee follows from Lemmas~\ref{lem:composition} and \ref{lem:post-processing}.
\end{proof}

Now we have established the privacy of Algorithm~\ref{algo:put_ineff}, we turn our attention to the utility guarantee. First, we prove a claim that will help us determine how the bound on the sensitivity of the statistic decreases in each round.
\begin{lemma}\label{lem:innerprodbound} If $X$ is drawn i.i.d.\ from a product distribution, then, with probability at least $1-2n\beta$,
\begin{equation}\label{eq:innerprodbound}
    \forall x\in X,\; |\langle x, \barx \rangle| \leq \frac{\|\barx\|_2^2}{n}+\sqrt{2}\|\barx\|_2\sqrt{\ln(1/\beta)}.
\end{equation}
\end{lemma}
\begin{proof}
Since the $X^{(j)}$'s are i.i.d., if we condition on the sum $\barx=\sum_{j=1}^n X^{(j)}$ then $\E[X^{(k)} \mid \barx]=\barx/n$ by symmetry. This implies that 
$\E[\langle X^{(k)}, \barx \rangle \mid \barx] =  \langle \E[X^{(k)} \mid \barx], \barx \rangle = \frac{1}{n}\lVert \barx\rVert_2^2$.

Moreover, we can rewrite the inner product as $\langle X^{(k)}, \barx \rangle =\sum_{i=1}^d X^{(k)}_i\barx_i=\sum_{i=1}^d Y^{(k)}_i$, where, for all $i\in[d]$ and $k\in[n]$, we have
\begin{equation}\label{eq:productdistr}
Y^{(k)}_i \mid \barx=X^{(k)}_i\barx_i \mid \barx= \begin{cases} +\barx_i, & \text{w.p. } \frac{1+\barx_i/n}{2} \\ -\barx_i, & \text{w.p. } \frac{1-\barx_i/n}{2}\end{cases}.
\end{equation}
That is, for all $k\in[n]$, the random variables $Y^{(k)}_i\ \mid \barx$ for $i\in[d]$ are independent with each $Y^{(k)}_i\in\{-|\barx_i|, |\barx_i|\}$.
By Hoeffding's inequality,
\begin{equation}\label{eq:condhoeff}
\Pr\left[ | \langle X^{(k)}, \barx \rangle - \|\barx\|_2^2/n | \geq t \mid \barx \right] \leq 
2\exp\left(-\frac{t^2}{2\|\barx\|_2^2}\right)
\end{equation}
Denote by $E_{\barx}$ the event that there exists $k\in[n]$ such that $\left| \langle X^{(k)}, \barx \rangle \right| \geq \|\barx\|_2^2/n + \sqrt{2}\|\barx\|_2\sqrt{\ln(1/\beta)}$.
Then, by union bound and inequality~\eqref{eq:condhoeff}, 
$\Pr\left[ E_{\barx} \mid \barx \right] \leq 2n\beta.$
Taking the expectation of the latter probability and since the bound holds for any $\barx$, we have that 
\[
  \Pr[E_{\barx}] = \E_{\barx}[ \mathbf{1}_{E_{\barx}} ] = \E_{\barx}[ \E[\mathbf{1}_{E_{\barx}}\mid \barx] ] = \E_{\barx}\left[\Pr\left[E_{\barx} \mid \barx \right]\right] \leq 2n\beta.
\]
as claimed.
\end{proof}

In what follows, as in Algorithm~\ref{algo:put_ineff}, we let $\Delta^* = 1000\max(d, \sqrt{nd}, M\ln(1/\beta)/\eps)\cdot\ln(1/\beta)$, where $\beta = 1/(10n)$ and $M = \lceil\log(n)\rceil$.

Based on the previous lemma, we now analyze the key subroutine, $\textsc{\BLET}$. We already established its privacy guarantee in Lemma~\ref{lemma:priv:baselipsch}. We will prove two additional properties:
\begin{enumerate}
  \item if $X$ is drawn from a product distribution and, for some $\Delta$, $\textsc{\BLET}$ returns $\accept$, then $X\in\mathcal{C}(\Delta')$, where $\Delta' \ll \Delta$ (Lemma~\ref{lemma:prod:baselipsch});
  \item if $X$ is drawn from the uniform distribution, $\textsc{\BLET}$ will return $\accept$ (Lemma~\ref{lemma:unif:baselipsch}).
\end{enumerate}

\begin{lemma}[Sensitivity reduction]\label{lemma:prod:baselipsch}
 Fix some $\Delta \geq 0$, and let $\Delta' = 11( d + \sqrt{ nd} + \Delta/(n\eps)+ \sqrt{ \Delta/\eps } )\ln(1/\beta)$. If the following four conditions hold: 
 \begin{enumerate}[(i)]
 \item $X$ is drawn from a product distribution, 
 \item $X$ satisfies \eqref{eq:innerprodbound}, 
 \item $X\in\mathcal{C}(\Delta)$, and 
 \item $\textsc{\BLET}(X,\eps,\Delta,\beta)$ returns $\accept$,
 \end{enumerate} then $X\in\mathcal{C}(\Delta')$ with probability at least $1-\beta$.
\end{lemma}

\anote{I itemized this because I found it quite difficult to parse before. Can be moved back}

\begin{proof}
Let $\Delta \geq 0$, and assume $X$ is drawn from a product distribution and satisfies $X\in\mathcal{C}(\Delta)$ and \eqref{eq:innerprodbound}.  
Since $X\in\mathcal{C}(\Delta)$, the Lipschitz extension $\hatt$ coincides with $T$ on line~\ref{step:blet:extension}, and therefore $\hatt(X)=T(X)=\|\barx\|_2^2-nd$. Since $\textsc{\BLET}$ returned \accept, we further know that
\[
    T(X) + r \leq 10n\sqrt{d} + \frac{4\Delta}{\eps}\ln\frac{1}{\beta}
\]
where $r$ is a $\Lap(4\Delta/\eps)$ random variable, and therefore with probability at least $1-\beta$ has magnitude at most $4\Delta\ln(1/\beta)/\eps$ by Lemma~\ref{laplace}. Thus, with probability $1-\beta$,
\[
  \|\barx\|_2^2\leq 10n\sqrt{d}+8\Delta\ln(1/\beta)/\eps + nd \leq 11(nd + \Delta\ln(1/\beta)/\eps).
  \]
Substituting the bound on $\|\barx\|_2^2$ in inequality~\eqref{eq:innerprodbound}, we get that with probability at least $1-\beta$, for all $x\in X$ 
\begin{equation}\label{eq:newdelta} %
|\langle x, \barx \rangle|  \leq 11\left( d + \sqrt{ nd \ln\frac{1}{\beta}} + \frac{\Delta}{n\eps}\ln\frac{1}{\beta} + \sqrt{ \frac{\Delta}{\eps} }\ln\frac{1}{\beta} \right)
\leq 11\left( d + \sqrt{ nd} + \frac{\Delta}{n\eps}+ \sqrt{ \frac{\Delta}{\eps} } \right)\ln\frac{1}{\beta}
\end{equation}
This concludes the proof, as the RHS corresponds to our setting of $\Delta'$.
\end{proof}

\begin{lemma}\label{lemma:unif:baselipsch}
If (i)~$X$ is drawn from the uniform distribution $\unif$, satisfies $T(X) \leq 10n\sqrt{d}$, and \eqref{eq:innerprodbound}, and (ii)~$\Delta \geq \Delta^*$, then $\textsc{\BLET}(X,\eps,\Delta,\beta)$ returns $\accept$ with probability at least $1-\beta$.
\end{lemma}
\begin{proof}
  Since $T(X) = \|\barx\|_2^2 - nd$, our assumption implies that $\|\barx\|_2^2 \leq nd + 10n\sqrt{d} \leq 11nd$. Plugging this in~\eqref{eq:innerprodbound}, we have that 
\[
    |\langle x, \barx \rangle| \leq \frac{11nd}{n}+2\sqrt{nd\ln(1/\beta)} \leq 22\max(d, \sqrt{nd\ln(1/\beta)}) \leq \Delta^* \leq \Delta
\]
so that $X\in\mathcal{C}(\Delta)$. But when this happens, we have $\hatt(X)=T(X)$ in line~\ref{step:blet:extension} of $\textsc{\BLET}$, and thus we can analyze what happens in line~\ref{step:blet:threshold}  by bounding $T(X)+r$, which by our assumption is at most $10n\sqrt{d} + r$. By Lemma~\ref{laplace}, with probability at least $1-\beta$ we have $|r|\leq 4\Delta\ln(1/\beta)/\eps$. Whenever this holds, the algorithm outputs \accept.
\end{proof}

\noindent\cnote{We will prove $T(X) \leq 10n\sqrt{d}$ once and for all in the proof of the overall theorem, we shouldn't pay a union bound every time. That's why I put that in the assumption of the lemma. Ditto for satisfying \eqref{eq:innerprodbound}.}

With the analysis of this subroutine being done, we turn to the proof of Theorem~\ref{thm:main_ineff}. Unlike the privacy guarantee, this only needs to hold for datasets drawn from a product distribution. The crux of the proof can be summarized as follows:
\begin{enumerate}
\item If $X$ is drawn i.i.d.\ from a product distribution and passes line~\ref{recurlaplace} in round $m$, then it belongs in $\mathcal{C}(\Delta^{(m+1)})$ with high probability. Thus in every round that the dataset has not been rejected, we have that with high probability $\hatt(X)=T(X)$, so we are just running a noisy version of the non-private test. Further with each iteration we refine our bound on the sensitivity required.
\item If $P=\unif$, then with high probability $X$ passes all the steps in the loop at line~\ref{recurlaplace}.
\item The number of rounds $M$ is sufficient to guarantee that the sensitivity and thus the amount of noise added to $\hatt(X)$ in the last test in line~\ref{lastlaplace} is small enough that one distinguishes between the two hypotheses with the desired sample complexity.
\end{enumerate}

For our argument, we will need to show the noise added in line~\ref{lastlaplace} of Algorithm~\ref{algo:put_ineff} is small enough that we can still distinguish between the two hypotheses. The magnitude of that noise depends on the sensitivity, $\Delta^{(M)}$, of $T(X)$ restricted to the set $\mathcal{C}(\Delta^{(M)})$. The next lemma upper bounds $\Delta^{(M)}$.

\begin{lemma}\label{lem:convergence}
For $M= \lceil\log(n)\rceil$, $\eps' = \eps/M$, and $n=\Omega(\log(1/\beta)/\eps')$,
    \[
        \Delta^{(M)} \leq \Delta^* = 1000\max(d, \sqrt{nd}, \ln(1/\beta)/\eps')\cdot\ln(1/\beta).
    \]
\end{lemma}
\begin{proof}
Let us recall the update rule for the bound on the sensitivity of the statistic $T(X)$: $\Delta^{(1)} = nd$, and, for $m\geq 1$,
\[
  \Delta^{(m+1)}\gets 11\left( d + \sqrt{ nd} + \frac{\Delta^{(m)}}{n\eps'}+ \sqrt{ \frac{\Delta^{(m)}}{\eps'} } \right)\ln\frac{1}{\beta}\,.
\]
If at any point in the iteration $\Delta^{(m)} \leq \Delta^*$, then we set $\Delta^{(M)} = \Delta^{(m)}$ and exit the loop, which trivially proves the lemma.

\noindent Otherwise, assume that $\Delta^{(m)} > 1000\max(d, \sqrt{nd}, \ln(1/\beta)/\eps')\ln(1/\beta)$ for some $m$. Then,
\begin{align*}
    \Delta^{(m+1)} 
    &<  11\left( \frac{\Delta^{(m)}}{1000\ln(1/\beta)} + \frac{\Delta^{(m)}}{1000\ln(1/\beta)} + \frac{\Delta^{(m)}}{n\eps'}+ \sqrt{ \frac{\Delta^{(m)}}{\eps'} } \right)\ln\frac{1}{\beta} \\
    &\leq 11\left( \frac{\Delta^{(m)}}{1000} + \frac{\Delta^{(m)}}{1000} + \frac{\Delta^{(m)}}{1000}+ \sqrt{ \frac{(\Delta^{(m)})^2}{1000 \ln^2\frac{1}{\beta}} }\ln\frac{1}{\beta} \right) \\
    &\leq \frac{385}{1000} \Delta^{(m)} 
\end{align*}
  which is strictly less that $\Delta^{(m)}/2$. This implies that, for $M=\lceil\log(\Delta^{(1)}/\Delta^\ast)\rceil$ rounds, one must have $\Delta^{(M)} \leq \Delta^*$. Since $\Delta^\ast \geq d$, $M=\lceil\log(n)\rceil$ rounds suffice.
\end{proof}

We now complete our proof, showing that the last test will give the expected guarantees, based on all the previous lemmas.
\begin{proof}[Proof of Theorem~\ref{thm:main_ineff}] 
The privacy guarantee was established in Lemma~\ref{th:priv_ineff}.
It remains to prove that if $P=\unif$ then Algorithm~\ref{algo:put_ineff} outputs $\accept$ with probability at least $2/3$ (\textit{completeness}) and that if $\|P-\unif\|_1\geq \alpha$ then Algorithm~\ref{algo:put_ineff} outputs $\reject$ with probability at least $2/3$ (\textit{soundness}). Before starting, note that given $\Delta^{(1)} = nd$, we have $\mathcal{C}(\Delta^{(1)}) = \{\pm 1\}^{n\times d}$, so the guarantee that $X\in \mathcal{C}(\Delta^{(1)})$ is automatic.

\paragraph{Completeness:} Suppose $X$ is drawn from $\unif$. By~Lemma~\ref{non-priv}, Chebyshev's inequality implies that
\[
\Pr[T(X)\geq 10n\sqrt{d}] \leq \frac{2n^2d}{100n^2d} = \frac{1}{50}\,.
\]
Moreover, since $\unif$ is \emph{a fortiori} a product distribution, we get that $X$ satisfies \eqref{eq:innerprodbound} with probability at least $1-2n\beta$. We hereafter assume both of these events hold. 

By Lemma~\ref{lemma:unif:baselipsch} and a union bound over the $M-1$ calls to $\textsc{\BLET}$ in line~\ref{step:call:blet}, we get that with probability at least $1-M\beta$ every call returns \accept. Assume this is the case. Invoking now Lemma~\ref{lemma:prod:baselipsch}, this implies that for all $m$, $X\in\mathcal{C}(\Delta^{(m)})$ except with probability at most $M\beta$, and in particular $X\in\mathcal{C}(\Delta^{(M)})$.

In this case, we have $\hatt(X)=T(X)$ in line~\ref{lastlaplace}. It remains to show that with high probability, $z = T(X) + r \le \frac{1}{4}n(n-1)\alpha^2$ in line~\ref{step:laststepM:threshold}. 
Since $r\sim \mathrm{Lap}\left( 4\Delta^{(M)}/\eps' \right)$, by Lemma~\ref{laplace}, it holds that with probability at least $49/50$, $|r|\leq 4\Delta^{(M)}(\ln 50)/\eps'$. Again, condition on this, and suppose for now that $n$ is such that 
\begin{equation}\label{eq:requiredn}
|r| \le \frac{4\Delta^{(M)}\ln 50}{\eps'} \leq \frac{n(n-1)\alpha^2}{8}.
\end{equation}
Then, we have, by Lemma~\ref{non-priv} and Chebyshev's inequality, and recalling that $\hatt(X)=T(X)$, we can bound the probability that Algorithm~\ref{algo:put_ineff} rejects in line~\ref{lastlaplace} as
\begin{align*}
\Pr\left[z > \frac{n(n-1)\alpha^2}{4}\right] 
&= \Pr\left[T(X) >  \frac{n(n-1)\alpha^2}{4} - r\right] 
\leq \Pr\left[T(X) >  \frac{n(n-1)\alpha^2}{8}\right] \\
&\leq \frac{\var(T(X))}{\left(n(n-1)\alpha^2/8\right)^2} \leq \frac{128d}{(n-1)^2\alpha^4}\,,
\end{align*}
which is at most $1/50$ for $n=\Omega(d^{1/2}/\alpha^2)$. 
Therefore, by an overall union bound, we get that the algorithm rejects with probability at most
\[
    \frac{1}{50} + 2n\beta + M\beta + M\beta + \frac{1}{50} + \frac{1}{50} \leq \frac{3}{50} + 4n\beta \leq 1/3
\]
the last inequality by our setting of $\beta= 1/(10n)$.

\noindent To finish the proof of correctness in the completeness case, it remains to determine the constraint on $n$ for inequality~\eqref{eq:requiredn} to hold. By Lemma~\ref{lem:convergence} (and our setting of $\Delta^\ast$), it is enough to have,
\[
    n^2\alpha^2 = \Omega\left(\frac{\Delta^\ast}{\eps'}\right) = \Omega\left(\frac{ \max(d,\sqrt{nd}, \log(1/\beta)/\eps' )\cdot \log(1/\beta) }{\eps'}\right).
\]
Recalling our choice of $\eps'=\eps/M$, $\beta = 1/(10n)$, and $M=\lceil\log(n)\rceil$, it suffice to choose
\begin{equation}\label{eq:scprivate}
n = \tilde{\Omega}\left(\frac{d^{1/2}}{\alpha\eps^{1/2}}+\frac{d^{1/3}}{\alpha^{4/3}\eps^{2/3}}+\frac{1}{\alpha\eps}\right).
\end{equation}

\paragraph{Soundness:} Suppose $X$ is drawn from a product distribution which is $\alpha$-far from uniform. By Lemma~\ref{lem:innerprodbound}, $X$ satisfies \eqref{eq:innerprodbound} with probability at least $1-2n\beta$: we hereafter condition on this event. Suppose further that the algorithm does not return $\reject$ in any of the $M-1$ calls to $\textsc{\BLET}$ in line~\ref{step:call:blet} (otherwise, we are done): we will show that it will then output \reject after line~\ref{step:laststepM:threshold} with high probability. By Lemma~\ref{lemma:prod:baselipsch} and an immediate induction, the above implies that, with probability at least $1-M\beta$, $X\in\mathcal{C}(\Delta^{(m)})$ for all $m$, and in particular $X\in\mathcal{C}(\Delta^{(M)})$ -- so that $\hatt_M(X)=T(X)$. We once more assume this event holds.

Similar to the completeness proof (and relying on~\eqref{eq:requiredn}), with probability at least $49/50$,  $|r|\leq \frac{n(n-1)\alpha^2}{8}$. Conditioning on this, we then have, by Lemma~\ref{non-priv} and Chebyshev's inequality, that the algorithm outputs \accept in the last threshold test of line~\ref{step:call:blet} with probability at most
\begin{align*}
\Pr\left[z\leq \frac{n(n-1)\alpha^2}{4}\right] & = \Pr\left[\hatt_M(X)\leq \frac{n(n-1)\alpha^2}{4} - r\right] 
\leq \Pr\left[T(X) \leq  \frac{3n(n-1)\alpha^2}{8}\right] \\
& \leq \Pr\left[T(X) \leq  \frac{3\E[T(X)]}{4}\right] 
\leq \frac{16\cdot \var(T(X))}{\E[T(X)]^2} \\
& \leq \frac{64d}{(n-1)^2\alpha^4} + \frac{128}{(n-1)\alpha^2}\,.
\end{align*}
which less than $1/50$ for $n=\Omega\left(d^{1/2}/\alpha^2\right)$. 
By an overall union bound, this means the algorithm outputs \accept with probability at most
\[
    \frac{1}{50} + 2n\beta + M\beta + \frac{1}{50} + \frac{1}{50} \leq 1/3
\]
as desired.\medskip

\noindent 
We conclude that if 
\begin{equation}\label{eq:ineff:sample:complexity:long}
n = \tilde{\Omega}\left(\frac{d^{1/2}}{\alpha^2} + \frac{d^{1/2}}{\alpha\eps^{1/2}}+\frac{d^{1/3}}{\alpha^{4/3}\eps^{2/3}}+\frac{1}{\alpha\eps}\right)
\end{equation}
then Algorithm~\ref{algo:put_ineff} can distinguish between the cases $P=\unif$ and $\|P-\unif\|_1\geq \alpha$ with probability at least $2/3$.   This concludes the proof of Theorem~\ref{thm:main_ineff}.
\end{proof}

\subsection{A Computationally Efficient Private Algorithm}\label{sec:uniformity:efficient}
We next turn our attention to the computationally efficient algorithm, Algorithm~\ref{algo:put_eff}. The main focus in this section will be a computationally efficient analogue of a single Lipschitz extension step from Algorithm~\ref{algo:put_ineff}. In this algorithm, rather than reducing the sensitivity iteratively as in Algorithm~\ref{algo:blet}, we perform a single preconditioning step to initially reduce the sensitivity. 

The preconditioning first checks that the bias of each individual coordinate is not too large (line~\ref{reject1}). Next, recall that since the sensitivity of $T(X)$ depends on the bound on $|\langle x,\barx \rangle|$, we want this inner product to be small for all $x\in X$.  If $P=\unif$ then with probability $1-\delta$, $\forall j\in [n],$ 
\begin{equation}\label{cond:innerproduct}
|\langle X^{(j)}, \bar{X}\rangle|\le \Deff, 
\end{equation}
 for $\Deff$ as defined in Algorithm~\ref{algo:put_eff}. This is the property that we want our dataset to satisfy and it resembles the condition we put on the sets $\mathcal{C}$ of the previous section. If any data point in $X$ does not satisfy the inner product condition in \eqref{cond:innerproduct}, we can $\reject$ the dataset. Denote by $\mathcal{C}'\in\{\pm 1\}^{n\times d}$ the set of datasets for which all points satisfy \eqref{cond:innerproduct}.  Algorithm~\ref{algo:put_eff} proceeds by first attempting to verify membership in $\mathcal{C}'$ by privately counting the number of datapoints that violate condition~\eqref{cond:innerproduct} when compared to a private version of $\barx$, called $\tx$.
We reject if the private test determines with high probability that there is a non-zero number of such points (lines~\ref{noisymean}-\ref{reject2}). If $X$ survives the preconditioning step then Algorithm~\ref{algo:put_eff} attempts to ensure that $X\in\mathcal{C}'$ by replacing data points that do not satisfy \eqref{cond:innerproduct} (again compared to a noisy $\barx$) with draws from the uniform distribution (steps~\ref{forloop}-\ref{put:replacesample}). 
Since $X$ was already close to $\mathcal{C}'$, not many data points are changed so the resulting dataset, $\hatx$, still ``looks like'' $X$ but now has bounded inner products with high probability. The Lipschitz extension $\hat{T}(X)$ can then be replaced with $T(\hatx)$. 

Algorithm~\ref{algo:put_eff} has a higher sample complexity than Algorithm~\ref{algo:put_ineff}. This comes from two places. Firstly, $|\langle X^{(j)}, \barx\rangle|$ can be an additive factor of $d\log(1/\delta) / \eps$ larger than $|\langle X^{(j)}, \tx\rangle|$, which is the quantity we actually test. Moreover, when counting the number of points that violate condition \eqref{cond:innerproduct}, we again incur an error of order $d\log(1/\delta) / \eps$, which means that $X$ and $\hatx$ may differ more than we expect. This results in a $\sqrt{d\log(1/\delta)} / (\eps\alpha)$ term that was not present in the sample complexity of the inefficient algorithm. Note that the recursion of Algorithm~\ref{algo:put_ineff} decreases the bound for the corresponding $\Deff$ parameter but not the overhead of these steps, which dominate the component of the sample complexity due to privacy for Algorithm~\ref{algo:put_eff}.

Unlike Algorithm~\ref{algo:put_ineff}, which satisfied pure differential privacy, Algorithm~\ref{algo:put_eff} as written will only satisfy approximate differential privacy. The $\delta$ term arises since there is a small probability that the projection $\hatx$ could fail to be in $\mathcal{C}'$, in which case the amount of noise added to $T(\hatx)$ is insufficient for privacy. However, for hypothesis tests with constant error probabilities, sample complexity bounds for differential privacy are equivalent, up to constant factors, to sample complexity bounds for other notions of
distributional algorithmic stability, such as $(\eps,\delta)$-DP~\cite{DworkKMMN06}, concentrated DP~\cite{DworkR16,BunS16}, KL- and TV-stability~\cite{WangLF16,BassilyNSSSU16} (see,e.g., \cite[Lemma
5]{AcharyaSZ18}). The transformation of a test from $(\eps, \delta)$-DP to $O(\eps)$-DP given in \cite{AcharyaSZ18} is efficient so Algorithm~\ref{algo:put_eff} implies the existence of an efficient pure DP algorithm with the same sample complexity up to constant factors.

\begin{algorithm}[ht] \caption{Efficient Private Uniformity Testing}\label{algo:put_eff}
\begin{algorithmic}[1]%
\Require{Sample $X = (X^{(1)},\dots,X^{(n)}) \in \pmo^{n\times d}$ drawn from $P^n$.  Parameters $\eps, \delta, \alpha > 0$.}
\State Let $\barx \gets \sum_{j=1}^{n} X^{(j)}$.
\Algphase{Stage 1: Pre-processing}
\State Let $r_1 \sim \Lap(2/\eps)$ and $z_1 \gets \max_{i\in[d]} |\barx_i| + r_1$. \label{noisymeanmax}
  \If{$z_1 >\sqrt{2n\ln\frac{d}{\delta}} + \frac{2}{\eps}\ln\frac{1}{\delta}$}
      \Return $\reject$. \label{reject1}
  \EndIf
  \State Let $\tx \gets \barx + R$, where $R \sim \mathcal{N}(0,\sigma^2\idcov)$ and $\sigma =  \frac{ \sqrt{8d \ln(5/4\delta)}}{\eps}$. \label{noisymean}
  \State Let $\Deff \gets 16\big(d\ln\frac{d}{\delta}+\frac{d}{n\eps^2}\ln^2\frac{1}{\delta}+\sqrt{nd}\sqrt{\ln\frac{d}{\delta}\cdot\ln\frac{n}{\delta}}+\frac{\sqrt{d}}{\eps}\ln\frac{1}{\delta}\sqrt{\ln\frac{n}{\delta}}\big)$. \label{Deff}
  \State Let $r_2 \sim \Lap(1/\eps)$ and $z_2 \gets |\{j\in[n]: |\langle X^{(j)}, \tx\rangle| > \Deff +\frac{4d}{\eps}\sqrt{\ln\frac{5}{4\delta}\cdot\ln\frac{n}{\delta}}\}|+r_2$~\label{definition:z2}.
  \If{$z_2 > \frac{\ln(1/\delta)}{\eps}$}
    \Return $\reject$. \label{reject2} 
  \EndIf
\Algphase{Stage 2: Filtering}
\For{$j=1, \ldots, n$ \label{forloop}}
	\If{$|\langle X^{(j)}, \tx \rangle| > \Deff + \frac{4d}{\eps}\sqrt{\ln\frac{5}{4\delta}\cdot\ln\frac{n}{\delta}}$}
		\State $\hatx^{(j)} \gets U^{(j)}$ where $U^{(j)}\sim\unif$ \label{resample}
	\Else
		\State $\hatx^{(j)} \gets X^{(j)}$\label{put:replacesample}
	\EndIf
\EndFor
\Algphase{Stage 3: Noise addition and thresholding}
\State Let $r_3 \sim \Lap\big(\big(4\Deff+\frac{48d}{\eps}\sqrt{\ln\frac{5}{4\delta}\cdot\ln\frac{n}{\delta}}\big)/\eps\big)$ and $z_3 \gets T(\hatx) + r_3$. \label{finalnoise}
\If{$z_3 > \frac{1}{4} n(n-1)\alpha^2$}
  \Return $\reject$ 
\EndIf
\State \Return $\accept$.
\end{algorithmic}
\end{algorithm}

\begin{theorem}\label{thm:main_eff}
Algorithm~\ref{algo:put_eff} is $(4\eps, 14\delta)$-differentially private and distinguishes between the cases $P=\unif$ versus $\|P-\unif\|_1 \ge \alpha$ with probability at least $2/3$, having sample complexity
\[n = \tilde{O}\left(\frac{d^{1/2}}{\alpha^2} + \frac{d^{1/2}}{\alpha\eps}\right).\]
\end{theorem}

We will focus first on proving that Algorithm~\ref{algo:put_eff} is $(4\eps, 14\delta)$-differentially private. The majority of the work in the proof will involve showing that with high probability $\hatx$ satisfies a property similar to~\eqref{cond:innerproduct}. Since this relies on previous steps, the privacy proof is not simply an application of the composition theorem for differential privacy. We first show that lines~\ref{forloop}-\ref{put:replacesample}, which replace elements of $X$ with draws from the uniform distribution, result in new elements that satisfy \eqref{cond:innerproduct}.

\begin{lemma}\label{lem:uniformresample}
Let $X$ be a dataset that passes line~\ref{reject1} of Algorithm~\ref{algo:put_eff}. Suppose $U^{(j)}\sim \unif$ for $j\in[n]$. Then with probability $1-3\delta$, for all $j\in[n]$, \[\left|\langle U^{(j)}, \barx \rangle \right| \le \Deff.\]
\end{lemma}
\begin{proof}
If $X$ passes line~\ref{reject1}, we have that $\max_{i\in[d]} |\barx_i|\le \sqrt{2n\ln\frac{d}{\delta}}+\frac{2}{\eps}\ln\frac{1}{\delta} + |r_1|$. By Lemma~\ref{laplace}, with probability $1-\delta$, \[\max_{i\in[d]}|\barx_i| \le  \sqrt{2n\ln\frac{d}{\delta}}+\frac{4}{\eps}\ln\frac{1}{\delta} .\] Thus, if $x\sim\unif$, we have
$|\E\langle x, \barx \rangle| = 0,$ 
and, for all $i\in[d]$,
\[|x_i \cdot \barx_i |\le \sqrt{2n\ln\frac{d}{\delta}}+\frac{4}{\eps}\ln\frac{1}{\delta}.\]
Now, using Hoeffding's inequality and setting \[t=\sqrt{2d\ln\left(\frac{n}{\delta}\right)}\left(\sqrt{2n\ln\frac{d}{\delta}}+\frac{4}{\eps}\ln\frac{1}{\delta}\right)\] we have that with probability $1-\frac{2\delta}{n}$, it holds that 
\[ |\langle x, \barx \rangle| \leq |\E\langle x, \barx \rangle| + t \leq 8\left(\sqrt{nd}\sqrt{\ln\frac{d}{\delta}\cdot \ln\frac{n}{\delta}} + \frac{\sqrt{d}}{\eps}\ln\frac{1}{\delta}\sqrt{\ln\frac{n}{\delta}}\right) \leq \Deff.\]
By union bound, for all $j\in[n]$, with probability $1-3\delta$,  $|\langle U^{(j)}, \barx \rangle| \leq \Deff$.
\end{proof}

\begin{lemma}\label{lem:hatbarx_corr}
Let $X$ be a dataset that passes line~\ref{reject1} and line~\ref{reject2} of Algorithm~\ref{algo:put_eff}. Then with probability $1-6\delta$, it holds that, for every point $x\in \hatx$, 
\begin{equation}\label{lem:innerproductbound}
|\langle x, \bar{\hatx} \rangle | \leq \Deff+\frac{12d}{\eps}\sqrt{\ln\frac{5}{4\delta}\cdot\ln\frac{n}{\delta}}.
\end{equation}
\end{lemma}
\begin{proof} Since $X$ passed step~\ref{reject2} and by Lemma~\ref{laplace}, with probability $1-\delta$, \[\left|\left\{j\in[n]: |\langle X^{(j)}, \tx\rangle| > \Deff + \frac{4d}{\eps}\sqrt{\ln\frac{5}{4\delta}\cdot\ln\frac{n}{\delta}}\right\}\right|\le \frac{2}{\eps}\ln\frac{1}{\delta}.\] Therefore, $X$ and $\hatx$ differ in at most $\frac{2}{\eps}\ln\frac{1}{\delta}$ data points, so for every $x\in \hatx$,
\begin{equation}\label{eq:modifiedcorr1}
|\langle x, \bar{\hatx} \rangle| \le |\langle x, \barx \rangle| + |\langle x, \barx-\bar{\hatx} \rangle| \le |\langle x, \barx \rangle| +2d\cdot \frac{2}{\eps}\ln\frac{1}{\delta}= |\langle x, \barx \rangle| +\frac{4d}{\eps}\ln\frac{1}{\delta}.
\end{equation}
For any $x$ that was resampled during line~\ref{resample}, that is, $x=U^{(j)}$ for some $j\in [n]$, by Lemma~\ref{lem:uniformresample}, we are done. Otherwise, if $x$ was not resampled then by assumption $|\langle x, \tx \rangle| \le\Deff + \frac{4d}{\eps}\sqrt{\ln\frac{5}{4\delta}\cdot\ln\frac{n}{\delta}}$ and 
\begin{equation}\label{eq:noisycorr1}
 |\langle x, \barx \rangle| \le |\langle x, \tx \rangle| + |\langle x, \barx-\tx \rangle | \le \Deff +  \frac{4d}{\eps}\sqrt{\ln\frac{5}{4\delta}\cdot\ln\frac{n}{\delta}}+ |\langle x, \barx-\tx \rangle|. 
 \end{equation}
Now, $\barx-\tx=R$ where $R \sim \mathcal{N}(0,\sigma^2 \idcov)$ and $\sigma = \sqrt{8d\ln(5/4\delta)}/\eps$.
By symmetry, \[\langle x, \barx-\tx \rangle = \langle \textbf{1}, R \rangle = \sum_{i=1}^{d} Y_i,\] where each $Y_i\sim\mathcal{N}(0,\sigma^2)$. Thus, with probability $1-2\delta$, for all $x\in X$,
\begin{equation}\label{eq:noisycorr}
|\langle x, \barx-\tx \rangle| \le \frac{4d}{\eps}\sqrt{\ln\frac{5}{4\delta}\cdot\ln\frac{n}{\delta}}.
\end{equation}
Therefore, with probability $1-6\delta$, by inequalities~\eqref{eq:modifiedcorr1} and~\eqref{eq:noisycorr1}, $\forall x\in \hatx$, \[|\langle x, \bar{\hatx} \rangle|\le \Deff+ \frac{12d}{\eps}\sqrt{\ln\frac{5}{4\delta}\cdot\ln\frac{n}{\delta}}.\]
This completes the proof.
\end{proof}

We have now shown that lines~\ref{reject1}-\ref{reject2} with high probability either reject $X$ or alter it to ensure that it satisfies \eqref{lem:innerproductbound}, which is close to the desired condition~\eqref{cond:innerproduct}. The following lemma, which follows directly from Lemma \ref{sensitivityT}, says condition~\eqref{lem:innerproductbound} is sufficient to ensure that $T$ has sufficiently low sensitivity.

\begin{lemma}[Sensitivity of $T$]\label{lem:sensitivity}
Suppose $\hatx$, $\hatx'$ are two sample sets of size $n$, which differ in the $n$th data point. If both $\hatx$ and $\hatx'$ satisfy \eqref{lem:innerproductbound} then 
\[|T(\hatx) - T(\hatx')| \leq 4\Deff+\frac{48d}{\eps}\sqrt{\ln\frac{5}{4\delta}\cdot\ln\frac{n}{\delta}}.\]
\end{lemma}

We have now established that with high probability, if $X$ survives until line~\ref{finalnoise} then this step adds enough noise to maintain privacy. The next lemma states that several of the other noise infusion steps are individually differentially private. Since these mechanisms are applications of well-known privacy mechanisms described in section~\ref{sec:DPbackground}, we will not include their proofs here.

\begin{lemma}\label{lem:dp-halfalgo} In the notation of Algorithm~\ref{algo:put_eff}, 
\begin{itemize}
\item the mechanism $X\mapsto z_1$ in line~\ref{noisymeanmax} is $\eps$-DP, 
\item the mechanism $X\mapsto \tilde{X}$ in line~\ref{noisymean} is $(\eps, \delta)$-DP, and 
\item given a fixed (data independent) $\tilde{X}$, $X\mapsto z_2$ in line~\ref{definition:z2} is $\eps$-DP; therefore,
\item the mechanism $X\mapsto z_2$ in line~\ref{definition:z2} is $(2\eps, \delta)$-DP (by adaptive composition).
\end{itemize}
\end{lemma}

We can now complete our proof of the privacy guarantees of Algorithm~\ref{algo:put_eff}. As mentioned earlier, this is not simply an application of the composition theorem for differential privacy since the privacy guarantee of line~\ref{finalnoise} relies on the failure rates of earlier steps. Given two neighboring datasets $X$ and $X'$, the proof couples the random variables in separate runs of Algorithm~\ref{algo:put_eff} on $X$ and $X'$, to ensure that $\hatx$ and $\hatx'$ are neighboring datasets so we can apply Lemma~\ref{lem:sensitivity}. 

\begin{lemma}\label{eff_priv} Algorithm~\ref{algo:put_eff} is $(4\eps, 14\delta)$-differentially private. \end{lemma}

\begin{proof}
Suppose without loss of generality that $X$ and $X'$ differ on the $n$th data point. 
We can think of the random process that maps $X$ to $\hatx$ as a series of random variables. We denote these random samples by $(r_1, R, r_2, U, u_n, r_3)$ where $r_1, R, r_2$ and $r_3$ are as in Algorithm~\ref{algo:put_eff}, $U$ is the random samples that occur in line~\ref{resample} on all data points except the $n$th data point and $u_n$ is the randomness potentially used on the $n$th data point in line~\ref{resample}. 
Note that since the algorithm can terminate before running the entire algorithm, not all runs of the algorithm will sample from all these random variables. 

Let $C_{X,\tilde{X}} = |\{j\in[n]: |\langle X^{(j)}, \tx\rangle| > \Deff + \frac{4d}{\eps}\sqrt{\ln\frac{5}{4\delta}\cdot\ln\frac{n}{\delta}}\}|$. If we denote the output of Algorithm~\ref{algo:put_eff} on dataset $X$ with randomness $(r_1,R,r_2,U,u_n, r_3)$ by $M(X\;|\; (r_1, R,r_2,U,u_n, r_3))$ then for any $r_1,R, r_2, U, u_n$ and $r_3$,
\begin{align*}
&M(X\;|\; (r_1, R,r_2,U,u_n, r_3)) \\
&= M(X'\;|\; (r_1+\lVert\bar{X}\rVert_\infty-\lVert\bar{X'}\rVert_\infty, R+\bar{X}-\bar{X'},r_2+C_{X,\tilde{X}}-C_{X', \tilde{X}},U, u_n, r_3+T(\widehat{X}_{(U,u_n)})-T(\widehat{X'}_{(U,u_n)})))
\end{align*}
Let $B_1$ be the event that both $X$ and $X'$ pass line~\ref{reject1} and line~\ref{reject2} and $B_2$ be the event that $\hat{\bar{X}}$ and $\hat{\bar{X'}}$ satisfy~\eqref{lem:innerproductbound}. By Lemma~\ref{lem:hatbarx_corr}, $\Pr[B_1\cap B_2^c]\le 12\delta$ and by Lemma~\ref{lem:sensitivity} if $B_1\cap B_2$ holds then, for all $S\subseteq \R$,\footnote{In what follows, we overlook the measurability issues, and implicitly restrict ourselves to measurable sets.} %
\[\Pr[r_3+T(\widehat{X}_{(U,u_n)})-T(\widehat{X'}_{(U,u_n)}) \in S]\le e^{\eps}\Pr[r_3\in S].\]
For $b\in\{\accept, \reject\}$, let $E_{X, b} = \{(r_1, R,r_2,U,u_n, r_3) \;|\; M(X \mid (r_1, R,r_2,U,u_n, r_3))=b\}$. Now,
\begin{align*}
\Pr[M(X)=b] &=\Pr[E_{X,b}\cap B_1^c]+\Pr[E_{X,b}\cap B_1] \\
&=\Pr[E_{X,b}\cap B_1^c]+ \Pr[E_{X,b} \cap B_1\cap B_2]+\Pr[E_{X,b}\cap B_1\cap B_2^c] \\
&\le\Pr[E_{X,b}\cap B_1^c]+ \Pr[E_{X,b} \cap B_1\cap B_2]+12\delta.
\end{align*} 
If $B_1^c$ holds then $X$ is rejected in either line~\ref{reject1} or line~\ref{reject2}. Since all the steps leading up to either of these lines are differentially private (i.e., $X\mapsto (z_1,z_2)$ is $(3\eps,\delta)$-DP by Lemma~\ref{lem:dp-halfalgo}), we have 
\[
    \Pr[E_{X,b}\cap B_1^c]\le e^{3\eps} \Pr[E_{X',b}\cap B_1^c] + \delta.
\]
Note that $E_{X,b}$ is also equal to
\begin{align*}
\{(r_1+\lVert\bar{X}\rVert_\infty-\lVert\bar{X'}\rVert_\infty, R+\bar{X}-\bar{X'},r_2+C_{X,\tilde{X}}-C_{X', \tilde{X}},U, &u_n, r_3+T(\widehat{X}_{(U,u_n)})-T(\widehat{X'}_{(U,u_n)}))\\
&\;|\; M(X'\mid (r_1, R,r_2,U,u_n, r_3))=b\},
\end{align*}
so by the definition of $B_1$ and $B_2$ and Lemma~\ref{lem:dp-halfalgo},
\begin{align*}
\Pr[E_{X,b}\cap B_1\cap B_2]&\le e^{4\eps}\Pr[E_{X',b}\cap B_1\cap B_2] + \delta.
\end{align*}
Therefore, 
\begin{align*}
\Pr[M(X)=b]
&\le \Pr[E_{X,b} \cap B_1^c ]+\Pr[E_{X,b} \cap B_1\cap B_2]+12\delta\\
&\le e^{3\eps} \Pr[E_{X',b}\cap B_1^c] +e^{4\eps}\Pr[E_{X',b}\cap B_1\cap B_2] +14\delta\\
&\le e^{4\eps} (\Pr[E_{X',b}\cap B_1^c] + \Pr[E_{X',b}\cap B_1\cap B_2]) +14\delta\\
&\le  e^{4\eps}\Pr[E_{X',b}]+14\delta \\
&= e^{4\eps}\Pr[M(X')=b]+14\delta\,,
\end{align*} 
concluding the proof.\qedhere

\end{proof}

We now turn to proving that Algorithm~\ref{algo:put_eff} distinguishes between the cases $P=\unif$ and $\|P-\unif\|_1\ge \alpha$ with probability 2/3. The crux of the proof can be summarized as 
\begin{enumerate}
\item If $P=\unif$, then with high probability $X$ passes the first check at line~\ref{reject1}.
\item  We can choose $\Deff$ such that:
\begin{enumerate}
\item If $X$ passes line~\ref{reject1} then $X=\hatx$ with high probability, so $T(X)=T(\hatx)$.
\item The amount of noise added to $T(\hatx)$ is small enough that one can still distinguish between the two hypotheses.
\end{enumerate}
\end{enumerate}

 The next two lemmas establish parts 1 and 2a of our proof outline.
 \begin{lemma}\label{uniformpassescond1}
With probability at least $1-3\delta$, if $P= \unif$ then $z_1 \le \sqrt{2n\ln \frac{d}{\delta}} + \frac{2}{\eps}\ln\frac{1}{\delta}$. 
\end{lemma}
\begin{proof}
If $P=\unif$ then $\E[\barx_i]=0$, $\forall i\in [d]$. Each $\barx_i$ is a sum of $n$ independent random variables in $[-1,1]$, therefore, by Hoeffding 's inequality, it holds that
\[\Pr[|\barx_i| \geq t] \leq 2\exp(-2t^2/4n).\]
For $t=\sqrt{2n\ln\frac{d}{\delta}}$,
\[\Pr\left[|\barx_i| \geq \sqrt{2n\ln\frac{d}{\delta}}\right] \leq \frac{2\delta}{d}.\]
Therefore, by a union bound, with probability at least $1-2\delta$, $|\barx_i|\leq \sqrt{2n\ln(d/\delta)}$ holds for all $i\in[d]$.
By Lemma~\ref{laplace}, it holds that with probability $1-\delta$, $|r_1| \leq \frac{2}{\eps}\ln\frac{1}{\delta}$.
We conclude that if $P=\unif$, then with probability $1-3\delta$, 
\[
  z_1=\max_{i\in[d]} |X_i| + r_1\leq \sqrt{2n\ln\frac{d}{\delta}} + \frac{2}{\eps}\ln\frac{1}{\delta}. \qedhere
\]
\end{proof}

\begin{lemma}\label{lem:nomodifications}
If $X$ passes line~\ref{reject1} then with probability $1-6\delta$, $X=\hatx$.
\end{lemma}
\begin{proof}
If $X$ passes line~\ref{reject1}, then by Lemma~\ref{laplace}, with probability $1-\delta$, \[\max_{i\in[d]}|\barx_i| \le  \sqrt{2n\ln\frac{d}{\delta}}+\frac{4}{\eps}\ln\frac{1}{\delta} .\] It follows that with probability $1-\delta$,
\begin{equation*}
\|\barx\|_2^2 =\sum_{i=1}^d \barx_i^2 \leq d\left( \sqrt{2n\ln\frac{d}{\delta}}+\frac{4}{\eps}\ln\frac{1}{\delta}\right)^2
 \leq 2d \left( 2n\ln\frac{d}{\delta} + \frac{16}{\eps^2}\ln^2\frac{1}{\delta} \right) \leq 16d\left( n\ln\frac{d}{\delta} + \frac{1}{\eps^2}\ln^2\frac{1}{\delta}\right).
\end{equation*}

Setting $\beta=\delta/n$ in Lemma~\ref{lem:innerprodbound}, we have that, with probability $1-2\delta$, $ \forall x\in X$,
\[| \langle x, \barx \rangle | \leq \frac{\|\barx\|_2^2}{n}+\sqrt{2}\|\barx\|_2\sqrt{\ln(n/\delta)}.\]
Substituting the bound on $\|\barx\|_2^2$ and by union bound we have that with probability $1-3\delta$,
\[ \forall x\in X \; | \langle x, \barx \rangle | \leq 16\left( d\ln\frac{d}{\delta} + \frac{d}{n\eps^2}\ln^2\frac{1}{\delta} + \sqrt{nd\ln\frac{d}{\delta}\cdot\ln\frac{n}{\delta}}+ \frac{\sqrt{d}}{\eps}\ln\frac{1}{\delta}\sqrt{\ln\frac{n}{\delta}}\right) = \Deff.\]

By inequality~\eqref{eq:noisycorr} and union bound, with probability $1-5\delta$, $\forall x\in X$,
\[
|\langle x, \tx\rangle| \le \Deff + \frac{4d}{\eps}\sqrt{\ln\frac{5}{4\delta}\cdot\ln\frac{n}{\delta}}.
\]

So, with probability $1-5\delta$, $|\{j\in[n]: |\langle X^{(j)}, \tx\rangle| > \Deff +\frac{4d}{\eps}\sqrt{\ln\frac{5}{4\delta}\cdot\ln\frac{n}{\delta}}\}|=0$. Since with probability $1-\delta$, $|r_2|<\frac{1}{\eps}\ln\frac{1}{\delta}$, it follows that $z_2 \le \frac{1}{\eps}\ln\frac{1}{\delta}$. Therefore, with probability $1-6\delta$, $X$ survives lines~\ref{reject1} and~\ref{reject2} and none of the points get changed in lines~\ref{forloop}-\ref{put:replacesample}, that is, $X=\hatx$.
\end{proof}

We are now ready to prove the main theorem of this section.
\begin{proof}[Proof of Theorem \ref{thm:main_eff}]
The privacy guarantee was established in Lemma~\ref{eff_priv}, it remains to prove completeness and soundness.
\paragraph{Completeness:} Suppose $P=\unif$. By Lemma~\ref{uniformpassescond1}, with probability $1-3\delta$, $X$ survives line~\ref{reject1}. Conditioned on surviving line~\ref{reject1}, by Lemma~\ref{lem:nomodifications}, with probability $1-6\delta$, we have that $X=\hatx$. Thus, by union bound, with probability $1-9\delta$, $X=\hatx$ and $T(\hatx)=T(X)$.

As in the proof of Theorem~\ref{thm:main_ineff}, the remainder of the proof involves showing that for the given choice of $n$, it holds that standard deviation of the test statistic does not overwhelm the signal, that is,
\begin{equation}\label{eq:samplecondition}
  \frac{\ln 20}{\eps}\left(4\Deff+\frac{48d}{\eps}\sqrt{\ln\frac{5}{4\delta}\cdot\ln\frac{n}{\delta}}\right)\le \frac{n(n-1)\alpha^2}{8}.
\end{equation}
If \eqref{eq:samplecondition} holds then with probability $0.95$, $|r_3|\le \frac{n(n-1)\alpha^2}{8}$. If in addition $n=\Omega\left(\frac{d^{1/2}}{\alpha^2}\right)$, then we can show that the final test returns $\accept$, except with constant probability. Overall, with probability at least $0.9-9\delta$, as in the proof of Theorem~\ref{thm:main_ineff}, Algorithm~\ref{algo:put_eff} will return $\accept$. For $\delta \le 0.02$, this translates to success probability at least $2/3$.

Condition~\eqref{eq:samplecondition} is satisfied provided
\begin{align}\label{uglySC}
\nonumber n=\Omega\Bigg(&\frac{d^{1/2}}{\alpha^2} 
+\frac{d^{1/2}}{\alpha\eps^{1/2}}\left(\ln\frac{d}{\delta}\right)^{1/2}
+\frac{d^{1/3}}{\alpha^{2/3}\eps}\left(\ln\frac{1}{\delta}\right)^{1/3}
+\frac{d^{1/3}}{\alpha^{4/3}\eps^{2/3}}\left(\ln\frac{d}{\delta}\right)^{1/3}\left(\ln\frac{d}{\alpha\eps\delta}\right)^{1/3}\\
&+\frac{d^{1/4}}{\alpha\eps}\left(\ln\frac{1}{\delta}\right)^{1/2}\left(\ln\frac{d}{\alpha\eps\delta}\right)^{1/4} 
+\frac{d^{1/2}}{\alpha\eps} \left(\ln\frac{1}{\delta}\right)^{1/4} \left(\ln\frac{d}{\alpha\eps\delta}\right)^{1/4} \Bigg)\,,
\end{align} which we will prove below matches our claimed sample complexity up to logarithmic factors.
\paragraph{Soundness:}
Let us assume that the algorithm does not return $\reject$ in line~\ref{reject1} or~\ref{reject2}, which would be the desired output in this case.
By Lemma~\ref{lem:nomodifications}, since $X$ passed line~\ref{reject1}, with probability $1-6\delta$, we have that $X=\hatx$.
The rest follows again as in the proof of Theorem~\ref{thm:main_ineff}.\medskip

The final sample complexity guarantee follows by observing that (up to polylogarithmic factors) one of these terms can never dominate the asymptotic sample complexity.
\begin{claim}\label{amgmsimplification}
For any choice of parameters $d, \alpha, \eps$, $\frac{d^{1/3}}{\alpha^{4/3}\eps^{2/3}} \leq \max\left\{ \frac{d^{1/2}}{\alpha^2}, \frac{d^{1/4}}{\alpha\eps}\right\}$.
\end{claim}
\begin{proof}
Let us assume that $\frac{d^{1/3}}{\alpha^{4/3}\eps^{2/3}} > \max\left\{ \frac{d^{1/2}}{\alpha^2}, \frac{d^{1/4}}{\alpha\eps}\right\}$. Then $\frac{d^{1/3}}{\alpha^{4/3}\eps^{2/3}} > \frac{1}{3}\cdot\frac{d^{1/2}}{\alpha^2} + \frac{2}{3}\cdot\frac{d^{1/4}}{\alpha\eps}$.
By the AM-GM inequality, it holds that 
\[
    \frac{1}{3}\cdot\frac{d^{1/2}}{\alpha^2} + \frac{2}{3}\cdot\frac{d^{1/4}}{\alpha\eps} \geq \frac{d^{1/6}}{\alpha^{2/3}} \cdot \frac{d^{2/12}}{\alpha^{2/3}\eps^{2/3}}=\frac{d^{1/3}}{\alpha^{4/3}\eps^{2/3}},
    \] which leads to a contradiction. Therefore, it must be that $\frac{d^{1/3}}{\alpha^{4/3}\eps^{2/3}} \leq \max\left\{ \frac{d^{1/2}}{\alpha^2}, \frac{d^{1/4}}{\alpha\eps}\right\}
    $.
\end{proof}

\noindent Ignoring the polylogarithmic factors and by Claim~\ref{amgmsimplification}, the sample complexity stated in~\eqref{uglySC} is simplified to
\[n=\tilde{O}\left(\frac{d^{1/2}}{\alpha^2} +\frac{d^{1/2}}{\alpha\eps}\right). \]
This concludes the proof of Theorem~\ref{thm:main_eff}.
\end{proof}

\section{Balanced Identity Testing of Product Distributions}\label{sec:balanced}
In previous sections, we considered only testing whether an unknown product distribution is uniform.
We will provide a generic reduction which preserves $\ell_2$-distance between means.
In the case where the distributions are ``balanced'' (i.e., the coordinate means are bounded away from being $-1$ or $1$), this will imply that our upper bounds carry over to this more general setting.

\begin{lemma}
  \label{lem:balanced}
  Suppose we are given a known binary product distribution $Q$ and a sample $X \sim P$, where $P$ is some unknown binary product distribution.
  There exists a (randomized) transformation $Y = f_Q(X)$ such that $Y \sim P'$, where $P'$ is some unknown binary product distribution such that $\|\E[P'] - \E[\unif]\|_2 = \|\E[P']\|_2 = \frac{1}{2}\|\E[P] - \E[Q]\|_2$.
\end{lemma}
\begin{proof}
  The transformation is as follows: for each coordinate $i \in [d]$, sample $b_i \sim \Ber(1/2)$.
  If $b_i = 0$, then $Y_i = X_i$: in this case, $\E[Y_i] = \E[P_i]$.
  Otherwise, let $Y_i = 1$ with probability $\frac{1 - \E[Q_i]}{2}$, and $-1$ with probability $\frac{1 + \E[Q_i]}{2}$ -- in this case, $\E[Y_i] = -\E[Q_i]$.
  Putting these cases together, we overall have that $\E[Y_i] = \frac{1}{2} \left(\E[P_i] - \E[Q_i]\right)$.
  Since each coordinate is independent, we have that $Y \sim P'$, where $P'$ is a binary product distribution with the same mean.
  Overall, this gives us that $\|\E[P']\|^2_2 = \frac{1}{4} \|\E[P] - \E[Q]\|^2_2$, which allows us to conclude the desired statement.
\end{proof}

\begin{corollary} \label{cor:balancedreduction}
  Let $\tau > 0 $ be some fixed constant, and $c = c(\tau)$ be some sufficiently small constant (which depends on $\tau$).
  Let $Q$ be some known product distribution over $\{\pm 1\}^d$, such that $-1 + \tau \leq \E[Q_i] \leq 1 - \tau$.

  Suppose there exists an algorithm which takes $n$ samples from an unknown product distribution $P'$ over $\{\pm 1\}^d$ and can distinguish between the following two cases with probability at least $2/3$: (U1) $P' = \unif$, (U2) $\|P' - \unif\|_1 \geq c\alpha$.
  Then there exists an algorithm which takes $n$ samples from an unknown product distribution $P$ over $\{\pm 1\}^d$ and can distinguish between the following two cases with probability at least $2/3$: (B1) $P = Q$, (B2) $\|P - Q\|_1 \geq \alpha$.
\end{corollary}
\begin{proof}
  The proof will follow via Lemma~\ref{lem:balanced}: given a set of samples for the latter problem, we convert them via the method of Lemma~\ref{lem:balanced}, and run the given algorithm for the former problem.
  The former case is immediate, since equal product distributions will have equal mean vectors.

  We thus consider the latter case, $\|P - Q\|_1 \geq \alpha$. The upper bound part of Lemma~\ref{fact:dist:product} implies that $\|\E[P] - \E[Q]\|_2 \geq \alpha/C$.
  Applying the conversion of Lemma~\ref{lem:balanced} gives that $\|\E[P'] - \unif\|_2 \geq \alpha/2C$.
  We then argue that $\|P' - \unif\|_1 \geq c\alpha$, which follows from the lower bound of Lemma~\ref{fact:dist:product} (applied with $P$ and $Q$ equal to our $P'$ and $\unif$, respectively), concluding the proof.
\end{proof}

\section{Extreme Product Distributions}
\label{sec:extreme}

In this section, we discuss algorithms and lower bounds for \emph{extreme product distributions}. 
Roughly speaking, an extreme product distribution is a product distribution with marginal distributions which are sufficiently close to being deterministically either $-1$ or $1$.
More precisely, we have the following:
\begin{definition}
  Fix any constant $C>0$. A \emph{$C$-extreme product distribution} is a product distribution over $\pmo^d$ with mean vector $(p_1, \dots, p_d)$ such that, for all $i \in [d]$, $|p_i - 1| \leq C/d$ or $|p_i + 1| \leq C/d$, and $\sum_{i=1}^d (1+p_i)/2 \leq C$. We often omit the dependence on the constant $C$ and refer to \emph{extreme distributions}.
\end{definition}
We will show that, up to constant factors, the sample complexity of identity testing for extreme product distributions and identity testing for univariate distributions are the same.
This statement holds even without privacy constraints, which we consider to be of independent interest. 
We will show this via a pair of sample-complexity preserving reductions, which will then allow us to immediately port results from the univariate private testing literature.

More precisely, we will show the following theorem.
\begin{theorem}
  \label{thm:extreme-reduction}
  For $n = \Omega((\log d)/\eps)$, there is an $O(n)$-sample $\eps$-differentially private algorithm for testing identity to any distribution $Q_{\mathrm{univ}}$ over $[d]$ if and only if there is an $O(n)$-sample $\eps$-differentially private algorithm for testing identity to any extreme product distributions $Q_{\mathrm{prod}}$ over $\{\pm 1\}^d$. 
\end{theorem}
This will be proven through a combination of Lemma~\ref{lem:extreme-to-univ} in Section~\ref{sec:extreme-to-univ}, and Lemma~\ref{lem:univ-to-extreme} in Section~\ref{sec:univ-to-extreme}.
With Theorem~\ref{thm:extreme-reduction} in place, we can conclude the following corollary, which is a consequence of the corresponding statements for private univariate testing in~\cite{AcharyaSZ18}.
\begin{corollary}
\label{cor:extreme-reduction}
For every extreme product distribution $Q \in \{\pm 1\}^d$, there exists an $\eps$-differentially private algorithm which takes $$n = O\left(\frac{d^{1/2}}{\alpha^2} + \frac{d^{1/2}}{\alpha\eps^{1/2}} + \frac{d^{1/3}}{\alpha^{4/3}\eps^{2/3}}+ \frac{1}{\alpha\eps}\right)$$ samples from some unknown product distribution $P \in \{\pm 1\}^d$ and distinguishes between the cases where $P = Q$ and $\|P - Q\|_1 \geq \alpha$ with probability at least $2/3$.
  Furthermore, every algorithm which has the same guarantees requires $\Omega(n)$ samples from $P$, as long as $\eps = \tilde{\Omega}(1/d)$.
\end{corollary}

\subsection{Reducing from Extreme Product Testing to Univariate Testing}
\label{sec:extreme-to-univ}
In this section, we show the following, which shows that any univariate identity testing algorithm implies a \emph{multivariate} identity testing algorithm for extreme product distributions.\footnote{We further note that the reduction is quite general, and can be straightforwardly adapted beyond differentially private algorithms.}
\begin{lemma}
  \label{lem:extreme-to-univ}
  There exists a constant $\gamma\in(0,1]$ (depending only on the parameter $C$ of extreme distributions) such that the following holds.
  If there is an $n$-sample $\eps$-DP algorithm for testing identity to a distribution $Q_\mathrm{univ}$ over $[d]$ (with distance parameter $\gamma\alpha$), then there is an $n'$-sample $\eps$-DP algorithm for testing identity to extreme product distributions over $\{\pm 1\}^d$ (with distance parameter $\alpha$) where $n'=O(n+(\log d)/\eps)$.
\end{lemma}

\begin{proof}
For convenience and ease of notation, we hereafter consider distributions over $\{0,1\}^d$ instead of the equivalent choice $\{\pm 1\}^d$, as this allows us to map more easily mean vectors of product distributions to probability vectors of univariate distributions. 
  Fix $C\geq 0$, and suppose there exists an algorithm $A$ for testing identity (with distance parameter $\alpha$) to distributions over $[d]$, with sample complexity $n(\alpha)$. Given a fixed $C$-extreme product distribution $Q_\mathrm{prod}$ (with mean vector $q\in[0,1]^d$) over $\{\pm 1\}^d$, and $n$ samples from an unknown product distribution $P$ over $\{\pm 1\}^d$ (with unknown mean vector $p\in[0,1]^d$), the claimed algorithm to test identity to $Q_\mathrm{prod}$ works as follows. First, without loss of generality (and up to flipping the corresponding coordinates in all samples from $P$), one can assume that $0\leq q_i \leq C/d$ for all $i\in[d]$.
  
The first step is to apply the differentially private algorithms from Lemmas~\ref{lemma:additional:infinitynormtest} and~\ref{lemma:additional:onenormtest}, with constant probability of failure $1/10$ and privacy parameter $\eps/3$ (and, for the second algorithm, for $\tau$ set to $\max(1,C)$), to the samples of $P$, and reject immediately if either test rejects (This spends a total ``privacy budget'' of $2\eps/3$, so that we still have $\eps/3$ to use in the rest, when calling the univariate purported tester).  These two tests are simultaneously correct with probability at least $4/5$; we then can continue assuming $\lVert p\rVert _\infty \leq 1/2$ and $\lVert p\rVert _1 \leq 8\max(1,C)$.
  
\noindent For any of the $n$ samples $X^{(1)},\dots, X^{(n)} \sim P$:
\begin{itemize}
  \item If $X^{(i)} = \mathbf{0}$, then output a sample with value $Y_i \gets d+1$;
  \item If $\lVert X^{(i)}\rVert = 1$, i.e., $X^{(i)}$ has exactly one non-zero coordinate, then output a sample $Y_i$ with the value of this coordinate;
  \item If $\lVert X^{(i)}\rVert > 1$, then output a sample with value $Y_i \gets d+2$.
\end{itemize}
This therefore generates $n$ i.i.d. samples from some distribution $P_\mathrm{univ}$ over $[d+2]$. Moreover, since $Q_\mathrm{prod}$ is fully known, the corresponding distribution $Q_\mathrm{univ}$ (which one would obtain by applying this transformation to samples from $Q_\mathrm{prod}$) is uniquely specified and known; the algorithm then invokes the univariate identity tester $A$ on the $n$ i.i.d. samples $Y_1,\dots, Y_n$ to test identity to $Q_\mathrm{univ}$.
\begin{itemize}
  \item If $P=Q_\mathrm{prod}$, then we have $P_\mathrm{univ}=Q_\textrm{univ}$.
  \item If $\lVert P-Q_\mathrm{prod} \rVert_1 > \alpha$, then $\lVert P_\mathrm{univ} - Q_\mathrm{univ} \rVert_1 > \gamma\alpha$, where $\gamma := \frac{e^{-C}}{8(1+16\max(1,C))}$. To prove this statement, we denote by $\mathbf{e}_j$ the $j$-th vector of the canonical basis of $\R^d$ and observe that by definition
  $
    P_\mathrm{univ}(d+1) = P(\mathbf{0}) = \prod_{i=1}^d (1 - p_i)
  $,
  while, for $1\leq j\leq d$,
  $
    P_\mathrm{univ}(j) = P(\mathbf{e}_j) = p_j\prod_{i\neq j} (1-p_i) = \frac{p_j}{1-p_j}\cdot P(\mathbf{0})
  $. Hence,
  \begin{align*}
      \lVert P_\mathrm{univ} - Q_\mathrm{univ} \rVert_1 &\geq{} \sum_{i=1}^{d+1} |P_\mathrm{univ}(i)-Q_\mathrm{univ}(i)| \\
      &={} |P(\mathbf{0})-Q_\mathrm{prod}(\mathbf{0})| + \sum_{i=1}^{d} |\frac{p_i}{1-p_i}\cdot P(\mathbf{0})-\frac{q_i}{1-q_i}\cdot Q_\mathrm{prod}(\mathbf{0})|
 \end{align*}
  If $|P(\mathbf{0})-Q_\mathrm{prod}(\mathbf{0})| > \gamma \alpha $, then we are good; otherwise, $|P(\mathbf{0})-Q_\mathrm{prod}(\mathbf{0})| \leq \gamma \alpha$, from which we can bound the second term as
  \begin{align*}    
     \sum_{i=1}^{d} &|\frac{p_i}{1-p_i}\cdot P(\mathbf{0})-\frac{q_i}{1-q_i}\cdot Q_\mathrm{prod}(\mathbf{0})| \\
      \geq{} &Q_\mathrm{prod}(\mathbf{0}) \sum_{i=1}^{d} \left| \frac{p_i}{1-p_i} -\frac{q_i}{1-q_i} \right| - \sum_{i=1}^{d}\frac{p_i}{1-p_i}  | P(\mathbf{0})- Q_\mathrm{prod}(\mathbf{0})|\\
      \geq{} &Q_\mathrm{prod}(\mathbf{0}) \sum_{i=1}^{d} \left| \frac{p_i}{1-p_i} -\frac{q_i}{1-q_i} \right| - \gamma\alpha \sum_{i=1}^{d}\frac{p_i}{1-p_i}\,.
  \end{align*}
  Observing that the function $f\colon x\in[0,1/2]\to \frac{x}{1-x}$ is smooth with $f'(x) \in [1/4, 1]$, we get $\frac{1}{4}\lvert x-y\rvert \leq \lvert f(x)-f(y)\rvert \leq \lvert x-y\rvert$ for all $x,y\in[0,1/2]$.and therefore
  \begin{align*}    
      \sum_{i=1}^{d} |\frac{p_i}{1-p_i}\cdot P(\mathbf{0})-\frac{q_i}{1-q_i}\cdot Q_\mathrm{prod}(\mathbf{0})|
      &\geq \frac{Q_\mathrm{prod}(\mathbf{0})}{4} \sum_{i=1}^{d} \left| p_i - q_i \right| - 2\gamma\alpha \sum_{i=1}^{d}p_i\\
      &\geq \frac{e^{-C+O(1/d)}}{4} \sum_{i=1}^{d} \left| p_i - q_i \right| - 16\max(1,C)\gamma\alpha\,,
  \end{align*}
  where for the last inequality we used the fact that $Q_\mathrm{prod}$ is $C$-extreme to bound $Q_\mathrm{prod}(\mathbf{0})$, and the fact that $\sum_{i=1}^{d}p_i \leq 8\max(1,C)$. For $d$ large enough, $\frac{e^{-C+O(1/d)}}{4}\geq \frac{e^{-C}}{8} = (1+16\max(1,C))\gamma$, and therefore (recalling the folklore subadditive bound on total variation distance with regard to product distributions)
  \begin{align*}    
      \sum_{i=1}^{d} |\frac{p_i}{1-p_i}\cdot P(\mathbf{0})-\frac{q_i}{1-q_i}\cdot Q_\mathrm{prod}(\mathbf{0})|
      &\geq (1+16\max(1,C))\gamma \sum_{i=1}^{d} \left| p_i - q_i \right| - 16\max(1,C)\gamma\alpha \\
      &\geq (1+16\max(1,C))\gamma \lVert P-Q_\mathrm{prod} \rVert_1 - 16\max(1,C)\gamma\alpha \\
      &> \gamma\alpha \,,
  \end{align*}
  as claimed.
\end{itemize}
Correctness then follows from that of the purported univariate identity tester, called with privacy parameter $\eps/3$ and failure probability $1/5$ (so that by a union bound, the overall correctness is $3/5$, and the whole procedure is $\eps$-differentially private).
\end{proof}

\noindent It only remains to provide the proofs of the two helper subroutines we used in the reduction:
\begin{lemma}
  \label{lemma:additional:infinitynormtest}
  There is an $\eps$-differentially private algorithm which can distinguish between the cases that a product distribution over $\{0,1\}^d$ with mean vector $p$ has $\|p\|_\infty \leq 1/4$ versus $\|p\|_\infty \geq 1/2$ using $n$ samples, for $n = O( (\log d)/\eps ).$
\end{lemma}
\begin{proof}
  This will follow by an application of Report Noisy Max (see~\cite{DworkR14}).
  Draw $n = \Omega(\frac{\log d}{\eps})$ samples from the product distribution, and consider the $d$ functions $f_1, \dots, f_d$, where $f_i$ computes the empirical marginal distribution for coordinate $i$.
  Note that each $f_i$ has sensitivity $1/n$, so by the guarantees of Report Noisy Max, it is $\eps$-differentially private to output $\hat f_{i^\ast} \triangleq \max_{i \in [d]} f_i + \Lap(1/n\eps)$.
  By a Chernoff bound, union bound, and tail bounds on Laplace random variables, the difference between $\hat f_i$ and $p_i$ will be bounded by $O(\sqrt{(\log d)/n} + (\log d)/(n\eps))$, simultaneously for all $i \in [d]$, with probability at least $99/100$.  By choosing $n = \Omega((\log d)/\eps)$, we upper-bound this error term by $1/8$.
  Therefore, thresholding the value of $\hat f_{i^\ast}$ at the value $3/8$ will distinguish the two cases, as desired.
\end{proof}

\begin{lemma}
  \label{lemma:additional:onenormtest}
  For any $\tau \geq 1$, there is an $\eps$-differentially private algorithm which can distinguish between the cases that a product distribution over $\{0,1\}^d$ with mean vector $p$ has $\|p\|_1 \leq \tau$ versus $\|p\|_1 \geq 8\tau$ using $n$ samples, for
  $n = O(1/\eps).$
\end{lemma}
\begin{proof}
  The algorithm will first draw $n$ samples, and compute the fraction $f$ of these draws which have at least $8\tau$ ones.
  Note that this statistic has sensitivity $1/n$, so to privatize it, we can add a $\Lap(1/n\eps)$ random variable, giving us a statistic $\hat f$.
  If the result is at most $3/8$, then we output that $\|p\|_1 \leq \tau$, else, we output that $\|p\|_1 \geq 8\tau$.

  Let $r$ be the probability that a single string has at least $4\tau$ ones.
  We start by showing that there exists a gap in the value of $r$ in the two cases.
  Let $N$ denote the number of $1$'s is a randomly drawn string from $p$. We have $\EE[X]\leq \tau$ and $\EE[X] \geq 8\tau\geq 8$ in the two cases, respectively.
  By Markov's inequality, this means that in the first case, $r \leq 1/4$.
   Moreover, a simple computation shows that
  $\var[X] = \EE[X] - \sum_{i=1}^d p_i^2 \leq \EE[X]$, so that, by Chebyshev's inequality, in the second case we get
  $1-r \leq \pr{}{|X-\EE[X]| > \EE[X]/2} \leq 4/\EE[X] \leq 1/2$, and therefore $r\geq 1/2$.  
  By a Chernoff bound and a tail bound on Laplace random variables, the difference between the $\hat f$ and $r$ will be bounded by $O\left(1/\sqrt{n} + 1/(n\eps)\right)$.
  If we choose $n = \Omega\left(1/\eps\right)$, we upper bound this error term by $1/16$, and thresholding at $3/8$ will distinguish the two cases, as desired.
\end{proof}

\subsection{Reducing from Univariate Testing to Extreme Product Testing}\label{sec:univ-to-extreme}
In this section, we will prove the following lemma:
\begin{lemma}
  \label{lem:univ-to-extreme}
  There exists an absolute constant $c > 0$ such that the following holds.
  If there is a $cn$-sample algorithm for testing identity to any extreme product distribution $Q_\mathrm{prod}$ over $\pmo^d$, then there is an $n$-sample algorithm for testing identity to any distribution $Q_\mathrm{univ}$ over $[d]$ such that $\lVert Q_\mathrm{univ} \rVert = O(1/d)$. Moreover, if the former algorithm is $\eps$-DP, then so is the latter.
\end{lemma}
This will be established via a sequence of reductions: from univariate testing to Poissonized univariate testing (Lemma~\ref{lemma:reduction:poisson}), to extreme product testing (Lemma~\ref{lemma:reduction:extreme}); before one final observation letting us get rid of one assumption stemming from the last reduction (Remark~\ref{remark:pesky:assumption}).
We describe these reductions in the following subsections.
\subsubsection{Univariate to Poissonized Univariate}
The first reduction we perform is from having a dataset of fixed size, to drawing a dataset of variable size. 
This technique is known as ``Poissonization,'' and is folklore in the distribution testing literature (see, e.g.,~\cite[Appendix D.3]{Canonne15a}).
We include the argument here for completeness.
Drawing a random number of samples ($\Poi(n)$), rather than a fixed budget ($n$), has the advantage that the frequency of each symbol $i \in [d]$ will be distributed as $\Poi(n\cdot P_i)$, independently.
At the same time, with high probability, $\Poi(n) \leq O(n)$, so one can simulate drawing $\Poi(n)$ samples with a fixed budget, at the cost of a constant factor overhead.

\begin{lemma}
  \label{lemma:reduction:poisson}
  If there is an algorithm which draws $\Poi(n)$ samples and tests identity to a distribution $Q_\mathrm{univ}$ over $[d]$ with probability of failure at most $\delta$ (and distance parameter $\alpha$), then there is a $2n$-sample algorithm for testing identity to $Q_\mathrm{univ}$ with probability of failure at most $\delta+1/10$ (and distance parameter $\alpha$).
\end{lemma}
\begin{proof}
  Consider the algorithm which draws $N \sim \Poi(n)$ samples.
  With high probability, this will draw at most $2n$ samples.
  More precisely (see, e.g.,~\cite{Canonne17}), we have that $\Pr\left[N \geq 2n\right] \leq \exp\left(-\frac{n^2}{2n}\right) = \exp\left(-\frac{n}{2}\right)$.
  For $n$ larger than some absolute constant, this is less than an arbitrarily small constant.

  With this in mind, we describe the $2n$-sample algorithm.
  It begins by drawing its own $N \sim \Poi(n)$.
  If $N > 2n$, it outputs arbitrarily.
  Otherwise, it runs the algorithm which takes $\Poi(n)$ samples, on a set of $N$ samples (drawn uniformly at random from its set of $2n$ samples).
  Correctness follows from correctness of the Poissonized tester: the only change is that the probability of failure increases by $\Pr[N > 2n]$, which as argued before, will be less than an arbitrarily small constant (e.g., $1/10$) for $n$ greater than some absolute constant.
\end{proof}
\subsubsection{Poissonized Univariate to Extreme Product}
For convenience and ease of notation, we hereafter consider distributions over $\{0,1\}^d$ instead of the equivalent choice $\{\pm 1\}^d$. 
\begin{lemma}
  \label{lemma:reduction:extreme}
  Fix any constant $C\geq 1$. Suppose $n \geq c\log d$, where $c>0$ is a suitably large absolute constant. If there is an algorithm which takes $n$ samples and tests identity to a distribution $Q_\mathrm{prod}$ on $\{0,1\}^d$  with probability of failure at most $\delta$ (and distance parameter $\alpha' := e^C\alpha/2$), then there is an algorithm which draws $\Poi(2n)$ samples and tests identity to any univariate distribution such that $\lVert Q_\mathrm{univ} \rVert_\infty\leq C/d$,  with probability of failure at most $\delta+1/10$ (and distance parameter $\alpha$). Moreover, if the former algorithm is $\eps$-DP, then so is the latter.
\end{lemma}
In addition, as will be clear from the proof below, the reduction then guarantees that the product distribution $Q_\mathrm{prod}$ obtained from such $Q_\mathrm{univ}$ will satisfy $\lVert \EE[Q_\mathrm{prod}] \rVert_\infty \leq C/d$, i.e., is a $C$-extreme product distribution. Therefore, the above reduction holds even when only requiring a testing algorithm for identity to extreme product distributions.
\begin{proof}
Consider the following process: taking $\Poi(2n)$ samples from a univariate distribution $P=(p_1,\dots,p_d)$ over $[d]$, one obtains independent random variables $N_1,\dots,N_d$ such that $N_i \sim \Poi(2n p_i)$. Now, draw $M_1,\dots,M_d\sim\Poi(2n)$ (mutually independent, and independent of $(N_i)_{i \in [d]}$), and set $N'_i \gets N_i + M_i$ for all $i \in [d]$.

Clearly, the random variables $N'_1,\dots, N'_d$ are mutually independent, and further $N'_i\mid M_i$ is distributed as $\Bin(M_i,p_i/(1+p_i))$, by properties of Poisson random variables and since $\frac{\EE[N_i]}{\EE[N_i]+\EE[M_i]} = \frac{2np_i}{2np_i+2n} = \frac{p_i}{1+p_i}$.

Let now $M \gets \min_{i\in[d]} M_i$. For each $i\in[d]$, define a random binary vector $V^{(i)}\in \{0,1\}^{M_i}$ obtained by choosing uniformly at random a subset of $[M_i]$ of size $N_i$ and filling the corresponding coordinates with $1$, setting the $M_i-N_i$ coordinates to $0$; then defining $T_i$ as
\[
  T_i \gets \sum_{j=1}^M V^{(i)}_{j}
\]
that is, the number of coordinates set to $1$ among the first $M$. We then have $T_i \sim \Bin(M,p_i/(1+p_i))$; further, conditioned on $M$, all the $T_i$'s are independent.

The outcome of this process is then a $(d+1)$-tuple $(M,T_1,\dots,T_d)$; we convert this into $M$ i.i.d. samples from the product distribution $P_\mathrm{prod}$ on $\{0,1\}^d$ such that $\EE[(P_\mathrm{prod})_i] = \frac{p_i}{1+p_i}$ in the natural way, by building an $M$-by-$d$ binary matrix with exactly $T_i$ ones in the $i$-th column, before permuting independently and uniformly at random each column.

To conclude, we claim that as long as $n \geq c\log d$ (for some absolute constant $c>0$), then $\pr{}{M \geq n} \geq \frac{9}{10}$. This follows from concentration of Poisson r.v.'s (again, see e.g.,~\cite{Canonne17}) and a union bound over the $d$ i.i.d. random variables $M_1,\dots,M_d$.

The tester for $Q_\mathrm{univ}$ then proceeds as follows: given $\Poi(2n)$ samples from an unknown $P$ over $[d]$, it applies the above procedure and, with probability at least $9/10$, succeeds in producing $n$ i.i.d. samples from the distribution $P_\mathrm{prod}$ over $\{0,1\}^d$ such that $\EE[(P_\mathrm{prod})_i] = \frac{p_i}{1+p_i}$; it then runs the identity tester for $Q_\mathrm{univ}$ on those $n$ samples. (When the reduction procedure fails, i.e., when $M < n$, then the tester outputs arbitrarily.) Correctness then follows from the below observations:
\begin{itemize}
  \item If $P=Q_\mathrm{univ}$, then $P_\mathrm{prod}=Q_\mathrm{prod}$;
  \item If $\lVert P-Q_\mathrm{univ}\rVert_1 \geq \alpha$, then $\lVert P_\mathrm{prod}-Q_\mathrm{prod}\lVert_1 \geq \alpha$.
\end{itemize}
The first item is obvious; we hereafter establish the second. %
Denoting by $\mathbf{e}_j$ the $j$-th standard vector of the canonical basis of $\{0,1\}^d$, and by $\mathbf{0}$ the all-zero vector, we have
\begin{align*}
\lVert P_\mathrm{prod}-Q_\mathrm{prod}\lVert_1
&\geq \lvert P_\mathrm{prod}(\mathbf{0})-Q_\mathrm{prod}(\mathbf{0}) \rvert 
+ \sum_{j=1}^d \lvert P_\mathrm{prod}(\mathbf{e}_j)-Q_\mathrm{prod}(\mathbf{e}_j) \rvert \\
&= \left\lvert \prod_{i=1}^d (1-p'_i)-\prod_{i=1}^d (1-q'_i) \right\rvert
+ \sum_{j=1}^d \left\lvert \frac{p'_j}{1-p'_j}\prod_{i=1}^d (1-p'_i)-\frac{q'_j}{1-q'_j}\prod_{i=1}^d (1-q'_i) \right\rvert \\
&= \left\lvert \prod_{i=1}^d \frac{1}{1+p_j}-\prod_{i=1}^d \frac{1}{1+q_j} \right\rvert
+ \sum_{j=1}^d \left\lvert p_j\prod_{i=1}^d \frac{1}{1+p_j} -q_j\prod_{i=1}^d \frac{1}{1+q_j} \right\rvert 
\end{align*}
where we wrote $p'_i \eqdef \frac{p_i}{1+p_i}$ (and similarly for $q'_i$) for conciseness. If the first term, $\lvert P_\mathrm{prod}(\mathbf{0})-Q_\mathrm{prod}(\mathbf{0}) \rvert$, is greater than $e^{-C}\alpha/2$, then we are done. Otherwise, we have, by the triangle inequality,
\begin{align*}
\sum_{j=1}^d \left\lvert p_j\prod_{i=1}^d \frac{1}{1+p_j} -q_j\prod_{i=1}^d \frac{1}{1+q_j} \right\rvert 
&\geq \prod_{i=1}^d \frac{1}{1+q_j} \sum_{j=1}^d \left\lvert p_j -q_j \right\rvert 
- \left\lvert \prod_{i=1}^d \frac{1}{1+p_j}-\prod_{i=1}^d \frac{1}{1+q_j} \right\rvert \sum_{j=1}^d p_j\\
&\geq \prod_{i=1}^d \frac{1}{1+q_j} \lVert P-Q_\mathrm{univ}\rVert_1 - e^{-C} \frac{\alpha}{2}
\end{align*}
using that $\sum_{j=1}^d p_j=1$. However, since we had assumed $\lVert Q_\mathrm{univ} \rVert_\infty\leq C/d$, we can bound
$
\prod_{i=1}^d \frac{1}{1+q_j} \geq 
\frac{1}{(1+C/d)^d}  \geq e^{-C} 
$, and thus overall $\lVert P_\mathrm{prod}-Q_\mathrm{prod}\lVert_1\geq \frac{e^{-C}}{2} \alpha$ in this case too.
\end{proof}

\begin{rem}
  \label{remark:pesky:assumption}
The reader may observe that the ``if and only if'' statement of Theorem~\ref{thm:extreme-reduction} does not seem to quite follow from the combination of Lemmata~\ref{lem:extreme-to-univ} and~\ref{lem:univ-to-extreme}. Indeed, the latter only establishes that a private identity testing algorithm for extreme product distributions yields a private identity testing algorithm for univariate distributions \emph{with small infinity norm}. However, this is not an issue, as a standard reduction due to Goldreich~\cite{Goldreich16} shows that any univariate uniformity testing algorithm implies a general univariate identity testing algorithm with only a constant loss in parameters; and this reduction preserves differential privacy. Theorem~\ref{thm:extreme-reduction} thus follows from combining this last reduction with Lemma~\ref{lem:univ-to-extreme}.
\end{rem}

\section{Lower Bounds}
\label{sec:lbs}

In this section, we discuss information-theoretic lower bounds for differentially private testing in high dimensions.
First, we restate the lower bound implied by Corollary~\ref{cor:extreme-reduction}.

\begin{theorem}
  Any $\eps$-differentially private algorithm (where $\ve = \tilde \Omega(1/d)$)
  which draws $n$ samples from an unknown product distribution $P \in \{\pm 1\}^d$ and, with probability at least $2/3$, distinguishes between the cases $P = Q$ and $\|P - Q\|_1 \geq \alpha$ where $Q$ is some given product distribution over $\{\pm 1\}^d$, requires 
  $$n = \Omega\left(\frac{d^{1/2}}{\alpha^2}  + \frac{d^{1/2}}{\alpha \ve^{1/2}} + \frac{d^{1/3}}{\alpha^{4/3}\ve^{2/3}} + \frac{1}{\alpha\eps}\right),$$
\end{theorem}

Note that this matches the upper bound of Theorem~\ref{thm:main_ineff} up to logarithmic factors.
However, the $Q$'s considered in this lower bound construction are \emph{extreme product distributions} (from Section~\ref{sec:extreme}), and it leaves open the possibility that there may exist better algorithms for the case where $Q$ is the uniform distribution.

Focusing on this case, we state the following lower bound for uniformity testing of product distributions, which also holds for Gaussian mean testing.
The first term in the lower bound is the non-private sample complexity of this problem~\cite{CanonneDKS17,DaskalakisDK18}, and the second term is the sample complexity of testing uniformity of a Bernoulli distribution (see, e.g.,~\cite{AcharyaSZ18}).
\begin{theorem}
  Any $\eps$-differentially private algorithm which draws $n$ samples from an unknown product distribution $P \in \{\pm 1\}^d$ and, with probability at least $2/3$, distinguishes between the cases $P = Q$ and $\|P - Q\|_1 \geq \alpha$ where $Q$ is the uniform distribution over $\{\pm 1\}^d$, requires 
  $$n = \Omega\left(\frac{d^{1/2}}{\alpha^2}  + \frac{1}{\alpha\eps}\right).$$
\end{theorem}

Lower bounds for multivariate distribution testing appear to be much more challenging to prove than in the univariate case, due to the necessity of maintaining independence of marginals in the construction of a coupling, which is the standard technique for proving lower bounds for private distribution testing.
Indeed, our lower bound in Corollary~\ref{cor:extreme-reduction} was proved by showing an equivalence to the univariate case, and using lower bounds from this setting.
Nonetheless, we conjecture that the same lower bound in Corollary~\ref{cor:extreme-reduction} holds for the uniform distribution. In particular, in the univariate case testing uniformity is known to be the ``hardest'' problem within identity testing~\cite{ValiantV14, Goldreich16}. If this could be shown for testing product distributions (while preserving privacy), this would imply the desired lower bound, and match our upper bound in Theorem~\ref{thm:main_ineff}.

\ifnum\anonymous=0
  \section*{Acknowledgments}

  Part of this work was done while GK, AM, JU, and LZ  were visiting the Simons Institute for the Theory of Computing, and while CC was a Motwani Postdoctoral Fellow at Stanford Universitym then a Goldstine Postdoctoral Fellow at IBM Research, Almaden.  GK was supported as a Microsoft Research Fellow, as part of the Simons-Berkeley Research Fellowship program. Some of this work was completed while visiting Microsoft Research, Redmond. AM was supported by NSF grant CCF-1763786, a Sloan Foundation Research Award, and a postdoctoral fellowship from BU's Hariri Institute for Computing.  JU and LZ were supported by NSF grants CCF-1718088, CCF-1750640, and CNS-1816028. CC was supported by a Goldstine Fellowship. The authors would like to thank Jayadev Acharya and Himanshu Tyagi for discussions that contributed to Theorem~\ref{theo:gaussian:product:reduction}.  
\fi
\bibliography{biblio}

\newcommand{\etalchar}[1]{$^{#1}$}
\begin{thebibliography}{SGHG{\etalchar{+}}19}

\bibitem[AAK{\etalchar{+}}07]{AlonAKMRX07}
Noga Alon, Alexandr Andoni, Tali Kaufman, Kevin Matulef, Ronitt Rubinfeld, and
  Ning Xie.
\newblock Testing k-wise and almost k-wise independence.
\newblock In {\em Proceedings of the 39th Annual ACM Symposium on the Theory of
  Computing}, STOC '07, pages 496--505, New York, NY, USA, 2007. ACM.

\bibitem[ABDK18]{AcharyaBDK18}
Jayadev Acharya, Arnab Bhattacharyya, Constantinos Daskalakis, and Saravanan
  Kandasamy.
\newblock Learning and testing causal models with interventions.
\newblock In {\em Advances in Neural Information Processing Systems 31},
  NeurIPS '18. Curran Associates, Inc., 2018.

\bibitem[ACFT19]{AcharyaCFT19}
Jayadev Acharya, Cl{\'e}ment~L. Canonne, Cody Freitag, and Himanshu Tyagi.
\newblock Test without trust: Optimal locally private distribution testing.
\newblock In {\em Proceedings of the 22nd International Conference on
  Artificial Intelligence and Statistics}, AISTATS '19, pages 2067--2076. JMLR,
  Inc., 2019.

\bibitem[ADK15]{AcharyaDK15}
Jayadev Acharya, Constantinos Daskalakis, and Gautam Kamath.
\newblock Optimal testing for properties of distributions.
\newblock In {\em Advances in Neural Information Processing Systems 28}, NIPS
  '15, pages 3577--3598. Curran Associates, Inc., 2015.

\bibitem[ADR18]{AliakbarpourDR18}
Maryam Aliakbarpour, Ilias Diakonikolas, and Ronitt Rubinfeld.
\newblock Differentially private identity and closeness testing of discrete
  distributions.
\newblock In {\em Proceedings of the 35th International Conference on Machine
  Learning}, ICML '18, pages 169--178. JMLR, Inc., 2018.

\bibitem[AKSZ18]{AcharyaKSZ18}
Jayadev Acharya, Gautam Kamath, Ziteng Sun, and Huanyu Zhang.
\newblock Inspectre: Privately estimating the unseen.
\newblock In {\em Proceedings of the 35th International Conference on Machine
  Learning}, ICML '18, pages 30--39. JMLR, Inc., 2018.

\bibitem[AS18]{AwanS18}
Jordan Awan and Aleksandra Slavkovi{\'c}.
\newblock Differentially private uniformly most powerful tests for binomial
  data.
\newblock In {\em Advances in Neural Information Processing Systems 31},
  NeurIPS '18, pages 4208--4218. Curran Associates, Inc., 2018.

\bibitem[ASZ18]{AcharyaSZ18}
Jayadev Acharya, Ziteng Sun, and Huanyu Zhang.
\newblock Differentially private testing of identity and closeness of discrete
  distributions.
\newblock In {\em Advances in Neural Information Processing Systems 31},
  NeurIPS '18, pages 6878--6891. Curran Associates, Inc., 2018.

\bibitem[ASZ19]{AcharyaSZ19}
Jayadev Acharya, Ziteng Sun, and Huanyu Zhang.
\newblock Hadamard response: Estimating distributions privately, efficiently,
  and with little communication.
\newblock In {\em Proceedings of the 22nd International Conference on
  Artificial Intelligence and Statistics}, AISTATS '19, pages 1120--1129. JMLR,
  Inc., 2019.

\bibitem[BBC{\etalchar{+}}19]{BezakovaBCSV19}
Ivona Bezakova, Antonio Blanca, Zongchen Chen, Daniel {\v{S}}tefankovi{\v{c}},
  and Eric Vigoda.
\newblock Lower bounds for testing graphical models: Colorings and
  antiferromagnetic {I}sing models.
\newblock In {\em Proceedings of the 32nd Annual Conference on Learning
  Theory}, COLT '19, pages 283--298, 2019.

\bibitem[BBDS13]{BlockiBDS13}
Jeremiah Blocki, Avrim Blum, Anupam Datta, and Or~Sheffet.
\newblock Differentially private data analysis of social networks via
  restricted sensitivity.
\newblock In {\em Proceedings of the 4th Conference on Innovations in
  Theoretical Computer Science}, ITCS '13, pages 87--96, New York, NY, USA,
  2013. ACM.

\bibitem[BCG17]{BlaisCG17}
Eric Blais, Cl{\'e}ment~L. Canonne, and Tom Gur.
\newblock Distribution testing lower bounds via reductions from communication
  complexity.
\newblock In {\em Proceedings of the 32nd Computational Complexity Conference},
  CCC '17, pages 28:1--28:40, Dagstuhl, Germany, 2017. Schloss
  Dagstuhl--Leibniz-Zentrum fuer Informatik.

\bibitem[BCSZ18a]{BorgsCSZ18b}
Christian Borgs, Jennifer Chayes, Adam Smith, and Ilias Zadik.
\newblock Private algorithms can always be extended.
\newblock {\em arXiv preprint arXiv:1810.12518}, 2018.

\bibitem[BCSZ18b]{BorgsCSZ18a}
Christian Borgs, Jennifer Chayes, Adam Smith, and Ilias Zadik.
\newblock Revealing network structure, confidentially: Improved rates for
  node-private graphon estimation.
\newblock In {\em Proceedings of the 59th Annual IEEE Symposium on Foundations
  of Computer Science}, FOCS '18, pages 533--543, Washington, DC, USA, 2018.
  IEEE Computer Society.

\bibitem[BDMN05]{BlumDMN05}
Avrim Blum, Cynthia Dwork, Frank McSherry, and Kobbi Nissim.
\newblock Practical privacy: The {SuLQ} framework.
\newblock In {\em Proceedings of the 24th ACM SIGMOD-SIGACT-SIGART Symposium on
  Principles of Database Systems}, PODS '05, pages 128--138, New York, NY, USA,
  2005. ACM.

\bibitem[BFF{\etalchar{+}}01]{BatuFFKRW01}
Tu\u{g}kan Batu, Eldar Fischer, Lance Fortnow, Ravi Kumar, Ronitt Rubinfeld,
  and Patrick White.
\newblock Testing random variables for independence and identity.
\newblock In {\em Proceedings of the 42nd Annual IEEE Symposium on Foundations
  of Computer Science}, FOCS '01, pages 442--451, Washington, DC, USA, 2001.
  IEEE Computer Society.

\bibitem[BFR{\etalchar{+}}00]{BatuFRSW00}
Tu\u{g}kan Batu, Lance Fortnow, Ronitt Rubinfeld, Warren~D. Smith, and Patrick
  White.
\newblock Testing that distributions are close.
\newblock In {\em Proceedings of the 41st Annual IEEE Symposium on Foundations
  of Computer Science}, FOCS '00, pages 259--269, Washington, DC, USA, 2000.
  IEEE Computer Society.

\bibitem[BKR04]{BatuKR04}
Tu\u{g}kan Batu, Ravi Kumar, and Ronitt Rubinfeld.
\newblock Sublinear algorithms for testing monotone and unimodal distributions.
\newblock In {\em Proceedings of the 36th Annual ACM Symposium on the Theory of
  Computing}, STOC '04, New York, NY, USA, 2004. ACM.

\bibitem[BN14]{BrennerN14}
Hai Brenner and Kobbi Nissim.
\newblock Impossibility of differentially private universally optimal
  mechanisms.
\newblock {\em SIAM Journal on Computing}, 43(5):1513--1540, 2014.

\bibitem[BNS{\etalchar{+}}16]{BassilyNSSSU16}
Raef Bassily, Kobbi Nissim, Adam Smith, Thomas Steinke, Uri Stemmer, and
  Jonathan Ullman.
\newblock Algorithmic stability for adaptive data analysis.
\newblock In {\em Proceedings of the 48th Annual ACM Symposium on the Theory of
  Computing}, STOC '16, pages 1046--1059, New York, NY, USA, 2016. ACM.

\bibitem[BNSV15]{BunNSV15}
Mark Bun, Kobbi Nissim, Uri Stemmer, and Salil Vadhan.
\newblock Differentially private release and learning of threshold functions.
\newblock In {\em Proceedings of the 56th Annual IEEE Symposium on Foundations
  of Computer Science}, FOCS '15, pages 634--649, Washington, DC, USA, 2015.
  IEEE Computer Society.

\bibitem[BS16]{BunS16}
Mark Bun and Thomas Steinke.
\newblock Concentrated differential privacy: Simplifications, extensions, and
  lower bounds.
\newblock In {\em Proceedings of the 14th Conference on Theory of
  Cryptography}, TCC '16-B, pages 635--658, Berlin, Heidelberg, 2016. Springer.

\bibitem[BUV14]{BunUV14}
Mark Bun, Jonathan Ullman, and Salil Vadhan.
\newblock Fingerprinting codes and the price of approximate differential
  privacy.
\newblock In {\em Proceedings of the 46th Annual ACM Symposium on the Theory of
  Computing}, STOC '14, pages 1--10, New York, NY, USA, 2014. ACM.

\bibitem[BV15]{BhattacharyaV15}
Bhaswar Bhattacharya and Gregory Valiant.
\newblock Testing closeness with unequal sized samples.
\newblock In {\em Advances in Neural Information Processing Systems 28}, NIPS
  '15, pages 2611--2619. Curran Associates, Inc., 2015.

\bibitem[BW18]{BalakrishnanW18}
Sivaraman Balakrishnan and Larry Wasserman.
\newblock Hypothesis testing for high-dimensional multinomials: A selective
  review.
\newblock {\em The Annals of Applied Statistics}, 12(2):727--749, 2018.

\bibitem[Can15]{Canonne15a}
Cl{\'e}ment~L. Canonne.
\newblock A survey on distribution testing: Your data is big. but is it blue?
\newblock {\em Electronic Colloquium on Computational Complexity (ECCC)},
  22(63), 2015.

\bibitem[Can17]{Canonne17}
Cl{\'e}ment~L. Canonne.
\newblock A short note on {P}oisson tail bounds.
\newblock
  \url{http://www.cs.columbia.edu/~ccanonne/files/misc/2017-poissonconcentration.pdf},
  2017.

\bibitem[CBRG18]{CampbellBRG18}
Zachary Campbell, Andrew Bray, Anna Ritz, and Adam Groce.
\newblock Differentially private {ANOVA} testing.
\newblock In {\em Proceedings of the 2018 International Conference on Data
  Intelligence and Security}, ICDIS '18, pages 281--285, Washington, DC, USA,
  2018. IEEE Computer Society.

\bibitem[CD20]{CummingsD20}
Rachel Cummings and David Durfee.
\newblock Individual sensitivity preprocessing for data privacy.
\newblock In {\em Proceedings of the 31st Annual ACM-SIAM Symposium on Discrete
  Algorithms}, SODA '20, Philadelphia, PA, USA, 2020. SIAM.

\bibitem[CDGR16]{CanonneDGR16}
Cl{\'e}ment~L. Canonne, Ilias Diakonikolas, Themis Gouleakis, and Ronitt
  Rubinfeld.
\newblock Testing shape restrictions of discrete distributions.
\newblock In {\em Proceedings of the 33rd Symposium on Theoretical Aspects of
  Computer Science}, STACS '16, pages 25:1--25:14, Dagstuhl, Germany, 2016.
  Schloss Dagstuhl--Leibniz-Zentrum fuer Informatik.

\bibitem[CDK17]{CaiDK17}
Bryan Cai, Constantinos Daskalakis, and Gautam Kamath.
\newblock Priv'it: Private and sample efficient identity testing.
\newblock In {\em Proceedings of the 34th International Conference on Machine
  Learning}, ICML '17, pages 635--644. JMLR, Inc., 2017.

\bibitem[CDKS17]{CanonneDKS17}
Cl{\'e}ment~L. Canonne, Ilias Diakonikolas, Daniel~M. Kane, and Alistair
  Stewart.
\newblock Testing {B}ayesian networks.
\newblock In {\em Proceedings of the 30th Annual Conference on Learning
  Theory}, COLT '17, pages 370--448, 2017.

\bibitem[CDKS18]{CanonneDKS18}
Cl{\'e}ment~L. Canonne, Ilias Diakonikolas, Daniel~M. Kane, and Alistair
  Stewart.
\newblock Testing conditional independence of discrete distributions.
\newblock In {\em Proceedings of the 50th Annual ACM Symposium on the Theory of
  Computing}, STOC '18, pages 735--748, New York, NY, USA, 2018. ACM.

\bibitem[CDVV14]{ChanDVV14}
Siu~On Chan, Ilias Diakonikolas, Gregory Valiant, and Paul Valiant.
\newblock Optimal algorithms for testing closeness of discrete distributions.
\newblock In {\em Proceedings of the 25th Annual ACM-SIAM Symposium on Discrete
  Algorithms}, SODA '14, pages 1193--1203, Philadelphia, PA, USA, 2014. SIAM.

\bibitem[CKM{\etalchar{+}}18]{CummingsKMTZ18}
Rachel Cummings, Sara Krehbiel, Yajun Mei, Rui Tuo, and Wanrong Zhang.
\newblock Differentially private change-point detection.
\newblock In {\em Advances in Neural Information Processing Systems 31},
  NeurIPS '18. Curran Associates, Inc., 2018.

\bibitem[CKM{\etalchar{+}}19]{CanonneKMSU19}
Cl{\'e}ment~L. Canonne, Gautam Kamath, Audra McMillan, Adam Smith, and Jonathan
  Ullman.
\newblock The structure of optimal private tests for simple hypotheses.
\newblock In {\em Proceedings of the 51st Annual ACM Symposium on the Theory of
  Computing}, STOC '19, New York, NY, USA, 2019. ACM.

\bibitem[CKS{\etalchar{+}}19]{CouchKSBG19}
Simon Couch, Zeki Kazan, Kaiyan Shi, Andrew Bray, and Adam Groce.
\newblock Differentially private nonparametric hypothesis testing.
\newblock In {\em Proceedings of the 2019 ACM Conference on Computer and
  Communications Security}, CCS '19, New York, NY, USA, 2019. ACM.

\bibitem[CWZ19]{CaiWZ19}
T.~Tony Cai, Yichen Wang, and Linjun Zhang.
\newblock The cost of privacy: Optimal rates of convergence for parameter
  estimation with differential privacy.
\newblock {\em arXiv preprint arXiv:1902.04495}, 2019.

\bibitem[DDK18]{DaskalakisDK18}
Constantinos Daskalakis, Nishanth Dikkala, and Gautam Kamath.
\newblock Testing {I}sing models.
\newblock In {\em Proceedings of the 29th Annual ACM-SIAM Symposium on Discrete
  Algorithms}, SODA '18, pages 1989--2007, Philadelphia, PA, USA, 2018. SIAM.

\bibitem[DFH{\etalchar{+}}15]{DworkFHPRR15}
Cynthia Dwork, Vitaly Feldman, Moritz Hardt, Toniann Pitassi, Omer Reingold,
  and Aaron Roth.
\newblock The reusable holdout: Preserving validity in adaptive data analysis.
\newblock {\em Science}, 349(6248):636--638, 2015.

\bibitem[DHS15]{DiakonikolasHS15}
Ilias Diakonikolas, Moritz Hardt, and Ludwig Schmidt.
\newblock Differentially private learning of structured discrete distributions.
\newblock In {\em Advances in Neural Information Processing Systems 28}, NIPS
  '15, pages 2566--2574. Curran Associates, Inc., 2015.

\bibitem[{Dif}17]{AppleDP17}
{Differential Privacy Team, Apple}.
\newblock Learning with privacy at scale.
\newblock
  \url{https://machinelearning.apple.com/docs/learning-with-privacy-at-scale/appledifferentialprivacysystem.pdf},
  December 2017.

\bibitem[DJW13]{DuchiJW13}
John~C. Duchi, Michael~I. Jordan, and Martin~J. Wainwright.
\newblock Local privacy and statistical minimax rates.
\newblock In {\em Proceedings of the 54th Annual IEEE Symposium on Foundations
  of Computer Science}, FOCS '13, pages 429--438, Washington, DC, USA, 2013.
  IEEE Computer Society.

\bibitem[DK16]{DiakonikolasK16}
Ilias Diakonikolas and Daniel~M. Kane.
\newblock A new approach for testing properties of discrete distributions.
\newblock In {\em Proceedings of the 57th Annual IEEE Symposium on Foundations
  of Computer Science}, FOCS '16, pages 685--694, Washington, DC, USA, 2016.
  IEEE Computer Society.

\bibitem[DKM{\etalchar{+}}06]{DworkKMMN06}
Cynthia Dwork, Krishnaram Kenthapadi, Frank McSherry, Ilya Mironov, and Moni
  Naor.
\newblock Our data, ourselves: Privacy via distributed noise generation.
\newblock In {\em Proceedings of the 24th Annual International Conference on
  the Theory and Applications of Cryptographic Techniques}, EUROCRYPT '06,
  pages 486--503, Berlin, Heidelberg, 2006. Springer.

\bibitem[DKN15]{DiakonikolasKN15a}
Ilias Diakonikolas, Daniel~M. Kane, and Vladimir Nikishkin.
\newblock Testing identity of structured distributions.
\newblock In {\em Proceedings of the 26th Annual ACM-SIAM Symposium on Discrete
  Algorithms}, SODA '15, pages 1841--1854, Philadelphia, PA, USA, 2015. SIAM.

\bibitem[DKW18]{DaskalakisKW18}
Constantinos Daskalakis, Gautam Kamath, and John Wright.
\newblock Which distribution distances are sublinearly testable?
\newblock In {\em Proceedings of the 29th Annual ACM-SIAM Symposium on Discrete
  Algorithms}, SODA '18, pages 2747--2764, Philadelphia, PA, USA, 2018. SIAM.

\bibitem[DLS{\etalchar{+}}17]{DajaniLSKRMGDGKKLSSVA17}
Aref~N. Dajani, Amy~D. Lauger, Phyllis~E. Singer, Daniel Kifer, Jerome~P.
  Reiter, Ashwin Machanavajjhala, Simson~L. Garfinkel, Scot~A. Dahl, Matthew
  Graham, Vishesh Karwa, Hang Kim, Philip Lelerc, Ian~M. Schmutte, William~N.
  Sexton, Lars Vilhuber, and John~M. Abowd.
\newblock The modernization of statistical disclosure limitation at the {U.S.}
  census bureau, 2017.
\newblock Presented at the September 2017 meeting of the Census Scientific
  Advisory Committee.

\bibitem[DMNS06]{DworkMNS06}
Cynthia Dwork, Frank McSherry, Kobbi Nissim, and Adam Smith.
\newblock Calibrating noise to sensitivity in private data analysis.
\newblock In {\em Proceedings of the 3rd Conference on Theory of Cryptography},
  TCC '06, pages 265--284, Berlin, Heidelberg, 2006. Springer.

\bibitem[DMR18]{DevroyeMR18b}
Luc Devroye, Abbas Mehrabian, and Tommy Reddad.
\newblock The total variation distance between high-dimensional {G}aussians.
\newblock {\em arXiv preprint arXiv:1810.08693}, 2018.

\bibitem[DP17]{DaskalakisP17}
Constantinos Daskalakis and Qinxuan Pan.
\newblock Square {H}ellinger subadditivity for {B}ayesian networks and its
  applications to identity testing.
\newblock In {\em Proceedings of the 30th Annual Conference on Learning
  Theory}, COLT '17, pages 697--703, 2017.

\bibitem[DR14]{DworkR14}
Cynthia Dwork and Aaron Roth.
\newblock The algorithmic foundations of differential privacy.
\newblock {\em Foundations and Trends{\textregistered} in Machine Learning},
  9(3--4):211--407, 2014.

\bibitem[DR16]{DworkR16}
Cynthia Dwork and Guy~N. Rothblum.
\newblock Concentrated differential privacy.
\newblock {\em arXiv preprint arXiv:1603.01887}, 2016.

\bibitem[DR18]{DuchiR18}
John~C. Duchi and Feng Ruan.
\newblock The right complexity measure in locally private estimation: It is not
  the fisher information.
\newblock {\em arXiv preprint arXiv:1806.05756}, 2018.

\bibitem[EPK14]{ErlingssonPK14}
{\'U}lfar Erlingsson, Vasyl Pihur, and Aleksandra Korolova.
\newblock {RAPPOR}: Randomized aggregatable privacy-preserving ordinal
  response.
\newblock In {\em Proceedings of the 2014 ACM Conference on Computer and
  Communications Security}, CCS '14, pages 1054--1067, New York, NY, USA, 2014.
  ACM.

\bibitem[FKT20]{FeldmanKT20}
Vitaly Feldman, Tomer Koren, and Kunal Talwar.
\newblock Private stochastic convex optimization: Optimal rates in linear time.
\newblock In {\em Proceedings of the 52nd Annual ACM Symposium on the Theory of
  Computing}, STOC '20, New York, NY, USA, 2020. ACM.

\bibitem[GGR96]{GoldreichGR96}
Oded Goldreich, Shafi Goldwasser, and Dana Ron.
\newblock Property testing and its connection to learning and approximation.
\newblock In {\em Proceedings of the 37th Annual IEEE Symposium on Foundations
  of Computer Science}, FOCS '96, pages 339--348, Washington, DC, USA, 1996.
  IEEE Computer Society.

\bibitem[GLP18]{GheissariLP18}
Reza Gheissari, Eyal Lubetzky, and Yuval Peres.
\newblock Concentration inequalities for polynomials of contracting {I}sing
  models.
\newblock {\em Electronic Communications in Probability}, 23(76):1--12, 2018.

\bibitem[GLRV16]{GaboardiLRV16}
Marco Gaboardi, Hyun{-}Woo Lim, Ryan~M. Rogers, and Salil~P. Vadhan.
\newblock Differentially private chi-squared hypothesis testing: Goodness of
  fit and independence testing.
\newblock In {\em Proceedings of the 33rd International Conference on Machine
  Learning}, ICML '16, pages 1395--1403. JMLR, Inc., 2016.

\bibitem[Gol16]{Goldreich16}
Oded Goldreich.
\newblock The uniform distribution is complete with respect to testing identity
  to a fixed distribution.
\newblock {\em Electronic Colloquium on Computational Complexity (ECCC)},
  23(15), 2016.

\bibitem[Gol17]{Goldreich17}
Oded Goldreich.
\newblock {\em Introduction to Property Testing}.
\newblock Cambridge University Press, 2017.

\bibitem[GR00]{GoldreichR00}
Oded Goldreich and Dana Ron.
\newblock On testing expansion in bounded-degree graphs.
\newblock {\em Electronic Colloquium on Computational Complexity (ECCC)},
  7(20), 2000.

\bibitem[GR18]{GaboardiR18}
Marco Gaboardi and Ryan Rogers.
\newblock Local private hypothesis testing: Chi-square tests.
\newblock In {\em Proceedings of the 35th International Conference on Machine
  Learning}, ICML '18, pages 1626--1635. JMLR, Inc., 2018.

\bibitem[GRS19]{GaboardiRS19}
Marco Gaboardi, Ryan Rogers, and Or~Sheffet.
\newblock Locally private confidence intervals: Z-test and tight confidence
  intervals.
\newblock In {\em Proceedings of the 22nd International Conference on
  Artificial Intelligence and Statistics}, AISTATS '19, pages 2545--2554. JMLR,
  Inc., 2019.

\bibitem[Ing94]{Ingster94}
Yuri~Izmailovich Ingster.
\newblock Minimax detection of a signal in $\ell_p$ metrics.
\newblock {\em Journal of Mathematical Sciences}, 68(4):503--515, 1994.

\bibitem[Ing97]{Ingster97}
Yuri~Izmailovich Ingster.
\newblock Adaptive chi-square tests.
\newblock {\em Zapiski Nauchnykh Seminarov POMI}, 244:150--166, 1997.

\bibitem[IS03]{IngsterS03}
Yuri~Izmailovich Ingster and Irina~A. Suslina.
\newblock {\em Nonparametric Goodness-of-fit Testing Under {G}aussian Models},
  volume 169 of {\em Lecture Notes in Statistics}.
\newblock Springer, 2003.

\bibitem[JKMW19]{JosephKMW19}
Matthew Joseph, Janardhan Kulkarni, Jieming Mao, and Zhiwei~Steven Wu.
\newblock Locally private {G}aussian estimation.
\newblock In {\em Advances in Neural Information Processing Systems 32},
  NeurIPS '19. Curran Associates, Inc., 2019.

\bibitem[Kam18]{Kamath18}
Gautam Kamath.
\newblock {\em Modern Challenges in Distribution Testing}.
\newblock PhD thesis, Massachusetts Institute of Technology, September 2018.

\bibitem[KBR16]{KairouzBR16}
Peter Kairouz, Keith Bonawitz, and Daniel Ramage.
\newblock Discrete distribution estimation under local privacy.
\newblock In {\em Proceedings of the 33rd International Conference on Machine
  Learning}, ICML '16, pages 2436--2444. JMLR, Inc., 2016.

\bibitem[KLSU19]{KamathLSU19}
Gautam Kamath, Jerry Li, Vikrant Singhal, and Jonathan Ullman.
\newblock Privately learning high-dimensional distributions.
\newblock In {\em Proceedings of the 32nd Annual Conference on Learning
  Theory}, COLT '19, pages 1853--1902, 2019.

\bibitem[KNRS13]{KasiviswanathanNRS13}
Shiva~Prasad Kasiviswanathan, Kobbi Nissim, Sofya Raskhodnikova, and Adam
  Smith.
\newblock Analyzing graphs with node differential privacy.
\newblock In {\em Proceedings of the 10th Conference on Theory of
  Cryptography}, TCC '13, pages 457--476, Berlin, Heidelberg, 2013. Springer.

\bibitem[KR17]{KiferR17}
Daniel Kifer and Ryan~M. Rogers.
\newblock A new class of private chi-square tests.
\newblock In {\em Proceedings of the 20th International Conference on
  Artificial Intelligence and Statistics}, AISTATS '17, pages 991--1000. JMLR,
  Inc., 2017.

\bibitem[KSF17]{KakizakiSF17}
Kazuya Kakizaki, Jun Sakuma, and Kazuto Fukuchi.
\newblock Differentially private chi-squared test by unit circle mechanism.
\newblock In {\em Proceedings of the 34th International Conference on Machine
  Learning}, ICML '17, pages 1761--1770. JMLR, Inc., 2017.

\bibitem[KV18]{KarwaV18}
Vishesh Karwa and Salil Vadhan.
\newblock Finite sample differentially private confidence intervals.
\newblock In {\em Proceedings of the 9th Conference on Innovations in
  Theoretical Computer Science}, ITCS '18, pages 44:1--44:9, Dagstuhl, Germany,
  2018. Schloss Dagstuhl--Leibniz-Zentrum fuer Informatik.

\bibitem[LRR13]{LeviRR13}
Reut Levi, Dana Ron, and Ronitt Rubinfeld.
\newblock Testing properties of collections of distributions.
\newblock {\em Theory of Computing}, 9(8):295--347, 2013.

\bibitem[McS34]{McShane34}
Edward~J. McShane.
\newblock Extension of range of functions.
\newblock {\em Bulletin of the American Mathematical Society}, 40(12):837--842,
  1934.

\bibitem[Nar22]{Narayanan22}
Shyam Narayanan.
\newblock Private high-dimensional hypothesis testing.
\newblock {\em arXiv preprint arXiv:2203.01537}, 2022.

\bibitem[NRS07]{NissimRS07}
Kobbi Nissim, Sofya Raskhodnikova, and Adam Smith.
\newblock Smooth sensitivity and sampling in private data analysis.
\newblock In {\em Proceedings of the 39th Annual ACM Symposium on the Theory of
  Computing}, STOC '07, pages 75--84, New York, NY, USA, 2007. ACM.

\bibitem[Pan08]{Paninski08}
Liam Paninski.
\newblock A coincidence-based test for uniformity given very sparsely sampled
  discrete data.
\newblock {\em IEEE Transactions on Information Theory}, 54(10):4750--4755,
  2008.

\bibitem[RRST16]{RogersRST16}
Ryan Rogers, Aaron Roth, Adam Smith, and Om~Thakkar.
\newblock Max-information, differential privacy, and post-selection hypothesis
  testing.
\newblock In {\em Proceedings of the 57th Annual IEEE Symposium on Foundations
  of Computer Science}, FOCS '16, pages 487--494, Washington, DC, USA, 2016.
  IEEE Computer Society.

\bibitem[RS16]{RaskhodnikovaS16}
Sofya Raskhodnikova and Adam~D. Smith.
\newblock Lipschitz extensions for node-private graph statistics and the
  generalized exponential mechanism.
\newblock In {\em Proceedings of the 57th Annual IEEE Symposium on Foundations
  of Computer Science}, FOCS '16, pages 495--504, Washington, DC, USA, 2016.
  IEEE Computer Society.

\bibitem[Rub12]{Rubinfeld12}
Ronitt Rubinfeld.
\newblock Taming big probability distributions.
\newblock {\em XRDS}, 19(1):24--28, 2012.

\bibitem[RX14]{RubinfeldX10}
Ronitt Rubinfeld and Ning Xie.
\newblock Testing non-uniform k-wise independent distributions over product
  spaces.
\newblock In {\em Proceedings of the 37th International Colloquium on Automata,
  Languages, and Programming}, ICALP '10, pages 565--581, 2014.

\bibitem[SGHG{\etalchar{+}}19]{SwanbergGGRGB19}
Marika Swanberg, Ira Globus-Harris, Iris Griffith, Anna Ritz, Adam Groce, and
  Andrew Bray.
\newblock Improved differentially private analysis of variance.
\newblock {\em Proceedings on Privacy Enhancing Technologies}, 2019(3), 2019.

\bibitem[She18]{Sheffet18}
Or~Sheffet.
\newblock Locally private hypothesis testing.
\newblock In {\em Proceedings of the 35th International Conference on Machine
  Learning}, ICML '18, pages 4605--4614. JMLR, Inc., 2018.

\bibitem[Smi11]{Smith11}
Adam Smith.
\newblock Privacy-preserving statistical estimation with optimal convergence
  rates.
\newblock In {\em Proceedings of the 43rd Annual ACM Symposium on the Theory of
  Computing}, STOC '11, pages 813--822, New York, NY, USA, 2011. ACM.

\bibitem[SU17]{SteinkeU17a}
Thomas Steinke and Jonathan Ullman.
\newblock Between pure and approximate differential privacy.
\newblock {\em The Journal of Privacy and Confidentiality}, 7(2):3--22, 2017.

\bibitem[SU19]{SealfonU19}
Adam Sealfon and Jonathan Ullman.
\newblock Efficiently estimating {E}rdos-{R}enyi graphs with node differential
  privacy.
\newblock In {\em Advances in Neural Information Processing Systems 32},
  NeurIPS '19. Curran Associates, Inc., 2019.

\bibitem[USF13]{UhlerSF13}
Caroline Uhler, Aleksandra Slavkovi{\'c}, and Stephen~E. Fienberg.
\newblock Privacy-preserving data sharing for genome-wide association studies.
\newblock {\em The Journal of Privacy and Confidentiality}, 5(1):137--166,
  2013.

\bibitem[Val11]{Valiant11}
Paul Valiant.
\newblock Testing symmetric properties of distributions.
\newblock {\em SIAM Journal on Computing}, 40(6):1927--1968, 2011.

\bibitem[VS09]{VuS09}
Duy Vu and Aleksandra Slavkovi{\'c}.
\newblock Differential privacy for clinical trial data: Preliminary
  evaluations.
\newblock In {\em 2009 IEEE International Conference on Data Mining Workshops},
  ICDMW '09, pages 138--143. IEEE, 2009.

\bibitem[VV14]{ValiantV14}
Gregory Valiant and Paul Valiant.
\newblock An automatic inequality prover and instance optimal identity testing.
\newblock In {\em Proceedings of the 55th Annual IEEE Symposium on Foundations
  of Computer Science}, FOCS '14, pages 51--60, Washington, DC, USA, 2014. IEEE
  Computer Society.

\bibitem[WHW{\etalchar{+}}16]{WangHWNXYLQ16}
Shaowei Wang, Liusheng Huang, Pengzhan Wang, Yiwen Nie, Hongli Xu, Wei Yang,
  Xiang-Yang Li, and Chunming Qiao.
\newblock Mutual information optimally local private discrete distribution
  estimation.
\newblock {\em arXiv preprint arXiv:1607.08025}, 2016.

\bibitem[WKLK18]{WangKLK18}
Yue Wang, Daniel Kifer, Jaewoo Lee, and Vishesh Karwa.
\newblock Statistical approximating distributions under differential privacy.
\newblock {\em The Journal of Privacy and Confidentiality}, 8(1):1--33, 2018.

\bibitem[WLF16]{WangLF16}
Yu-Xiang Wang, Jing Lei, and Stephen~E. Fienberg.
\newblock A minimax theory for adaptive data analysis.
\newblock {\em arXiv preprint arXiv:1602.04287}, 2016.

\bibitem[WLK15]{WangLK15}
Yue Wang, Jaewoo Lee, and Daniel Kifer.
\newblock Revisiting differentially private hypothesis tests for categorical
  data.
\newblock {\em arXiv preprint arXiv:1511.03376}, 2015.

\bibitem[YB18]{YeB18}
Min Ye and Alexander Barg.
\newblock Optimal schemes for discrete distribution estimation under locally
  differential privacy.
\newblock {\em IEEE Transactions on Information Theory}, 64(8):5662--5676,
  2018.

\end{thebibliography}
\bibliographystyle{alpha}
\newpage
\appendix
\section{Missing Proofs of Section~\ref{sec:put}}\label{app:proofs}

We prove here the guarantees of the non-private test statistic $T$.
\begin{lemma}[Non-private Test Guarantees, Lemma~\ref{non-priv}]\label{non-priv1}
For the test statistic $T(X)=\sum_{i=1}^d \left(\sum_{j=1}^n (X_i^{(j)})^2-n\right)$ defined in \eqref{teststatistic}, the following hold:
\begin{itemize}
    \item If $P=\unif$ then $\E[T(X)]=0$ and $\var(T(X)) \leq 2n^2d$.
    \item If $\|P-\unif\|_1\ge \alpha$ then $\E[T(X)]> \frac{1}{2}n(n-1)\alpha^2$.
    \item $\var(T(X)) \leq 2n^2d + 4n \E[T(X)]$.
\end{itemize}
\end{lemma}
\begin{proof}
  Note that, for any $1\leq i\leq d$, $\sum_{j=1}^{n} X^{(j)}_1 = 2Y_i-n$ where $Y_i$ is Binomially distributed with parameters $n$ and $\frac{1+p_i}{2}$, and $Y_1,\dots,Y_d$ are independent. Therefore, we have
  \[
      T(X) = \sum_{i=1}^d \left(\left(2Y_i-n \right)^2 - n\right)
      = \sum_{i=1}^d \left(4Y_i^2 -4n Y_i + n(n-1)\right)
  \]
  and a simple computation yields $\E[T(X)] = n(n-1)\sum_{i=1}^d p_i^2 = n(n-1) \lVert p\rVert_2^2$. 
  If $P=\unif$, then this directly implies $\E[T(X)] = 0$; moreover, if $\lVert P-\unif\rVert_1 \geq \alpha$, then $\lVert p\rVert_2^2 \geq \alpha^2/2$ (by Lemma~\ref{fact:dist:product}). Thus, this establishes the claimed bounds on the expectation of the statistic.
  
  Turning to the variance, assume first that $P=\unif$, i.e., $\lVert p\rVert_2 = 0$. In this case,
  \[
    \var(T(X)) = \sum_{i=1}^d \var\left(\left(2Y_i-n \right)^2 - n\right) = \sum_{i=1}^d \E[(\left(2Y_i-n \right)^2 - n)^2]
    = 2n(n-1)d,
  \]
  expanding the square and using the expression for the first to fourth moments of a $\textrm{Bin}(n,1/2)$ random variable.
  For general $P$, one can compute explicitly this quantity, to obtain
  \begin{align*}
      \var(T(X)) 
      &= 2n(n-1) \sum_{i=1}^d (1 + (2 n-4) p_i^2 - (2 n - 3) p_i^4) %
      \leq 2n(n-1)d + 4n \E[T(X)]. \qedhere
  \end{align*}
\end{proof}

\ifnum\arxivymous=1
\section{An Efficient Private Algorithm for Gaussian Mean Testing}\label{app:gauss}

We present in full our computationally efficient private algorithm for multivariate Gaussian mean testing (Algorithm~\ref{algo:gauss}). In this setting, $P$ is a multivariate Gaussian distribution with identity covariance matrix and unknown mean, that is, $P=\cN(\mu,\idcov)$ for some unknown $\mu=(\mu_1, \ldots, \mu_d)\in \R^d$. By drawing samples from $P$, we aim to distinguish between the cases $P=\cN(\mathbf{0}, \idcov)$ and $\|P-\cN(\mathbf{0}, \idcov)\|_1 \ge \alpha$, with probability at least $2/3$.

\begin{algorithm}[H] 
\caption{Efficient Private Gaussian Mean Testing}\label{algo:gauss}
\begin{algorithmic}[1] 
\Require{Sample $X = (X^{(1)},\dots,X^{(n)}) \in \R^{n\times d}$ drawn from $P^n$.  Parameters $\eps, \delta, \alpha > 0$.}
\Algphase{Stage 1: Pre-processing}
\If{$n<\max\left\{25\ln\frac{d}{\delta}, \frac{5}{\eps}\ln\frac{1}{\delta}\right\}$}
	\Return $\reject$. \label{gauss:reject0}
\EndIf
\State Let $c_i(X)\gets \sum\limits_{j=1}^n \indic\{X_i^{(j)}\le 0\}$ and $m_i(X)\gets  \frac{c_i(X)}{n}-\frac{1}{2}$ for all $i\in[d]$.\label{gauss:defcandm}
\State Let $r_1 \sim \Lap(1/\eps n)$ and $z_1 \gets\max\limits_{i\in[d]} |m_i(X)| + r_1$.\\
\If{$z_1 > \frac{\sqrt{\ln(d/\delta)}}{\sqrt{n}} + \frac{\ln(1/\delta)}{\eps n}$}
	\Return $\reject$.\label{gauss:reject1}
\EndIf
\State Let $B\gets 3\sqrt{\ln\frac{nd}{\delta}}$ and truncate all samples so that $X_i^{(j)}\in [-B,B] ~\forall i\in[d],j\in[n]$.\label{gauss:trunc}%
\State Let $\barx \gets \sum_{j=1}^{n} X^{(j)}$.\\
\State Let $r_2 \sim \Lap(2B/\eps)$ and $z_2 \gets\max\limits_{i\in[d]} |\barx_i| + r_2$.\\
\If{$z_2 > 3\sqrt{2n\ln\frac{nd}{\delta}\cdot\ln\frac{d}{\delta}} + \frac{6}{\eps}\sqrt{\ln\frac{nd}{\delta}}\ln\frac{1}{\delta}$}
	\Return $\reject$. \label{gauss:cond1}
\EndIf
\State Let $\tx \gets \barx + R$, where $R \sim \mathcal{N}(\mathbf{0},\sigma^2 \idcov)$ and $\sigma =  \frac{B\sqrt{8d\ln(5/4\delta)}}{\eps}$. \label{gauss:noisysum} 
\State Let $\Deff^G \gets 144\left(d\ln\frac{d}{\delta}+\frac{d}{n\eps^2}\ln^2\frac{1}{\delta}+\sqrt{nd}\sqrt{\ln\frac{d}{\delta}\cdot\ln\frac{n}{\delta}}+\frac{\sqrt{d}}{\eps}\ln\frac{1}{\delta}\sqrt{\ln\frac{n}{\delta}}\right)\ln\frac{nd}{\delta}$.
\State Let $r_3 \sim \Lap(1/\eps)$ and $z_3 \gets |\{j\in[n]: |\langle X^{(j)}, \tx\rangle| > \Deff^G + \frac{36d}{\eps}\ln\frac{nd}{\delta}\sqrt{\ln\frac{n}{\delta}\cdot\ln\frac{5}{4\delta}}\}|+r_3$. 
\If{$z_3 > \frac{\ln(1/\delta)}{\eps}$}
	\Return $\reject$. \label{gauss:reject2} %
\EndIf
\Algphase{Stage 2: Filtering}
\For{$j=1, \ldots, n$ \label{gauss:forloop}}
	\If{$|\langle X^{(j)}, \tx\rangle| > \Deff^G+\frac{36d}{\eps}\ln\frac{nd}{\delta}\sqrt{\ln\frac{n}{\delta}\cdot\ln\frac{5}{4\delta}}$}
		\State $\hatx^{(j)} \gets N^{(j)}$, where $N^{(j)}\sim \cN(\mathbf{0}, \idcov)$ \label{gauss:resample}
	\Else
		\State $\hatx^{(j)} \gets X^{(j)}$ \label{gauss:replacesample}
	\EndIf
\EndFor
 \Algphase{Stage 3: Noise addition and thresholding}
\State Define the function $T(\hatx) = \sum_{i=1}^d ( \bar{\hatx}_i^2 - n)$. 
\State Let $r_4 \sim \Lap\left(\left(5\Deff^G+\frac{432d}{\eps}\ln\frac{nd}{\delta}\sqrt{\ln\frac{n}{\delta}\cdot\ln\frac{5}{4\delta}}\right)/\eps\right)$ and $z_4 \gets T(\hatx) + r_4$. \label{gauss:finalnoise} 
\If{$z_4 > \frac{n^2\alpha^2}{324}} $
	\Return $\reject$
\EndIf
\State \Return $\accept$.
\end{algorithmic}
\end{algorithm}

Algorithm~\ref{algo:gauss} uses a noisy version of the same statistic that is used in the non-private folklore test and in our uniformity testing algorithms from Section~\ref{sec:put}:
\begin{equation}\label{eq:gauss:statistic}
T(X) = \sum\limits_{i=1}^d \left( \barx_i^2-n\right)
\end{equation}

For any two neighboring datasets $X\sim X'$ differing on the $n$-th sample, the sensitivity of $T$ is bounded by
\begin{equation}\label{eq:gauss:sensitivityT}
|T(X)-T(X')| \leq 2 | \langle X^{(n)}, \barx\rangle | + 2 | \langle X'^{(n)}, \barx'\rangle| + \|X'^{(n)}\|_2^2.
\end{equation}
This follows from inequality~\eqref{eq:wholesensitivity} in the proof of Lemma~\ref{sensitivityT}.
By this bound, the desired condition that datasets must satisfy in order for $T$ to have low sensitivity -- a condition similar to~\eqref{cond:innerproduct} -- is
\begin{equation}\label{cond:gauss:innerproductandnormtwo}
\forall j\in[n] ~ |\langle X^{(j)}, \bar{X}\rangle|\le \Deff^G \text{ and } \|X^{(j)}\|_2^2\le \Deff^G.
\end{equation}

If indeed $P= \cN(\mathbf{0}, \idcov)$, then, with high probability, condition~\eqref{cond:gauss:innerproductandnormtwo} is satisfied. Following the same thinking as with uniformity testing, Algorithm~\ref{algo:gauss} performs the same type of tests and modifications to make sure the dataset satisfies a condition similar to~\eqref{cond:gauss:innerproductandnormtwo} before and if it reaches the final test, which involves the statistic $T$.

The additional challenge of multivariate Gaussian mean testing is that the samples are not bounded. While truncating the samples is necessary, we also need to ensure that any Gaussian dataset will likely remain unchanged through the execution of the algorithm, so that the final test, which involves the statistic $T$, is accurate. To achieve this, we estimate the low sensitivity quantity $|m_i(X)|=\frac{1}{n} |\sum_{j=1}^n \indic\{X_i^{(j)}\le 0\}-\frac{n}{2}|$, which, for Gaussian datasets, acts as a proxy for the mean of the $i$-th coordinate. Due to the good concentration properties of Gaussian random variables, the datasets that pass the test of line~\ref{gauss:reject1} are guaranteed to lie in a small range $[-B,B]$ with high probability, so that the truncation, which follows in line~\ref{gauss:trunc}, will not change any of the samples. The new bound $B$ comes in the new condition~\eqref{cond:gauss:innerproductandnormtwo}, as we will define $\Deff^G =B^2\Deff$.

The main theorem of this section is the following:
\begin{theorem}\label{th:gauss:main}
Algorithm~\ref{algo:gauss} is $(5\eps, 17\delta)$-differentially private and for $P=\cN(\mu, \idcov)$ it distinguishes between the cases $P=\cN(\mathbf{0}, \idcov)$ and $\|P-\cN(\mathbf{0}, \idcov)\|_1\ge \alpha$ with probability at least $2/3$, having sample complexity
\[ n = \tilde{O}\left( \frac{d^{1/2}}{\alpha^2} +  \frac{d^{1/2}}{\alpha \eps}\right).\]
\end{theorem}

First, we prove the guarantees of the non-private test, on which our algorithm is based.
\begin{lemma}[Non-private Test Guarantees]
  \label{lemma:nonprivate:gaussian}
 For $T$ defined as in~\eqref{eq:gauss:statistic}, the following hold:
\begin{itemize}
\item $\E[T(X)] = n^2\normtwo{\mu}^2$.
\item $\var(T(X))= 2n^2d + 4n^3\normtwo{\mu}^2$.
\end{itemize}
\end{lemma}
\begin{proof}
  We have, recalling that $P$ has identity covariance matrix, that it is a product distribution with $i$-th marginal distributed as $\mathcal{N}(\mu_i, 1)$,  
  \begin{align*}
      \E[T(X)] 
      &= \sum_{i=1}^d \left(\sum_{1\leq j_1,j_2\leq n} \E[X_i^{(j_1)}X_i^{(j_2)}]-n\right) = \sum_{i=1}^d \left( \sum_{j=1}^n \E[(X_i^{(j)})^2] + \sum_{j_1\neq j_2} \E[X_i^{(j_1)}]\E[X_i^{(j_2)}] \right) - nd \\
      &= \sum_{i=1}^d \left( \sum_{j=1}^n \E[(X_i^{(j)})^2] + \left( \sum_{j=1}^n \E[X_i^{(j)}] \right)^2 - \sum_{j=1}^n \E[X_i^{(j)}]^2  \right) - nd \\
      &= \sum_{i=1}^d \left( \sum_{j=1}^n \var( X_i^{(j)} ) + \left( \sum_{j=1}^n \E[X_i^{(j)}] \right)^2  \right) - nd = \sum_{i=1}^d \left( n + n^2\mu_i^2 \right) - nd 
      = n^2 \normtwo{\mu}^2,
  \end{align*}
  as claimed. Now, for the variance. Using independence among coordinates, we have
\begin{align*}
      \var(T(X)) &= \sum_{i=1}^d \var\left(\sum_{1\leq j_1,j_2\leq n} X_i^{(j_1)}X_i^{(j_2)}-n\right)
      = \sum_{i=1}^d \var\left(\sum_{j_1,j_2} X_i^{(j_1)}X_i^{(j_2)}\right) \\
      &= \sum_{i=1}^d \left( \E\left[  \left(\sum_{j_1,j_2} X_i^{(j_1)}X_i^{(j_2)} \right)^2 \right] - \E\left[ \sum_{j_1,j_2} X_i^{(j_1)}X_i^{(j_2)} \right]^2 \right).
\end{align*}
We have already, that 
  $
    \sum_{j_1,j_2} \E\left[ X_i^{(j_1)}X_i^{(j_2)} \right] = n^2\mu_i^2 + n
  $. This takes care of the second term; as for the first, we expand
\begin{align*}
  \E\Big[  \Big(&\sum_{j_1,j_2} X_i^{(j_1)}X_i^{(j_2)} \Big)^2 \Big]
     = \sum_{1\leq j_1,j_2,j_3,j_4\leq n} \E\left[ X_i^{(j_1)}X_i^{(j_2)}X_i^{(j_3)}X_i^{(j_4)} \right] \\
     &= n (\mu_i^4+6\mu_i^2 +3) + 24\binom{n}{4}\mu_i^4 + 8\binom{n}{2}\mu_i^2(\mu_i^2+3) + 6\binom{n}{2}(\mu_i^2+1)^2 + 36\binom{n}{3}\mu_i^2 (\mu_i^2+1)\\
     &= 3n^2+6n^3\mu_i^2+n^4\mu_i^4\,.
\end{align*} 
Combining the two, we get
\[
  \var(T(X)) = \sum_{i=1}^d \left( 3n^2+6n^3\mu_i^2+n^4\mu_i^4 - (n^2\mu_i^2 + n)^2  \right)
  = \sum_{i=1}^d \left( 2n^2+4n^3\mu_i^2 \right)
  = 2n^2 d +4n^3\normtwo{\mu}^2,
\]
as stated.
\end{proof}

To prove the privacy guarantee of our main Theorem~\ref{th:gauss:main}, we need to show that with high probability $\hatx$ satisfies a property similar to equation \eqref{cond:gauss:innerproductandnormtwo}, which can not be derived by a simple application of the composition theorem. We first show that the new replacement samples potentially drawn in lines~\ref{gauss:forloop}-\ref{gauss:replacesample} satisfy \eqref{cond:gauss:innerproductandnormtwo}. 

\begin{lemma}\label{lem:gauss:resample}
Let $X$ be a dataset that passes line~\ref{gauss:cond1} of Algorithm~\ref{algo:gauss}. Suppose $N^{(j)}\sim \cN(\mathbf{0}, \idcov)$ for $j\in[n]$.Then with probability $1-5\delta$, $\forall j\in[n]$,
\begin{itemize}
\item $N^{(j)}_i\in[-B,B]$, $\forall i\in[d]$, where $B=3\sqrt{\ln\frac{nd}{\delta}}$ as defined in Algorithm~\ref{algo:gauss}.
\item $\left|\langle N^{(j)}, \barx \rangle \right| \le \Deff^G$.
 \end{itemize}
\end{lemma}
\begin{proof}
We have that for all $j\in[n]$ and $i\in[d]$, $N^{(j)}_i\sim \cN(0,1)$. By the Gaussian tail bound, we get that for all $i\in[d]$ and $j\in[n]$, with probability $1-2\delta$, $|N^{(j)}_i| \le \sqrt{2\ln\frac{nd}{\delta}}\le B$. Therefore, the first point holds, with probability $1-2\delta$.
Conditioned on that and since $X$ has passed line~\ref{gauss:cond1}, by following the same steps as the proof of Lemma~\ref{lem:uniformresample}, we get that with probability $1-3\delta-2\delta=1-5\delta$, for all $x\in X$, $ |\langle x, \barx\rangle| \le \Deff B^2 = \Deff^G$, which is the stated bound.
\end{proof}

\begin{lemma}\label{lem:gauss:hatbarxbounds}
Let $X$ be a dataset that passes line~\ref{gauss:reject2} of Algorithm~\ref{algo:gauss}. For every point $x\in \hatx$ it holds that, with probability $1-10\delta$, 
\begin{equation*} \|x\|_2^2 \leq \Deff^G \text{ and } |\langle x, \bar{\hatx} \rangle | \leq  \Deff^G+\frac{108d}{\eps}\ln\frac{nd}{\delta}\sqrt{\ln\frac{n}{\delta}\cdot\ln\frac{5}{4\delta}} .\end{equation*}
\end{lemma}
\begin{proof}
Since $X$ passed line~\ref{gauss:reject2}, it has also passed the truncation step in line~\ref{gauss:trunc}. By Lemma~\ref{lem:gauss:resample}, all the new data points are also bounded in $[-B,B]$, with probability $1-5\delta$. It follows that for all $x\in \hatx$, $\|x\|_2^2 \le dB^2\le 9d\ln\frac{nd}{\delta}\le \Deff^G$. It remains to prove the second inequality.

Since $X$ passed line~\ref{gauss:reject2}, $z_3\le \frac{\ln(1/\delta)}{\eps}$. By Lemma~\ref{laplace}, with probability $1-\delta$, $|r_3|\le \frac{\ln(1/\delta)}{\eps}$, so
\[|\{j\in[n]: |\langle X^{(j)}, \tx\rangle| > \Deff^G +\frac{36d}{\eps}\ln\frac{nd}{\delta}\sqrt{\ln\frac{n}{\delta}\cdot\ln\frac{5}{4\delta}}\}|\le \frac{2}{\eps}\ln\frac{1}{\delta}.\]
Therefore, $X$ and $\hatx$ differ in at most $\frac{2}{\eps}\ln\frac{1}{\delta}$ data points. Thus, with probability $1-6\delta$, $\forall x\in \hatx$,
\begin{equation}\label{eq:gauss:modifiedcorr1}
|\langle x, \bar{\hatx} \rangle |\le |\langle x, \barx \rangle |+ |\langle x, \barx-\bar{\hatx} \rangle|
\le |\langle x, \barx \rangle| + \frac{4dB^2}{\eps}\ln\frac{1}{\delta}
= |\langle x, \barx \rangle| + \frac{36d}{\eps}\ln\frac{nd}{\delta}\cdot\ln\frac{1}{\delta}
.\end{equation}
If $x$ was resampled in line~\ref{gauss:resample} then by Lemma~\ref{lem:gauss:resample}, we are done. 
Otherwise, if $x$ was not resampled then by assumption $|\langle x, \tx \rangle| \le\Deff^G+\frac{36d}{\eps}\ln\frac{nd}{\delta}\sqrt{\ln\frac{n}{\delta}\cdot\ln\frac{5}{4\delta}}$. 
It holds that
\begin{equation}\label{eq:gauss:noisycorr1} |\langle x, \barx \rangle| \le |\langle x, \tx \rangle |+ |\langle x, \barx-\tx \rangle| \le \Deff^G +\frac{36d}{\eps}\ln\frac{nd}{\delta}\sqrt{\ln\frac{n}{\delta}\cdot\ln\frac{5}{4\delta}} + |\langle x, \barx-\tx \rangle| 
.\end{equation}
Now, $\barx-\tx=R$ where $R \sim \mathcal{N}(\mathbf{0},\sigma^2 \idcov)$ and $\sigma = B\sqrt{8d\ln(5/4\delta)}/\eps$, as in line~\ref{gauss:noisysum} of Algorithm~\ref{algo:gauss}.
By symmetry, $|\langle x,R \rangle| \leq B\cdot \left|\sum_{i=1}^{d} Y_i\right|$, where each $Y_i\sim\mathcal{N}(0,\sigma^2)$. 
It follows that with probability $1-2\delta$, $\left|\sum_{i=1}^{d} Y_i\right| \le \sigma\sqrt{2d\ln\frac{n}{\delta}} \le \frac{4dB}{\eps}\sqrt{\ln\frac{n}{\delta}\cdot\ln\frac{5}{4\delta}}$. 
So $\forall x\in X$,
 \begin{equation}\label{eq:gauss:noisycorr} |\langle x, \barx-\tx \rangle| \le \frac{4dB^2}{\eps}\sqrt{\ln\frac{n}{\delta}\cdot\ln\frac{5}{4\delta}} \le \frac{36d}{\eps}\ln\frac{nd}{\delta}\sqrt{\ln\frac{n}{\delta}\cdot\ln\frac{5}{4\delta}}.\end{equation}
 Therefore, by union bound and inequalities~\eqref{eq:gauss:modifiedcorr1} and~\eqref{eq:gauss:noisycorr1}, with probability $1-8\delta$, for all $x\in\hatx$, 
 \[\|x\|_2^2 \le \Deff^G \text{ and } |\langle x, \bar{\hatx} \rangle| \le \Deff^G+\frac{108d}{\eps}\ln\frac{nd}{\delta}\sqrt{\ln\frac{n}{\delta}\cdot\ln\frac{5}{4\delta}}.\] 
 This concludes the proof of the lemma.
 \end{proof}

\begin{lemma}\label{lem:gauss:priv} Algorithm~\ref{algo:gauss} is $(5\eps, 17\delta)$-differentially private. \end{lemma}
\begin{proof}
By Lemma~\ref{lem:gauss:hatbarxbounds} and inequality~\eqref{eq:gauss:sensitivityT}, with probability $1-16\delta$, for any two neighboring datasets $X\sim X'$ that reach line~\ref{gauss:finalnoise}, $\hat{X}$ and $\hat{X}'$ satisfy condition~\eqref{cond:gauss:innerproductandnormtwo}. Therefore, for $\hatx\sim\hatx'$,
\[|T(\hatx)-T(\hatx')| \le 5\Deff^G+\frac{432d}{\eps}\ln\frac{nd}{\delta}\sqrt{\ln\frac{n}{\delta}\cdot\ln\frac{5}{4\delta}}.\]
This ensures that with probability $1-16\delta$, the noise added in the last test in line~\ref{gauss:finalnoise} is sufficient to ensure privacy. As in Section~\ref{sec:put}, the first observation towards proving the privacy guarantee is that the applications of the Laplace mechanism in lines~\ref{gauss:reject1},~\ref{gauss:cond1}, and~\ref{gauss:reject2} are $\eps$-DP and the application of the Gaussian mechanism in line~\ref{gauss:noisysum} is $(\eps, \delta)$-DP. Now, by coupling the random variables in separate runs of Algorithm~\ref{algo:gauss} on $X$ and $X'$, to ensure that $\hatx$ and $\hatx'$ are neighboring databases exactly as in the proof of Lemma~\ref{eff_priv}, we get that our algorithm is $(5\eps, 17\delta)$-DP.
\end{proof}

We now turn to the utility guarantee of our main theorem. As in the previous section, the crux of the proof is as follows:
\begin{enumerate}
\item If $P =\cN(\mathbf{0},\idcov)$, then with high probability $X$ passes the first two checks at line~\ref{gauss:reject1} and~\ref{gauss:cond1}.
\item If $P =\cN(\mu,\idcov)$, then:
\begin{enumerate}
\item If $X$ passes the first two checks at line~\ref{gauss:reject1} and~\ref{gauss:cond1}, then $X=\hatx$ with high probability, so $T(X)=T(\hatx)$.
\item The amount of noise added to $T(\hatx)$ is small enough that one can still distinguish between the two hypotheses.
\end{enumerate}
\end{enumerate}

We will now prove some important properties of the first test of line~\ref{gauss:reject1} of Algorithm~\ref{algo:gauss}, which guarantee that for multivariate Gaussian distributions, the estimate $|m_i(X)|$ is close to the absolute mean $|\mu_i|$ of each coordinate.
\begin{lemma}\label{lem:gauss:m_i}
Suppose $P=\cN(\mu, \idcov)$ and let $\{m_i(X)\}_{i=1}^d$ as defined in line~\ref{gauss:defcandm} of Algorithm~\ref{algo:gauss}. \begin{itemize}
\item If $\mu_i=0$, then $\E[m_i(X)]=0$.
\item $|\E[m_i(X)]| \geq 0.84\cdot\min\left\{\frac{|\mu_i|}{\sqrt{2}}, 1\right\}$.
\item With probability $1-2\delta$, $|m_i(X) - \E[m_i(X)]| \leq \frac{\sqrt{\ln(d/\delta)}}{\sqrt{n}}$ for all $i\in[d]$.
\end{itemize}\end{lemma}
\begin{proof}
Let $Y_i^{(j)}=\indic\{X_i^{(j)} \le 0\}$ for $i\in[d], j\in[n]$. Let us calculate $\E[m_i(X)]$.
\[\E\left[\frac{c_i(X)}{n}-\frac{1}{2}\right] =\frac{1}{n}\sum_{j=1}^n \E[Y_i^{(j)}] - \frac{1}{2}=\frac{1}{n}\sum_{j=1}^n \Pr[X_i^{(j)} \le 0]-\frac{1}{2}=\frac{1}{n}\sum_{j=1}^n \Phi(-\mu_i) -\frac{1}{2}= \erf\left(-\frac{\mu_i}{\sqrt{2}}\right),\]
where $\Phi(x)$ is the CDF of $\cN(0, 1)$ and $\erf(x)=\frac{1}{\sqrt{\pi}} \int_{-x}^{x} \exp(-t^2) \,dx $ the error function.
The first point follows, since $\erf(0)=0$. The second point follows by Lemma~\ref{lem:gauss:erfbound}.
It remains to prove the third point. By Hoeffding's inequality, for all $i\in[d]$, 
\[\Pr\left[ \left|\sum_{j=1}^n Y_i^{(j)} - \E\left[\sum_{j=1}^n Y_i^{(j)}\right] \right| > t \right]\le 2\exp(-2t^2/n).\]
Setting $t=\sqrt{n\ln\frac{d}{\delta}}$ and by union bound, we get that with probability $1-2\delta$, for all $i\in[d]$, $|m_i(X)-\E[m_i(x)]| =\frac{1}{n} |c_i(X)-\E[c_i(X)]|\le \frac{\sqrt{\ln(d/\delta)}}{\sqrt{n}}$ .
\end{proof}

We can now prove that if $P=\cN(\mathbf{0},\idcov)$, then $X$ will pass the first two tests with high probability.
\begin{lemma}\label{lem:gauss:zeromeanpassescond1}
With probability $1-6\delta$, if $P=\cN(\mathbf{0},\idcov)$ then Algorithm~\ref{algo:gauss} does not reject neither in line~\ref{gauss:reject1}, nor in line~\ref{gauss:cond1}.
\end{lemma}
\begin{proof}
Suppose $P = \cN(\mathbf{0},\idcov)$, that is, $\mu_i=0$ for all $i\in[d]$. 
By Lemma~\ref{lem:gauss:m_i}, $\E[m_i(X)]=0$ and $|m_i(X)|\le \frac{\sqrt{\ln(d/\delta)}}{\sqrt{n}}$ for all $i\in[d]$, with probability at least $1-2\delta$.
It follows that, with probability $1-2\delta$, 
\begin{equation}\label{eq:zeromeanmaxm_i} \max\limits_{i\in[d]} |m_i(X)| \le \frac{\sqrt{\ln(d/\delta)}}{\sqrt{n}}. \end{equation}
Let $r_1\sim \Lap(1/\eps n)$ as in Algorithm~\ref{algo:gauss}. By Lemma~\ref{laplace}, with probability $1-\delta$, $|r_1|\le \frac{\ln(1/\delta)}{\eps n}$.
Combined with~\eqref{eq:zeromeanmaxm_i}, we get that, with probability $1-3\delta$,
\[z_1=\max\limits_{i\in[d]}|m_i(X)| + r_1 \le \frac{\sqrt{\ln(d/\delta)}}{\sqrt{n}} +  \frac{\ln(1/\delta)}{\eps n},\]
showing that the dataset $X$ will pass line~\ref{gauss:reject1}.

For the test of line~\ref{gauss:cond1}, we first observe that, since $P = \cN(\mathbf{0},\idcov)$, $\E[\barx_i]=0$ for all $i\in [d]$ and that, due to the truncation in line~\ref{gauss:trunc}, for all the samples it holds that $|X_i^{(j)}|\le B$, where $B=3\sqrt{\ln\frac{nd}{\delta}}$. 
Following the same steps as Lemma~\ref{uniformpassescond1}, we get that with probability $1-3\delta$, 
\[z_2 \le 3\sqrt{2n\ln\frac{nd}{\delta}\cdot\ln\frac{d}{\delta}} + \frac{6}{\eps}\sqrt{\ln\frac{nd}{\delta}}\ln\frac{1}{\delta},\]
showing that the dataset $X$ will pass line~\ref{gauss:cond1}.

Therefore, $X$ will pass both the test in line~\ref{gauss:reject1} and the test in line~\ref{gauss:cond1}, with probability $1-6\delta$.
\end{proof}

Using Lemma~\ref{lem:gauss:m_i}, we can also prove that all Gaussian datasets that pass the first two tests remain unchanged for the rest of the algorithm. The following lemma states that the dataset will not be modified during the truncation phase in line~\ref{gauss:trunc}.
\begin{lemma}\label{lem:gauss:truncation}
Suppose $P=\cN(\mu, \idcov)$. If dataset $X$ passes line~\ref{gauss:reject1} of Algorithm~\ref{algo:gauss}, then, with probability $1-5\delta$, no truncation occurs in line~\ref{gauss:trunc}.
\end{lemma}
\begin{proof}
Since $X$ passed line~\ref{gauss:reject1}, it holds that $z_1=\max\limits_{i\in[d]}|m_i(X)| + r_1 \le \frac{\sqrt{\ln(d/\delta)}}{\sqrt{n}} +  \frac{\ln(1/\delta)}{\eps n}$. By Lemma~\ref{laplace}, with probability $1-\delta$, $|r_1| \le \frac{\ln(1/\delta)}{\eps n}$. So with probability $1-\delta$, for all $i\in[d]$, 
\begin{equation}\label{eq:gauss:boundempm_i}
|m_i(X)| \le \frac{\sqrt{\ln(d/\delta)}}{\sqrt{n}} +  \frac{2\ln(1/\delta)}{\eps n}.
\end{equation}
By Lemma~\ref{lem:gauss:m_i}, with probability $1-2\delta$, $|m_i(X)-\E[m_i(X)]|\le \frac{\sqrt{\ln(d/\delta)}}{\sqrt{n}}$ for all $i\in[d]$.
By union bound and inequality~\eqref{eq:gauss:boundempm_i}, with probability $1-3\delta$, for all $i\in[d]$, $|\E[m_i(X)]| \le \frac{2\sqrt{\ln(d/\delta)}}{\sqrt{n}} + \frac{2\ln(1/\delta)}{\eps n}$. 
Again by Lemma~\ref{lem:gauss:m_i}, the last inequality implies that with probability $1-3\delta$, for all $i\in[d]$,
\begin{equation}\label{eq:gauss:boundonexpmean}
0.84\cdot \min\left\{\frac{|\mu_i|}{\sqrt{2}},1\right\} \le \frac{2\sqrt{\ln(d/\delta)}}{\sqrt{n}} + \frac{2\ln(1/\delta)}{\eps n}.
\end{equation}
Since the algorithm has passed line~\ref{gauss:reject0}, $n\ge\max\left\{25\ln\frac{d}{\delta}, \frac{5}{\eps}\ln\frac{1}{\delta}\right\}$, and we get $\frac{2\sqrt{\ln(d/\delta)}}{\sqrt{n}} + \frac{2\ln(1/\delta)}{\eps n}  \le 0.8 < 0.84$. So it must be that $|\mu_i|<\sqrt{2}$, since inequality~\eqref{eq:gauss:boundonexpmean} can not be satisfied otherwise. Then it holds that
\begin{equation}\label{eq:gauss:mu_ibound}
\frac{0.84|\mu_i|}{\sqrt{2}} \le \frac{2\sqrt{\ln(d/\delta)}}{\sqrt{n}} + \frac{2\ln(1/\delta)}{\eps n} \Rightarrow |\mu_i| \leq \frac{3.4\sqrt{\ln(d/\delta)}}{\sqrt{n}} + \frac{3.4\ln(1/\delta)}{\eps n} .\end{equation}
Since all $X_i^{(j)}\sim \cN(\mu_i,1)$, we have the tail bound
\[\Pr\left[\left |X_i^{(j)}-\mu_i \right|> t\right] \le 2\exp(-t^2/2).\]
Setting $t=\sqrt{2\ln\frac{nd}{\delta}}$, we get that with probability $1-2\delta$, for all $i\in[d]$ and $j\in[n]$, $|X_i^{(j)}| \le |\mu_i|+t$. Replacing the chosen $t$ and by inequality~\eqref{eq:gauss:mu_ibound},
\[|X_i^{(j)}| \le \frac{3.4\sqrt{\ln(d/\delta)}}{\sqrt{n}} + \frac{3.4\ln(1/\delta)}{\eps n} + \sqrt{2\ln\frac{nd}{\delta}} < 3\sqrt{\ln\frac{nd}{\delta}},\]
where the last inequality follows from our condition on $n$.

We have proven that, if $X$ passes line~\ref{gauss:reject1}, then with probability $1-5\delta$, all samples already fall within the desired range $[-B,B]$, so the truncation does not affect the dataset.
\end{proof}

The following lemma states that all datasets that survive line~\ref{gauss:cond1} satisfy the desired conditions~\eqref{cond:gauss:innerproductandnormtwo}, which is sufficient to establish that with high probability, the dataset will not be modified during the resampling phase in lines~\ref{gauss:forloop}-\ref{gauss:replacesample}.
\begin{lemma}\label{lem:gauss:boundoninnerproduct}
Let $X$ be a dataset that passes line~\ref{gauss:cond1} of Algorithm~\ref{algo:gauss}. With probability $1-3\delta$, for all $ x\in X$,
\begin{equation*} |\langle x, \barx\rangle| \le 144\left(d\ln\frac{d}{\delta}+\frac{d}{n\eps^2}\ln^2\frac{1}{\delta}+\sqrt{nd}\sqrt{\ln\frac{d}{\delta}\cdot\ln\frac{n}{\delta}}+\frac{\sqrt{d}}{\eps}\ln\frac{1}{\delta}\sqrt{\ln\frac{n}{\delta}}\right)\ln\frac{nd}{\delta}. \end{equation*}
\end{lemma}
\begin{proof}
Observe that, due to the truncation in line~\ref{gauss:trunc}, $\forall k\in [n]$ and $\forall i\in[d]$, $|X_i^{(k)}| \le B$. Following the same steps as the proof of Lemma~\ref{lem:nomodifications}, we get that with probability $1-3\delta$, for all $x\in X$, $ |\langle x, \barx\rangle| \le \Deff B^2$, which is the stated bound.
\end{proof}

We will now prove the main theorem of this section. 

\begin{proof}[Proof of Theorem~\ref{th:gauss:main}]
The privacy guarantee was established in Lemma~\ref{lem:gauss:priv}. Now, for the utility guarantee. 

\medskip\noindent\emph{Completeness:} Suppose $P =\cN(\mathbf{0},\idcov)$. By Lemma~\ref{lem:gauss:zeromeanpassescond1}, with probability $1-6\delta$, $X$ passes line~\ref{gauss:reject1} and line~\ref{gauss:cond1}. Also, by Lemma~\ref{lem:gauss:truncation}, with probability $1-5\delta$, it does not get truncated in line~\ref{gauss:trunc} in the meantime. So with probability $1-11\delta$, $X$ has passed line~\ref{gauss:cond1} and reached line~\ref{gauss:reject2}, unchanged. 

By Lemma~\ref{lem:gauss:boundoninnerproduct} and union bound, with probability $1-14\delta$, for the chosen $\Deff^G$, $ |\langle x, \barx\rangle| \le \Deff^G$ holds for all $x\in X$.

As we have showed, by inequality~\eqref{eq:gauss:noisycorr}, with probability $1-2\delta$, for all $x\in X$, \[|\langle x, \barx-\tx\rangle| \le \frac{36d}{\eps}\ln\frac{nd}{\delta}\sqrt{\ln\frac{n}{\delta}\cdot\ln\frac{5}{4\delta}}.\]
It follows that with probability $1-16\delta$, $\forall x\in X$,  $|\langle x, \tx \rangle | \le \Deff^G + \frac{36d}{\eps}\ln\frac{nd}{\delta}\sqrt{\ln\frac{n}{\delta}\cdot\ln\frac{5}{4\delta}}$. \\
By Lemma~\ref{laplace}, with probability $1-\delta$, $|r_3|\le \frac{\ln(1/\delta)}{\eps}$, so with probability $1-17\delta$, $z_3 \leq \frac{\ln(1/\delta)}{\eps}$. Thus, the dataset survives line~\ref{gauss:reject2} as well and no points are modified in lines~\ref{gauss:forloop}-\ref{gauss:replacesample}. Since, with probability $1-17\delta$, the dataset reaches the last test in line~\ref{gauss:finalnoise} unchanged, it only remains to prove that the last test is accurate with high probability.

By Lemma~\ref{laplace}, with probability $0.95$, $|r_4|\le \frac{\ln(20)}{\eps}\left(5\Deff^G+\frac{432d}{\eps}\ln\frac{nd}{\delta}\sqrt{\ln\frac{n}{\delta}\cdot\ln\frac{5}{4\delta}}\right)$. We want the following to hold:
\begin{equation}\label{eq:sccondition}
\frac{\ln(20)}{\eps}\left(5\Deff^G+\frac{432d}{\eps}\ln\frac{nd}{\delta}\sqrt{\ln\frac{n}{\delta}\cdot\ln\frac{5}{4\delta}}\right) \le \frac{n^2\alpha^2}{8\cdot 81}
\end{equation}
If this is true, we have $|r_4| \le \frac{n^2\alpha^2}{8\cdot 81}$. Then, with probability $0.95-17\delta$,
\begin{align*}
\Pr\left[z_4 > \frac{n^2\alpha^2}{4\cdot 81}\right] & =\Pr\left[T(X) >  \frac{n^2\alpha^2}{4\cdot 81} - r_3\right] \tag{$T(\hatx)=T(X)$}\\
&\leq \Pr\left[T(X) >  \frac{n^2\alpha^2}{8\cdot 81}\right] \tag{$|r_4|\leq \frac{n^2\alpha^2}{8\cdot 81}$} \\
& = \Pr\left[T(X) > \frac{n^2\alpha^2}{8\cdot 81} + \E[T(X)]\right] \tag{Lemma~\ref{lemma:nonprivate:gaussian}}\\
&  \leq \frac{8^2\cdot 81^2 \var(T(X))}{n^4\alpha^4} \tag{Chebyshev's inequality}\\
& \leq \frac{8^2\cdot 81^2\cdot 2d}{n^2\alpha^4} \tag{Lemma~\ref{lemma:nonprivate:gaussian}}
\end{align*}
For $n = \Omega\left(\frac{d^{1/2}}{\alpha^2}\right)$
, it holds that $\Pr\left[z_4 > \frac{n^2\alpha^2}{4\cdot 81}\right] \leq 0.05$.
So with probability $0.9-17\delta$, $z_4\leq \frac{n^2\alpha^2}{324}$. Condition~\eqref{eq:sccondition} and $n=\Omega\left(\frac{d^{1/2}}{\alpha^2}\right)$ is satisfied for
\begin{align*}
n=\Omega\Bigg(&\frac{d^{1/2}}{\alpha^2} 
+ \frac{d^{1/2}}{\alpha\eps^{1/2}}\left(\ln\frac{1}{\delta}\right)^{1/2}\left(\ln\frac{d}{\alpha\eps\delta}\right)^{1/2} \\
&+ \frac{d^{1/3}}{\alpha^{2/3}\eps}\left(\ln\frac{1}{\delta}\right)^{2/3}\left(\ln\frac{d}{\alpha\eps\delta}\right)^{1/3}
+\frac{d^{1/3}}{\alpha^{4/3}\eps^{2/3}}\left(\ln\frac{d}{\delta}\right)^{1/3} \left(\ln\frac{d}{\alpha \eps \delta}\right)\\
&+ \frac{d^{1/4}}{\alpha \eps}\left(\ln\frac{1}{\delta}\right)^{1/2}\left(\ln\frac{d}{\alpha\eps\delta}\right)^{3/4}
+ \frac{d^{1/2}}{\alpha \eps}\left(\ln\frac{1}{\delta}\right)^{1/4}\left(\ln\frac{d}{\alpha\eps\delta}\right)^{3/4} \Bigg),
\end{align*}

Notice that for this value of $n$, Algorithm~\ref{algo:gauss} does not reject in line~\ref{gauss:reject0} either.
Thus, for this value of $n$ and $\delta \leq 0.01$, Algorithm~\ref{algo:gauss} returns $\accept$, with probability $2/3$. 

\medskip\noindent\emph{Soundness:} Suppose $P =\cN(\mu,\idcov)$ and $\|P-\cN(\mathbf{0},\idcov)\|_1\ge \alpha$. Let us assume that the algorithm does not return \texttt{REJECT} in line~\ref{gauss:reject0},~\ref{gauss:reject1}, or~\ref{gauss:cond1}, which is the desired output in this case. Since $X$ has passed line~\ref{gauss:cond1}, by Lemma~\ref{lem:gauss:truncation} and Lemma~\ref{lem:gauss:boundoninnerproduct}, with probability $1-8\delta$, $X$ does not get truncated in line~\ref{gauss:trunc} and for all $x\in X$, $ |\langle x, \barx\rangle| \le \Deff^G$.
Similar to the completeness proof, with probability $0.95-11\delta$, the dataset reaches the last test in line~\ref{gauss:finalnoise} unchanged and for the chosen $n$, $|r_4|\leq \frac{n^2\alpha^2}{8\cdot 81}$.
Before proving the accuracy of the last test in line~\ref{gauss:finalnoise}, note that since $\|P-\cN(\mathbf{0},\idcov)\|_1\ge \alpha$, by Lemma~\ref{fact:dist:gaussians}, we have that $\|\mu\|_2^2 \ge \alpha^2/81$. So, by Lemma~\ref{lemma:nonprivate:gaussian}, 
\begin{equation}\label{eq:nonprivate:gaussian:exp}
\E[T(X)]\ge \frac{n^2\alpha^2}{81}.
\end{equation} 
With probability $0.95-13\delta$,
\begin{align*}
\Pr\left[z_4\leq \frac{n^2\alpha^2}{4\cdot 81}\right] & = \Pr\left[T(X)\leq \frac{n^2\alpha^2}{4\cdot 81} - r_4\right] \tag{$T(\hatx)=T(X)$}\\
&\leq \Pr\left[T(X) \leq  \frac{3n^2\alpha^2}{8\cdot 81}\right] \tag{$|r_4|\leq \frac{n^2\alpha^2}{8\cdot 81}$}\\
& \leq \Pr\left[T(X) \leq  \frac{3\E[T(X)]}{8}\right] \tag{by~\eqref{eq:nonprivate:gaussian:exp}}\\
& = \Pr\left[\E[T(X)] - T(X) \geq \frac{5\E[T(X)]}{8}\right]\\
& \leq \frac{8^2\var(T(X))}{5^2(\E[T(X)])^2} \tag{Chevyshev's inequality}\\
& \leq \frac{8^2\cdot 81^2 \cdot 2d}{5^2 n^2\alpha^4} + \frac{8^2\cdot 4\cdot 81}{5^2 n \alpha^2} \tag{Lemma~\ref{lemma:nonprivate:gaussian} and~\eqref{eq:nonprivate:gaussian:exp}}\\
\end{align*}
For $n=\Omega\left(\frac{d^{1/2}}{\alpha^2}\right)$, it holds that $\Pr\left[z_4\leq \frac{n^2\alpha^2}{4\cdot 81}\right] \leq 0.05$.
So with probability $0.9-11\delta$, $z_4 > \frac{n^2\alpha^2}{324}$. Thus, for $\delta \leq 0.01$ and for the stated value of $n$, we conclude that with probability $2/3$, Algorithm~\ref{algo:gauss} returns $\reject$. 

\medskip\noindent Ignoring the logarithmic factors, the sample complexity is simplified using Claim~\ref{amgmsimplification}.
\end{proof}

\fi

\end{document}